\documentclass[
                 color-hyperref,               
                 graphics,                     
                 force-numerical-bibliography, 
                 redefine-eqnarray,            
                 eucal,                        
                 doublestroke,                 
                 upright,                      
               ]{aihpd2021a}

\usepackage{float}
\usepackage{physics}
\usepackage{dsfont}
 \usepackage{booktabs}
\title{Quantum error-correcting codes and their geometries}

\authors[S. Ball, A. Centelles and F. Huber]{
Simeon Ball\thanks{The first author acknowledges the support of the Spanish Ministry of Science and Innovation grants MTM2017-82166-P and PID2020-113082GB-I00 funded by MCIN/AEI/10.13039/501100011033.}, Aina Centelles
and Felix Huber\thanks{The third author acknowledges the support of the
Spanish MINECO (Severo Ochoa SEV-2015-0522),  
Fundaci\'o Cellex and Mir-Puig, 
Generalitat de Catalunya (SGR 1381 and CERCA Programme), 
and the European Union under Horizon2020 (PROBIST 754510).}}

\address[simeon@ma4.upc.edu]  {Simeon Ball,
   Departament de Matem\`atiques, 
Universitat Polit\`ecnica de Catalunya, 
M\`odul C3, Campus Nord,
Carrer Jordi Girona 1-3,
08034 Barcelona, Spain}

\address[aina.centelles@estudiant.upc.edu]  {Aina Centelles,
Facultat de Matem\`atiques, 
Universitat Polit\`ecnica de Catalunya, 
Carrer de Pau Gargallo, 14,
08028 Barcelona, Spain}

\address[felix.huber@icfo.eu]{Felix Huber,
ICFO – The Institute of Photonic Sciences,
Mediterranean Technology Park, 
Avinguda Carl Friedrich Gauss, 3,
08860 Castelldefels (Barcelona), Spain}

\editor{Gil Kalai}     
  \noaddress       

\theoremstyle{plain}
\newtheorem{theorem}{Theorem}[section]

\newtheorem{lemma}[theorem]{Lemma}

\newtheorem{rprob}{Research Problem}
\theoremstyle{definition}

 \newtheorem{prop}[theorem]{Proposition}
\theoremstyle{remark}

\newtheorem{example}[theorem]{Example}

\newcommand{\s}{\sigma}
\newcommand{\ot}{\otimes}
\newcommand{\one}[0]{\mathds{1}}

\newcommand{\HH}{\mathcal{H}}
\newcommand{\QQ}{\mathcal{Q}}
\newcommand{\XX}{\mathcal{X}}
\newcommand{\YY}{\mathcal{Y}}
\newcommand{\overbar}[1]{\mkern 1.5mu\overline{\mkern-1.5mu#1\mkern-1.5mu}\mkern 1.5mu}

\DeclareMathOperator{\wt}{wt}
\DeclareMathOperator{\supp}{supp}

\begin{document}

\begin{abstract}
This is an expository article aiming to introduce the reader to the underlying mathematics and geometry of quantum error correction. Information stored on quantum particles is subject to noise and interference from the environment. Quantum error-correcting codes allow the negation of these effects in order to successfully restore the original quantum information. We briefly describe the necessary quantum mechanical background to be able to understand how quantum error-correction works. We go on to construct quantum codes: firstly qubit stabilizer codes, then qubit non-stabilizer codes, and finally codes with a higher local dimension. We will delve into the geometry of these codes. This allows one to deduce the parameters of the code efficiently, deduce the inequivalence between codes that have the same parameters, and presents a useful tool in deducing the feasibility of certain parameters. We also include sections on quantum maximum distance separable codes and the quantum MacWilliams identities.
\end{abstract}



\tableofcontents 

\

\

We have used various sources in the preparation of this article, principally Gottesman \cite{Gottesman2009, GottesmanThesis}, Glynn et al \cite{GGMG} and Ketkar et al \cite{KKKS2006}. The most original parts of these notes are Section~\ref{unionstabilizers} and Section~\ref{sectionMDS}. Section~\ref{sectionquDit} is based on Ketkar et al \cite{KKKS2006} but massaged so that appears as a straightforward generalisation of the qubit case of Section~\ref{sectionqubit}. Although the main results of Section~\ref{sectiongeometry} are from Glynn et al \cite{GGMG}, in a deviation from their approach we have chosen to prove these results without using the ${\mathbb F}_4$ trick, which we do not consider until later in Section~\ref{quDitconstructions}. The interested reader is referred to the books by Sakurai~\cite{Sakurai1994} and Nielsen \& Chuang \cite{NielsenChuang2000} for standard treatments of quantum mechanics and quantum information theory, to the book by Haroche \& Raimond~\cite{HarocheRaimond2006} for a thorough treatment of current experiments in quantum mechanics, and to the book by Aaronson~\cite{Aaronson2013} for further connections to mathematics, computer science, physics, and philosophy. For those uninitiated in quantum mechanics or quantum computing, we strongly recommend the delightful mnemotic essay on quantum computing by Matuschak and Nielsen at \url{https://quantum.country/qcvc}.

\section{Quantum codes}

\subsection{Introduction} \label{sect:intro}

A {\em qubit} is a two-state or two-level quantum-mechanical system. For example, 
the intrinsic angular momentum ({\em spin}) of an electron is such a system. 
It can only take two values when measured in arbitrary spatial direction, 
say by measuring the electrons deflection when passing by an inhomogeneous magnetic field. 
The two corresponding spin-states are commonly referred to as as 
``spin up'' and ``spin down'' states with respect to that direction.
Another example is the polarization of light. 
Here the two states can be taken to be vertically and horizontally polarized light; 
another choice is light that is left circularly and right circularly polarized.
In general, a continuum of different photon polarizations are possible.
Yet only two distinct states are observed when e.g. putting beamsplitters or polarization filters 
in the path of a light beam.

This raises the question: why are only ever two discrete values 
corresponding to two discrete states observed, 
if electrons and photons can take on a continuum of possible 
spin-directions or polarizations?
The answer lies with what measurements on quantum systems reveal. 
It turns out that for a two-state quantum-mechanical system, 
any individual measurements can only ever reveal the answer to a binary question.
In other words, the measurement indicates in which of two mutually exclusive states 
the qubit can be found after the measurement. 
Thus while qubits can take on a continuity of states and a continuity of measurements can be performed, 
only two-valued results can ever be obtained. Thus the notion of a qubit as a {\em quantum bit}.
We will not dwell on the strangeness of quantum mechanics further, the interested reader is referred 
to discussions of the Stern-Gerlach and double-slit experiments
such as found in the books by Sakurai~\cite{Sakurai1994} and Haroche \& Raimond~\cite{HarocheRaimond2006}~\footnote{For a visualisation of these experiments, see 
\url{http://toutestquantique.fr/en/spin/} and \url{http://toutestquantique.fr/en/duality/}}.

In mathematical terms a qubit is represented by a unit vector in 
${\mathbb C}^2$. The spin up and spin down (or any other choice of a pair of physically completely distinguishable states) 
are represented by an orthonormal basis $\ket{0}$ and $\ket{1}$.
The notation $\ket{0}$ is a shorthand for the vector 
$\begin{bmatrix}
1 \\ 
0                                                       
\end{bmatrix}
$
and $\ket{1}$ 
stands for 
$\begin{bmatrix}
0 \\ 1                                                       
\end{bmatrix}
$.
The two {\em kets} $\ket{0}$ and $\ket{1}$ are also known as the {\em computational basis} vectors.

Consider now the state 
\begin{equation}
 \ket{\psi} = \frac{1}{\sqrt{2}}(\ket{0} + \ket{1})=\frac{1}{\sqrt{2}}\begin{bmatrix}
1 \\ 
1                                                       
\end{bmatrix}.
\end{equation}
While $\ket{\psi} \in {\mathbb C}^2$ represents a physically unique state, 
it is, upon measurement in the
spin-up -- spin-down direction, found in either of these two directions
with equal probability. 
Sometimes this situation is referred to as the system being ``in two states simultaneously''.
A more accurate description is that the system is ``in {\em superposition} of spin-up and spin-down'',
or in other words, the system is correctly described as a {\em linear combination} of spin-up and spin-down.

A typical qubit reads
$$
\ket{\alpha}=\alpha_0\!\ket{0}+\alpha_1\!\ket{1}.
$$

As usual, $\overline{z}$ is the complex conjugate of the complex number $z$. When measured, the qubit is with probability $\overline{\alpha_0}\alpha_0$ found in state~$\ket{0}$ (``spin-up'') and with probability $\overline{\alpha_1}\alpha_1$ found in state~$\ket{1}$ (``spin down''). Since the sum of these two probabilities must be one, we have that for a qubit
\begin{equation}\label{eq:normalization}
\overline{\alpha_0}\alpha_0+\overline{\alpha_1}\alpha_1=1.
\end{equation}

The ``ket'' notation $\ket{\alpha}$ is used for a column vector, whilst the ``bra'' notation $\bra{\alpha}$ is used for a row vector whose coordinates are the complex conjugates of the coordinates of $\ket{\alpha}$. Thus, the ``bra'' $\bra{\alpha}$ is a linear form. The {\em inner product} or ``bra-ket'' on ${\mathbb C}^2$ is defined as
$$
\bra{\alpha}\ket{\beta}=\overline{\alpha_0}\beta_0+\overline{\alpha_1}\beta_1.
$$
The normalisation condition in Eq.~\eqref{eq:normalization} then reads as $\braket{\alpha}{\alpha} = 1$,
and qubits are represented by complex vectors in $\mathbb C^2$ of unit length.

A {\em unitary transformation} of ${\mathbb C}^2$ is given by a non-singular $2 \times 2$ matrix $U$ which preserves this inner product, so
$$
\bra{U\alpha}\ket{U\beta}=\bra{\alpha}\ket{\beta},
$$
for all $\bra{\alpha}$ and $\ket{\beta}$. The set of such unitaries forms the special unitary group SU(2).

In particular, 
$$
\bra{U\alpha}\ket{U\alpha}=\bra{\alpha}\ket{\alpha}=1.
$$

The matrix
$$
U=\left(\begin{array}{cc} 0 & -i \\ i & 0 \end{array} \right)
$$
is an example of a unitary transformation since
$$
\bra{U\alpha}\ket{U\beta}=(\overline{-i\alpha_1} \bra{0}+\overline{i\alpha_0} \bra{1})(-i\beta_1 \ket{0}+i\beta_0 \ket{1})
$$
$$
=\overline{i\alpha_0} (i\beta_0)+\overline{-i\alpha_1}(-i\beta_1)=\bra{\alpha}\ket{\beta}.
$$
Note that $\{\ket{0},\ket{1}\}$ is an orthonormal basis, so 
$$
\bra{0}\ket{0}=\bra{1}\ket{1}=1
$$ 
and 
$$
\bra{0}\ket{1}=\bra{1}\ket{0}=0.
$$

The {\em Hermitian conjugate} $M^{\dagger}$ of the linear operator $M$ is the operator which satisfies
$$
\bra{M\psi}\ket{\phi}=\langle \psi|M^{\dagger}\phi\rangle.
$$
An operator $M$ is {\em Hermitian} if $M=M^{\dagger}$. In matrix terms this is equivalent to the conjugate transpose being the same as the matrix itself. For example,
$$
\left(\begin{array}{cc} 1 & 2+i \\ 2-i & 2 \end{array} \right)
$$
defines a Hermitian operator on ${\mathbb C}^2$.

Let $M$ be a linear operator defined on a complex space with orthonormal basis $B$. The {\em trace} of $M$ is defined as
$$
\mathrm{tr}(M)=\sum_{\ket{\psi} \in B} \bra{\psi} M \ket {\psi}.
$$
We can easily prove that the trace of an operator does not depend on the basis chosen. Firstly, note that 
$$
\mathrm{tr}(MN)=\sum_{\ket{\psi} \in B} \bra{\psi} MN \ket {\psi}=\sum_{\ket{\psi},\ket{\phi} \in B} \bra{\psi} M\ket{\phi}\bra{\phi}N \ket {\psi}.
$$
$$
\sum_{\ket{\psi},\ket{\phi} \in B} \bra{\phi}N \ket {\psi}\bra{\psi} M\ket{\phi}=\sum_{\ket{\phi} \in B} \bra{\phi} NM \ket {\phi}=\mathrm{tr}(NM),
$$
hence 
$$
\mathrm{tr}(PMP^{-1})=\mathrm{tr}(P^{-1}PM)=\mathrm{tr}(M).
$$
In matrix terms, the trace is equal to the sum of the elements on the principal diagonal.

The {\em Pauli matrices},
$$
\sigma_0=\left(\begin{array}{cc} 1 & 0 \\ 0 & 1 \end{array} \right),
\sigma_x=\left(\begin{array}{cc} 0 & 1 \\ 1 & 0 \end{array} \right),
\sigma_z=\left(\begin{array}{cc} 1 & 0 \\ 0 & -1 \end{array} \right),
\sigma_y=\left(\begin{array}{cc} 0 & -i \\ i & 0 \end{array} \right),
$$
are unitary linear transformations of ${\mathbb C}^2$ which form a basis for the space of $2 \times 2$ matrices. 
In general, any error - also those which are not unitary - affecting a single qubit can be written as a linear combination of the Pauli matrices. 
We sometimes denote 
$\s_0, \s_x, \s_y, \s_z$ simply as $I,X,Y,Z$ respectively. Note that the Pauli matrices are both unitary and Hermitian.
They are also mutually orthogonal under the {\em Hilbert-Schmidt inner product} 
$$\langle A, B\rangle = \tr(A^\dag B).$$

A {\em measurement} or {\em observable} is represented by a hermitian operator. For example, the spin-up -- spin-down measurement $\hat{\sigma}_z$ is represented by the Pauli matrix $\sigma_z$~\footnote{This direction is commonly referred to 
as the ``z-direction'' in the x-y-z axis scheme.}.

The outcome of an individual measurement can only take two values. 
These correspond to the eigenvalues of $\sigma_z$ which are $+1$ and $-1$. After the measurement, 
the state is then found in the corresponding eigenstate: in $\ket{0}$ if the outcome $+1$ was obtained,
and in $\ket{1}$ if the outcome $-1$ was obtained. 
These occur with probabilities 
$$
p_0 = |\!\braket{\alpha}{0}\!|^2
$$ and
$$
p_1 = |\!\braket{\alpha}{1}\!|^2 ,
$$ 
respectively.

An {\em expectation value} is obtained by the repeated measurement 
of identically prepared spin particles.
Measuring the spin value of $\hat{\sigma}_z$ on a qubit 
$$
\ket{\alpha} = \alpha_o \ket{0} + \alpha_1 \ket{1}$$ 
yields the expectation value 
$$\langle \hat{\sigma}_z \rangle = \bra{\alpha} \sigma_z \ket{\alpha} = \tr(\sigma_z \dyad{\alpha}) = \alpha_0^2 - \alpha_1^2.$$

One can check that this leads to the correct expectation value
of 
$$
\langle \hat{\sigma}_z\rangle = p_0 \cdot (+1) + p_1 \cdot (-1) = \alpha_0^2 - \alpha_1^2 
= \bra{\alpha} \sigma_z \ket{\alpha}\,.$$

The above treatment can be generalised. Denote by $\hat{A}$ an observable which is represented by a Hermitian matrix $A$. Let $m_i$ and $\ket{m_i}$ be its eigenvalues and corresponding eigenvectors.
Measuring an observable $\hat{A}$ on a quantum state $\ket{\alpha}$
yields the values $m_i$ with probability $p_i = |\!\bra{\alpha}\ket{m_i}\!|^2$.
The state is found in the corresponding eigenstates afterwards.

This leads to the expectation value 
$$
\langle \hat{A} \rangle = \bra{\alpha} A \ket{\alpha} = \tr(A \dyad{\alpha}{\alpha}).
$$

The description of multiple quantum systems takes place in the tensor product space
of the individual Hilbert spaces.
Thus a system of $n$ qubits is described in the 
$n$-fold tensor product space of the one-qubit spaces.
One arrives at the $2^n$-dimensional Hilbert space 
$({\mathbb C}^2)^{\otimes n} = \mathbb C^2 \ot \cdots \ot \mathbb C^2$ ($n$ times).

A {\em density matrix} is used to describe a classical probability distribution 
(also called a statistical mixture or statistical ensemble) over quantum states.
Suppose that some source emits the quantum state $\ket{\phi_i}$ with probability $p_i$.
One requires that $p_i \geq 0$ and $\sum_i p_i = 1$.
From the discussion in the previous section, 
it is clear that the measurement of an observable $\hat A$ must yield an expectation value of 
$$
\langle \hat A \rangle = \sum_i p_i \bra{\phi_i}  A \ket{\phi_i} .
$$
By linearity, this can be rewritten as 
$$
\langle \hat A \rangle = \tr( A \sum_i p_i \dyad{\phi_i}).
$$
Indeed the operator 
$$
\rho = \sum_{i=1}^r p_i \dyad{\phi_i}
$$
captures all there is to know about a quantum system 
and $\rho$ is known as the density matrix describing it. 

For a complex matrix $\rho$ to represent a quantum state, one requires $\rho = \rho^\dag$, 
$\bra{\psi} \rho \ket{\psi} \geq 0$ for all $\ket{\psi}$ (positive-semidefinite) and $\tr(\rho) = 1$. 
Comparing with classical probability theory, 
this corresponds to a real valued, non-negative, and normalized probability distribution.
The density matrix formalism can indeed be seen as a generalization of classical probability theory 
and quantum mechanics can be taken to be the study of the cone formed by
complex positive-semidefinite matrices, and transformations thereof.
This is an analogy to the probability simplex encountered in classical probability theory.

Now we can state what we left out in preceding discussion about measurements:
consider the case when some eigenvalues of the measurement operator $A = \sum m_i \dyad{m_i}$ are equal, i.e. the spectrum of $A$ is degenerate. 
What is the probability for obtaining outcome~$i$ and what is the post-measurement state?
Let $P_j$ be the projector onto the eigenspace with eigenvalue $m_j$ of $A$.
Then a measurement yields outcome $m_j$ with probability $p_j = \tr(P_j \rho)$ and
the density operator immediately after the measurement reads

$$
\frac{P_j \rho P_j}{\tr(P_j \rho)}\,.
$$

The {\em time evolution} of an isolated qubit is given by a unitary operator in SU(2).
$$
\ket{\alpha} \mapsto U(t) \ket{\alpha}\,.
$$
On a closed quantum system of $n$ qubits, the time evolution is given by unitary operators on 
$\mathcal{H}_\text{system} = (\mathbb C^2)^{\ot n}$.
In case of a quantum system interacting with its environment 
such unitaries can also act on a larger system 
$$
\mathcal{H}_\text{system} \ot \mathcal{H}_\text{environment}.
$$
A unitary on such a larger system can, on $\mathcal{H}_\text{system}$, be represented in the 
(non-unique) operator-sum or Krauss decomposition as
$$
\ket{\alpha} \longmapsto \sum_i K_i \dyad{\alpha}{\alpha} K_i^\dag \quad \text{with the constraint} \quad \sum_i K_i^\dag K_i = \one.
$$
Throughout $\one$ will denote the identity map. The operators $K_i$ are also known as {\em Krauss operators}.

More generally, this reads for a density matrix as
$$
\rho \longmapsto \sum_i K_i \rho K_i^\dag \quad \text{with the constraint} \quad \sum_i K_i^\dag K_i = \one .
$$
The above map is also known as a {\em quantum channel} or {\em completely positive map} and represents the most general form of physical change a quantum state can undergo.
In the case of a classical (conventional) bit, an error is represented by the bit-flip $0 \leftrightarrows 1$. 
For qubits, we regard any non-identity unitary transformation or non-identity quantum channel as an {\em error}. 
We can decompose any unitary or quantum channel in terms of a matrix basis. 

A good choice is the {\em Pauli group}:
it is generated by all possible tensor products of the $4$ Pauli matrices, 
together with phases $\pm 1$ or $\pm i$. 
Observe that $\sigma_x$, $\sigma_z$ and $\sigma_y$ anti-commute. That is,
$$
\sigma_x \sigma_y=-\sigma_y \sigma_x\,, \quad 
\sigma_x \sigma_z=-\sigma_z \sigma_x\,, \quad
\sigma_y \sigma_z=-\sigma_z \sigma_y
$$
and that
$$
\sigma_x \sigma_y = i\sigma_z\,, \quad\quad
\sigma_y \sigma_z = i\sigma_x\,, \quad\quad
\sigma_z \sigma_x = i\sigma_y\,.
$$
Thus, the Pauli group $\mathcal{P}_n$ is a non-abelian group consisting of the $4^n$ tensor products of $\sigma_0$, $\sigma_x$, $\sigma_z$ and $\sigma_y$, which together with the four phases is a group of size $4^{n+1}$.

A {\em quantum error-correcting code} is a linear subspace $Q$ of ${(\mathbb C^2)}^{\otimes n}$ 
into which a number of logical qubits can be encoded such that all errors of 
a certain type can be detected and/or corrected.
The question we ask is thus: given a noisy channel $\mathcal{E}$, does there exist a recovery channel $\mathcal{R}$,
such that every density matrix $\rho$, for which the image of $\rho$ is contained in $Q$, can be recovered? In other words, for all density matrices
$\rho$ with spectral decomposition
$$
\rho= \sum_i p_i \ket{\phi_i}\bra{\phi_i},
$$ 
where $\ket{\phi_i} \in Q$, we require that 
$$
\mathcal{R} \circ \mathcal{E}(\rho) = \rho.
$$

\subsection{A $1$-qubit error-correcting quantum code}

A classical {\em code} is a subset of $A^n$, where $A$ is a finite set called the {\em alphabet} and $n$ is the {\em length} of the code. The repetition code is the simplest type of code in which each element $a\in A$ is encoded as $(a,a,\ldots,a)$, an $n$-tuple of $a$'s. For example, the binary repetition code of length $3$ is $\{(000),(111)\}$ and we encode 
$$
0 \mapsto 000
$$ and
$$
1 \mapsto 111.
$$
This encoding allows us to correct up to one error by taking a majority decision. 
In other words we decode the codewords
$$
000,001,010,100 \quad \mathrm{as} \quad 0
$$
and
$$
111,011,110,101 \quad \mathrm{as} \quad 1.
$$

Can we apply the same strategy to obtain a quantum code? Not quite. 
A quantum repetition code (on three qubits for example) does not exist, since we cannot map
$$
\ket{\alpha} \mapsto \ket{\alpha} \otimes \ket{\alpha}\otimes \ket{\alpha}\,.
$$
It would contradict the following (no-cloning) theorem.

\begin{theorem} (no-cloning)
There is no linear map which takes $\ket{\alpha}$ to $\ket{\alpha}\otimes \ket{\alpha}$ for all $\ket{\alpha} \in {(\mathbb C^2)}^{\otimes n}$.
\end{theorem}

\begin{proof}
Suppose there was such a map. Then
$$ \ket{\alpha} \mapsto \ket{\alpha} \ot \ket{\alpha}\,,$$
$$ \ket{\beta} \mapsto \ket{\beta} \ot \ket{\beta}$$
Such a map however is not linear, as
$$
\ket{\alpha} +\ket{\beta}  \mapsto (\ket{\alpha}+\ket{\beta}) \otimes
(\ket{\alpha}+\ket{\beta})
$$
$$
 \neq 
 \ket{\alpha} \otimes \ket{\alpha} 
 +\ket{\beta} \otimes \ket{\beta}.
$$
\end{proof}

However, we could try the following repetition-type code
$$
\alpha_0\ket{0}+\alpha_1 \ket{1} \mapsto \alpha_0\ket{000}+\alpha_1 \ket{111}.
$$
Above and from now on, we simplify notation
$
\ket{0} \otimes \ket{0}\   \mathrm{as} \ \ket{00},
$
etc.

Suppose now a ``bit-flip'' $\sigma_x$ happens on the second position. This gives
$$
\sigma_0 \ot \sigma_x \ot \sigma_0 \big(\alpha_0\ket{000}+\alpha_1 \ket{111} \big) 
= \alpha_0\ket{010}+\alpha_1 \ket{101}.
$$

One can correct such an error by majority decision,
$$
 \alpha_0\ket{010}+\alpha_1 \ket{101} \quad \mathrm{decodes} \ \mathrm{as}\quad   \alpha_0\ket{000}+\alpha_1 \ket{111}.
$$
One needs a measurement that indicates exactly where the bit-flip has occurred. This can be done, as will be explained in Example~\ref{spread1}.

However, we cannot correct a single $\sigma_z$ error since
$$
 \alpha_0\ket{000}-\alpha_1 \ket{111}
$$
is also a possible state of our code.

Shor \cite{Shor1995} was the first to introduce a quantum code which can correct any single-qubit error. 
He circumvented this apparent problem by introducing a majority decision on the signs to correct a $\sigma_z$ error.

\begin{example} \label{shorcode1} (Shor code)

The coding space for the Shor code is $({\mathbb C}^2)^{\otimes 9}$ and a qubit is encoded as
$$
\ket{\alpha} \mapsto \ket{\alpha_L}
$$
according to 
$$
\ket{0_L}=(\ket{000}+\ket{111})\otimes   (\ket{000}+\ket{111})\otimes (\ket{000}+\ket{111})
$$
and
$$
\ket{1_L}=(\ket{000}-\ket{111})\otimes (\ket{000}-\ket{111})\otimes (\ket{000}-\ket{111}).
$$
Hence, by linearity,
$$
\begin{array}{rl}
\alpha_0\ket{0}+\alpha_1 \ket{1}  \mapsto & \alpha_0 (\ket{000}+\ket{111})\otimes   (\ket{000}+\ket{111})\otimes (\ket{000}+\ket{111}) \\
 + & \alpha_1 (\ket{000}-\ket{111})\otimes (\ket{000}-\ket{111})\otimes (\ket{000}-\ket{111}).
 \end{array}
$$

Suppose that we have a $\sigma_x$ error (bit-flip) occuring on the $4$-th bit. Then the $\alpha_0$ term would change to
$$
(\ket{000}+\ket{111})\otimes   (\ket{100}+\ket{011})\otimes (\ket{000}+\ket{111})
$$
which we would detect and correct by taking the majority decision as with the classical error-correcting code, so we decode
$$
\ket{100}+\ket{011} \ \mathrm{as} \ \ket{000}+\ket{111}.
$$

Now suppose we have $\sigma_z$ error (phase error) occuring on the $7$-th bit. Then the $\alpha_0$ term would be
$$
(\ket{000}+\ket{111})\otimes   (\ket{000}+\ket{111})\otimes (\ket{000}-\ket{111})
$$
which we would detect and correct by taking the majority decision on the signs.

Since $\sigma_y=i\sigma_x \sigma_z$, we can also correct $\sigma_y$ errors since the two decisions we made above are independent of each other. Note that the scalar $i$ does not play a role in the decoding.

\end{example}

\subsection{The orthogonal projection onto a subspace}

Let $Q$ be a subspace of ${(\mathbb C^2)}^{\otimes n}$ and let $Q^{\perp}$ be its orthogonal subspace with respect to the standard inner product defined on ${(\mathbb C^2)}^{\otimes n} \cong {\mathbb C}^{2^n}$. Any vector $\ket{\psi}$ can be written (uniquely) as the sum of a vector $P\ket{\psi} \in Q$ and $P^{\perp}\ket{\psi}\in Q^{\perp}$. The map 
$$
\ket{\psi} \rightarrow P\ket{\psi}
$$
is a linear map, called the {\em orthogonal projection} onto $Q$. 

\begin{lemma} \label{projectiondecomp}
If $\{\ket{\psi_1},\ket{\psi_2},\ldots,\ket{\psi_k} \}$ is an orthonormal basis for $Q$ then
$$
P=\sum_{i=1}^k \ket{\psi_i}\bra{\psi_i}.
$$
\end{lemma}

\begin{proof}
For any $j \leqslant k$,
$$
P\ket{\psi_j}=\sum_{i=1}^k \ket{\psi_i}\bra{\psi_i}\ket{\psi_j}=\ket{\psi_j},
$$
so  
$
P\ket{\psi}=\ket{\psi}
$ 
for all $\ket{\psi} \in Q$.

Furthermore,
$$
P\ket{\psi}=\sum_{i=1}^k \ket{\psi_i}\bra{\psi_i}\ket{\psi}=0
$$
for all $\ket{\psi}\in Q^{\perp}$.
\end{proof}

Clearly, by definition, $P^2=P$. By Lemma~\ref{projectiondecomp}, $P$ is Hermitian, since it is the sum of Hermitian operators. The following lemma implies that this is enough to characterise $P$.

\begin{lemma} \label{squareisitself}
If $P$ is a linear Hermitian operator for which $P^2=P$ and whose image is $Q$ then $P$ is the orthogonal projection onto $Q$.
\end{lemma}

\begin{proof}
The operator $P$ is Hermitian, so it is diagonalisable with real eigenvalues. 
Since $P^2=P$, its eigenvalues are $0$ and $1$. By the spectral decomposition theorem,
$$
P=\sum_{i=1}^k \ket{\psi_i}\bra{\psi_i},
$$
where $\{\ket{\psi_1},\ket{\psi_2},\ldots,\ket{\psi_k} \}$ is an orthonormal basis for its eigenspace with eigenvalue $1$. Since 
 $$
 P\ket{\psi_j}=\ket{\psi_j}
 $$
 for all $j=1,\ldots,k$, the eigenspace with eigenvalue $1$ contains $\mathrm{im}(P)$, the image of $P$.

The eigenspace with eigenvalue $0$ is $\mathrm{im}(P)^{\perp}$. Thus, $P$ is the orthogonal projection onto $\mathrm{im}(P)$.
\end{proof}

\subsection{Error-detection and correction} \label{errorsection}

For the reliable transmission of an (unknown) quantum system over a noisy channel,
we are now faced with three major challenges.

\begin{enumerate}
 \item Measurement disturbance.
 As explained in Section~\ref{sect:intro}, 
measurements induce an ``update'' of the state that is measured. 
Thus, when obtaining error syndromes in order to understand
what error has occurred, the underlying quantum state may be altered.

 \item  Continuous set of errors.
The set of errors is continuous and not discrete. How can we distinguish and correct for 
an error set this large?

\item  No-cloning.
Unknown quantum states cannot be copied. 
Thus an approach of adding redundancy as done for a classical repetition code is bound to fail.
\end{enumerate}

How can these challenges be overcome?
Firstly, the syndrome measurements are chosen such that they stabilise the set of quantum states 
that consist of the code. In this way, all code states remain unchanged when extracting the syndromes, 
while erroneous states are changed in reversible fashion.
Second, the linearity of quantum mechanics implies that when some discrete set 
of errors can be corrected, then one can correct all errors which lie in their span. 
We shall not show a proof of this here, but one can be found in \cite[Theorem 2]{Gottesman2009} and \cite{BrunLidar2013}.
Lastly, the encoded quantum information is distributed amongst many systems and thus 
``hidden'' from any noisy channel. In this way the state does not have to be copied and no
redundancy is added. This not only gives rise to the below Knill-Laflamme conditions on error correction, but also provides an information theoretic interpretation of quantum error-correction.

In quantum error-correction one is faced with the following task.
Let 
$$
\mathcal{N}(\cdot) = \sum_\mu E_\mu (\cdot) E_\mu^\dag, \quad \text{where}
\quad
\sum_\mu E_\mu ^\dag E_\mu= \one,
$$ 
be a quantum channel.
Given the channel $\mathcal{N}$, for which codes $Q$ does there exist a recovery
channel $\mathcal{R}$ such that $\mathcal{R} \circ \mathcal{N}(\rho) = \rho$ for all 
$$
\rho= \sum_i p_i \ket{\phi_i}\bra{\phi_i},
$$ 
where $\ket{\phi_i} \in Q$?

It turns out that the set of correctable states form subspaces.
The following theorem gives a necessary and sufficient condition for a recovery channel to exist.

\begin{theorem}[Knill-Laflamme conditions] \label{KLtheorem}
Let $\mathcal{Q}$ be a subspace of $(\mathbb C^d)^{\ot n}$.
The channel $\mathcal{N}(\cdot) = \sum_\mu E_\mu (\cdot) E_\mu^\dag$ can be corrected 
by a code $\QQ$
if and only if for all $\ket{\phi}, \ket{\psi}$ in $\QQ$ and errors $E_\mu, E_\nu$
$$
\bra{\phi} E_\mu^\dag E_\nu \ket{\psi} = c_{\mu\nu} \bra{\phi}\ket{\psi}, 
$$
for some $c_{\mu \nu} \in \mathbb C$.
\end{theorem}

This condition implies the following two essential properties.\\
\noindent 1. Orthogonal code states remain orthogonal under the action of errors,
$$
\mathrm{if} \ \  \braket{\phi}{\psi} = 0  \ \ \mathrm{then} \ \  \bra{\phi} E_\mu^\dag E_\nu \ket{\psi} = 0,
$$
and thus orthogonal codewords remain orthogonal under the noise.\\
\noindent 2. 
The expectation value of $E_\mu^\dag E_\nu$ is constant when 
$\ket{\phi}$ ranges over the set of code states. In other words, 
for all quantum states $\ket{\phi}, \ket{\psi} \in Q$,
$$
\tr[\dyad{\phi} E_\mu^\dag E_\nu] = \bra{\phi} E_\mu^\dag E_\nu \ket{\phi} = \bra{\psi} E_\mu^\dag E_\nu \ket{\psi} = c_{\mu \nu},
$$
In this way, the encoded quantum information is ``hidden'' from the noisy channel.

Lastly, a set of errors $\mathcal{E}$ is said to be {\em detectable} if and only if 
all errors $E_\mu^\dag E_\nu$ with $E_\mu, E_\nu \in \mathcal{E}$ are correctable.

\subsection{Error weights}
We define the {\em weight} $\operatorname{wt}(M)$ of an operator $M$ in the Pauli group $\mathcal{P}_n$ to be 
the number of tensor factors which are not equal to $\sigma_0$. For example,
$$
M=\sigma_x \otimes \sigma_z \otimes \sigma_0 \otimes \sigma_y \otimes \sigma_0
$$
has weight three.

In classical codes the distance between any two elements of $A^n$ is the number of coordinates in which they differ. If the minimum distance of a code $C$ is at least $2t+1$ then $C$ is a $t$-error correcting code (i.e. we can correct errors if up to $t$ coordinates of a codeword change). In quantum codes the same holds, if a quantum code can detect all errors of weight less than $2t+1$ then it is a $t$-error correcting code.

\section{Qubit stabilizer codes} \label{sectionqubit}

\subsection{Definition and examples}

Most quantum codes presently known are stabilizer codes, and their usefulness lies partially 
in the fact that their connection with classical codes allows for them to be described in an efficient way. Here, we will mainly deal with stabilizer codes, although we will also see examples of quantum codes in Section~\ref{unionstabilizers} which are not stabilizer codes.

A {\em qubit stabilizer code} $Q(S)$ is the joint eigenspace with eigenvalue $1$ of the elements 
of an abelian subgroup $S$ of $\mathcal{P}_n$ not containing $-\one$. 
The subgroup $S$ is also known as the {\em stabilizer}.

We will often define $S$ as being generated by a set of $n-k$ commuting 
independent generators $M_1,\ldots,M_{n-k}$ of $\mathcal{P}_n$.
By independent, we mean that $M_1, \dots, M_{n-k}$ generate $S$, 
 $$
 \langle M_1, \dots, M_{n-k} \rangle = \big\{\prod M_1^{\alpha_1} \cdots M_{n-k}^{\alpha_{n-k}} \,\big|\, \alpha_1, \dots, \alpha_{n-k} \in \{0,1\} \big\} = S
 $$
while any smaller subset does not.
Thus, the set of $M_i$'s are called {\em generators}.

It is important to note that we require $-\one \not\in S$, since otherwise $Q(S)=\{ 0 \}$. We also assume that there is no coordinate in which every element of $S$ has a $\sigma_0$ in that coordinate, as we could simply delete this coordinate and this would not affect the error correcting capabilities of the code.

Note that the phase of any element in $S$ is $\pm 1$, since if
$$
M=\pm i \sigma_1 \otimes \cdots \otimes \sigma_n
$$
then
$$
M^2=- \one \in S,
$$
which, as mentioned above, implies that $Q(S)=\{ 0 \}$.

\begin{example}
Suppose $n=2$ and $S$ is generated by a single Pauli operator $M=\sigma_x \otimes \sigma_z$. 

Let $\ket{\alpha} \in ({\mathbb C}^2)^{\otimes 2}$. Then $\ket{\alpha}$ can be written as
$$
\ket{\alpha}=\alpha_{00} \ket{00}+\alpha_{01} \ket{01}+\alpha_{10} \ket{10}+\alpha_{11} \ket{11}
$$
for some $\alpha_{ij} \in {\mathbb C}$.
Now,
$$
M\ket{\alpha}=\alpha_{00} \ket{10}-\alpha_{01} \ket{11}+\alpha_{10} \ket{00}-\alpha_{11} \ket{01}
$$
Thus, $\ket{\alpha}$ is in the eigenspace of $M$ with eigenvalue $1$ if and only if
$$
\alpha_{00}=\alpha_{10}, \ \ \alpha_{01}=-\alpha_{11}.
$$
We note that the dimension of $Q(S)$ is $2$.
\end{example}

We often use the short-hand notation $\sigma_0=I$, $\sigma_x=X$, $\sigma_y=Y$  and $\sigma_z=Z$, so in the previous example we might write $M=XZ$.

\begin{example} \label{nis3}
Suppose $n=3$ and $S$ is generated by $M_1,M_2,M_3$, where
\begin{align}
M_1&=\sigma_{0} \otimes \sigma_x \otimes \sigma_z\nonumber\\
M_2&=\sigma_{0} \otimes \sigma_y \otimes \sigma_x\nonumber\\
M_3&=\sigma_{x} \otimes \sigma_z \otimes \sigma_y\nonumber\,. 
\end{align}
In the shorthand notation we would write that $S$ is defined by 
$$
\begin{array}{rccc} M_1= & I & X & Z \\ M_2= & I & Y & X\\ M_3= & X & Z & Y\end{array}.
$$
Observe that $M_iM_j=M_jM_i$ for all $i$ and $j \in \{1,2,3\}$. For example
$$
M_2M_1=(\sigma_{0} \otimes \sigma_y \otimes \sigma_x)(\sigma_{0} \otimes \sigma_x \otimes \sigma_z)
=\sigma_{0} \otimes (-i\sigma_z) \otimes(- i\sigma_y)=-
\sigma_{0} \otimes \sigma_z \otimes \sigma_y
$$
and
$$
M_1M_2=(\sigma_{0} \otimes \sigma_x \otimes \sigma_z)(\sigma_{0} \otimes \sigma_y \otimes \sigma_x)
=\sigma_{0} \otimes i\sigma_z \otimes i\sigma_y=-
\sigma_{0} \otimes \sigma_z \otimes \sigma_y.
$$
This can be checked quickly by verifying that 
different Pauli matrices $\{\sigma_x,\sigma_y,\sigma_z\}$ 
coincide in the same position in $M_i$ and $M_j$ ($i\neq j$) an even number of times.

To find a basis for the stabilizer code, suppose that
$$
\ket{\alpha}=\sum_{ijk} \alpha_{ijk} \ket{ijk}.
$$
is in the code space, i.e. that $\alpha$ is in the $+1$-eigenspace of all $M_i$.

Since
$$
M_1\ket{\alpha}=\sum_{j=0}^{1}(\alpha_{j00} \ket{j10}-\alpha_{j01} \ket{j11}+\alpha_{j10} \ket{j00}-\alpha_{j11} \ket{j01})
$$

We have that $\ket{\alpha}$ is in the $+1$-eigenspace $\tilde M_1 = \operatorname{Im}(I + M_1)$ of $M_1$ if and only if
$$
\alpha_{j00}=\alpha_{j10}\quad \text{and}\quad \alpha_{j01}=-\alpha_{j11}.
$$
Similarly,
$$
M_2\ket{\alpha}= i \sum_{j=0}^{1}
(\alpha_{j00} \ket{j11} + \alpha_{j01} \ket{j10} - \alpha_{j10} \ket{j01} - \alpha_{j11} \ket{j00})
$$
Thus, $\ket{\alpha}$ is in the $+1$-eigenspace $\tilde M_2$ if and only if
$$
i\alpha_{j00}=\alpha_{j11} \quad \text{and}\quad  \alpha_{j01}=-i\alpha_{j10}.
$$

Finally,
\begin{align}
M_3\ket{\alpha}= 
i(&\alpha_{000} \ket{101}-\alpha_{001} \ket{100}-\alpha_{010}\ket{111}+\alpha_{011} \ket{110} \nonumber\\
+\,&\alpha_{100} \ket{001}-\alpha_{101} \ket{000}-\alpha_{110}\ket{011}+\alpha_{111} \ket{010})\,, \nonumber 
\end{align}
so
$\ket{\alpha}$ is in the $+1$-eigenspace $\tilde M_3$ if and only if
$$
i\alpha_{000}=\alpha_{101}, \ \ \alpha_{100}=-i\alpha_{001},\ \ \alpha_{111}=-i\alpha_{010},\ \ \alpha_{110}=i\alpha_{011}.
$$
Thus,
$$
Q(S)= \tilde M_1 \cap \tilde M_2 \cap \tilde M_3
$$
is the one-dimensional subspace spanned by
$$
\ket{000}-i \ket{001}+ \ket{010}+i \ket{011}-\ket{100}+i \ket{101}- \ket{110}-i \ket{111}.
$$

\end{example}

In fact, we seldom actually calculate a basis as for $Q(S)$ as it is not necessary in practice. We have only calculated this previous example so one gets a feel of how laborious this is even for small parameters. From a practical point of view it is enough to know the orthogonal projection $P$ for the subspace $Q$. 

\subsection{The dimension and minimum distance of a stabilizer code}

Let $S$ be an abelian subgroup of $\mathcal{P}_n$. Let $Q(S)$ be the subspace defined as the joint eigenspace of eigenvalue $1$ of the elements of $S$. Let $P = P(S)$ be the orthogonal projection onto the subspace $Q(S)$.

\begin{lemma} \label{sumthemall}
The orthogonal projection is
$$
P = \frac{1}{|S|} \sum_{E\in S} E.
$$
\end{lemma}

\begin{proof}
Since $S$ is an abelian subgroup, one has $$MP=PM=P$$ for all $M \in S$.

Suppose that $\ket{\psi} \in Q(S)$. Then,
$
P\ket{\psi}=\ket{\psi}
$ and therefore $\ket{\psi} \in \mathrm{im}(P)$.

Vice versa, if $\ket{\psi} \in \mathrm{im}(P)$ then, for all $M \in S$,
$$
M \ket{\psi}=MP \ket{\phi}=P \ket{\phi}=\ket{\psi},
$$
so $\ket{\psi} \in Q(S)$.
Thus, $Q(S)=\mathrm{im}(P)$.

Since $E^{\dagger}=E$ for all $E \in \mathcal{P}_n$, we have that $P^{\dagger}=P$. Moreover,
$$
P^2=P\frac{1}{|S|} \sum_{M\in S} M=\frac{1}{|S|} \sum_{M\in S} PM=\frac{1}{|S|} \sum_{M\in S}M=P.
$$

By Lemma~\ref{squareisitself}, $P=P(S)$. 

\end{proof}

For the proof of the next theorem, it is worth noting that
$$
\mathrm{tr}( \sigma_1 \otimes \cdots \otimes \sigma_n)=\mathrm{tr}( \sigma_1)\cdots\mathrm{tr}( \sigma_n).
$$
Thus, for all $E\in \mathcal{P}_n$ with phase $\pm 1$, where $E\neq \pm \one$, $\mathrm{tr}(E)=0$ and that $\mathrm{tr}(\one)=2^n$.

\begin{theorem} \label{qubitdim}
The stabilizer code $Q(S)$ which is the joint $+1$-eigenspace of an abelian subgroup $S$ generated by $n-k$ independent elements has dimension $2^{k}$.
\end{theorem}

\begin{proof}
By Lemma~\ref{sumthemall}, the orthogonal projection onto $Q(S)$ is
$$
P=\frac{1}{|S|} \sum_{M\in S} M.
$$
The image of $P$ is its eigenspace of eigenvalue one and also $Q(S)$. 

The operator $P$ is Hermitian and thus diagonalisable.  Since $P^2=P$ its eigenvalues are $0$ and $1$. 
The trace of $P$ is equal to the sum of its eigenvalues, which in the case of $P$ is the dimension of the eigenspace of eigenvalue one. 
Therefore, the dimension of $Q(S)$ is equal to the trace of $P(S)$. 

It only remains to note that 

$$\mathrm{tr}(M)=0$$

for all $M \in \mathcal{P}_n$ with the exception of $M=\one$, in which case $\mathrm{tr}(\one)=2^n$.
Thus, $ \dim Q=2^n/|S|=2^k$.

\end{proof}

Having ascertained the dimension of a stabilizer code, we go on to determine its minimum distance. 

Let $\mathrm{Centraliser}(S)$ denote the set of elements of $\mathcal{P}_n$ that commute with all elements of $S$, i.e. the centraliser of $S$ in the group $\mathcal{P}_n$.

\begin{lemma} \label{undetectstab}
$E$ is an undetectable error for $Q(S)$ if and only if $E \in \mathrm{Centraliser}(S)~\setminus~S$.
\end{lemma}

\begin{proof}
We proceed by contradiction.

($\Rightarrow$) Suppose that $E$ is undetectable but that $E \not\in \mathrm{Centraliser}(S) \setminus S$. 
 
Since any two elements of $\mathcal{P}_n$ either commute or anti-commute, $E \not\in \mathrm{Centraliser}(S)$ implies there is a $M \in S$ such that
$$
EM=-ME.
$$
Take any $\ket{\psi},\ket{\phi} \in Q(S)$ with $\bra{\psi}\ket{\phi}=0$. Then
$$
\bra{\psi}E\ket{\phi}=\bra{\psi}ME\ket{\phi}=-\bra{\psi}EM\ket{\phi}=-\bra{\psi}E\ket{\phi},
$$
which implies $\bra{\psi}E\ket{\phi}=0$. 

If $E \in S$ then
$$
\bra{\psi}E\ket{\phi}=\bra{\psi}\ket{\phi},
$$

Hence, by Theorem~\ref{KLtheorem}, $E$ is detectable, a contradiction.

($\Leftarrow$) Suppose that $E$ is detectable with $E \in \mathrm{Centraliser}(S) \setminus S$. 
Let $\ket{\psi} \in Q(S)$. Since $E \in \mathrm{Centraliser}(S)$, 
$$
ME\ket{\psi}=EM\ket{\psi}=E\ket{\psi}
$$
holds for all $M\in S$, which implies that $E\ket{\psi} \in Q$. 

Extend $\{\ket{\psi}\}$ to an orthonormal basis $B$ for $Q$. 
Since $E$ is detectable, 
$$
\bra{\phi}E\ket{\psi}=0
$$
for all $\ket{\phi} \in B \setminus \{\ket{\psi}\}$. This implies that 
$E\ket{\psi}$ is in the subspace $(B \setminus \{\ket{\psi}\})^{\perp}$.
Since this subspace has as a basis $\{ \ket{\psi}\}$,
$$
E\ket{\psi}=\lambda_{\psi} \ket{\psi},
$$
for some $\lambda_{\psi} \in {\mathbb C}$. Hence, $\ket{\psi}$ is an eigenvector of $E$.

By Theorem~\ref{KLtheorem},
$$
\bra{\phi}E\ket{\phi}=\lambda_E,
$$
for all $\ket{\phi} \in B$.
Since $\bra{\psi}\ket{\psi}=1$, this implies that $\lambda_{\psi}=\lambda_E$. 

The same argument as made above for $\ket{\psi}$ holds for all $\ket{\phi}\in Q(S)$. Thus, for all $\ket{\phi}\in Q(S)$,
$$
E\ket{\phi}=\lambda_E \ket{\phi}.
$$

Since $E \not\in S$,  $\lambda_E \neq 1$. 

The subgroup generated by $S$ and $\lambda_E^{-1}E$ defines a smaller stabilizer code, so there is a $\ket{\psi}\in Q$ such that
$$
\lambda_E^{-1}E \ket{\psi} \neq \ket{\psi},
$$
contradicting the above.
Hence, $E$ is not detectable.
\end{proof}

In the case that $k=0$, we have that $Q(S)$ is a $1$-dimensional subspace so cannot be used to store quantum information and all errors are correctable according to the definition. However, we do not rule out considering such codes since for any proper subgroup $S'$ of $S$, the code $Q(S')$ will be of interest. Since the elements of $S \setminus S'$ will be in $\mathrm{Centraliser}(S') \setminus S'$, Theorem{~\ref{stabdistance} indicates that it makes sense to define the minimum distance of $Q(S)$ to be equal to the minimum weight of the non-identity elements of $S$. These codes are called {\em self-dual}, for reasons that will become clear in Theorem~\ref{stabtheorem}.

\begin{theorem} \label{stabdistance}
If $k \geqslant 1$ then the minimum distance of the $2^k$-dimensional stabilizer code $Q(S)$ with stabilizer group $S$ is equal to the minimum weight of the errors in $\mathrm{Centraliser}(S) \setminus S$.
\end{theorem}

\begin{proof}
By Lemma~\ref{undetectstab}, $Q(S)$ can detect all errors which are not elements of $\mathrm{Centraliser}(S) \setminus S$. 
In particular, it can also detect all errors of weight less than the minimum weight of an error in $\mathrm{Centraliser}(S) \setminus S$. 
\end{proof}

If there are elements of $S$ whose weight is less than the minimum distance of $Q(S)$ then the code is called {\em impure}. If this is not the case then the code is called {\em pure}.

We should mention that there is also the concept of a degenerate code. According to Calderbank et al. \cite{CRSS1998}, a nondegenerate code is one for which different errors produce linearly independent results when applied to elements of the code. Whereas a code is pure if distinct errors produce orthogonal results. It is
straightforward to verify that, for additive codes, ‘pure’ and
‘nondegenerate’ coincide. In general, however, a pure code is
nondegenerate but the converse need not be true.

We use the shorthand notation $(\!(n,K,d)\!)$ to denote a quantum code of ${(\mathbb C^2)}^{\otimes n}$ of dimension $K$ and minimum distance $d$. 
The notation  $[\![n,k,d]\!]$ denotes a quantum code of dimension $2^k$. 
If it is a stabilizer code $Q(S)$ then $d$ is equal to the minimum weight of the elements in $\mathrm{Centraliser}(S) \setminus S$.

We now rewrite the Shor code from Example~\ref{shorcode1} as a stabilizer code. 

\begin{example} \label{shorcode2} (A $[\![9,1,3]\!]$ code)
Let $S$ be the subgroup generated by the following elements of $\mathcal{P}_9$.
$$
\begin{array}{cccc}
 \vspace{.1 cm}
M_1 & = & \sigma_z \otimes \sigma_z \otimes \sigma_{0} \otimes \sigma_{0} \otimes \sigma_{0} \otimes \sigma_{0} \otimes \sigma_{0} \otimes \sigma_{0} \otimes \sigma_{0} \\  \vspace{.1 cm}
M_2 & = & \sigma_{0} \otimes \sigma_z \otimes \sigma_z \otimes \sigma_{0} \otimes \sigma_{0} \otimes \sigma_{0} \otimes \sigma_{0} \otimes \sigma_{0} \otimes \sigma_{0}\\ \vspace{.1 cm}
M_3 & = &  \sigma_{0} \otimes \sigma_{0} \otimes \sigma_{0} \otimes\sigma_z \otimes \sigma_z \otimes \sigma_{0} \otimes \sigma_{0} \otimes \sigma_{0} \otimes \sigma_{0}\\ \vspace{.1 cm}
M_4 & = & \sigma_{0} \otimes \sigma_{0} \otimes \sigma_{0} \otimes\sigma_{0} \otimes \sigma_z \otimes \sigma_z \otimes  \sigma_{0} \otimes \sigma_{0} \otimes \sigma_{0}\\ \vspace{.1 cm}
M_5 & = &  \sigma_{0} \otimes \sigma_{0} \otimes \sigma_{0} \otimes  \sigma_{0} \otimes \sigma_{0} \otimes \sigma_{0}\otimes  \sigma_z \otimes \sigma_z \otimes \sigma_{0}\\ \vspace{.1 cm}
M_6 & = & \sigma_{0} \otimes \sigma_{0} \otimes \sigma_{0} \otimes  \sigma_{0} \otimes \sigma_{0} \otimes \sigma_{0} \otimes\sigma_{0} \otimes \sigma_z \otimes \sigma_z \\ \vspace{.1 cm}
M_7  & = &  \sigma_{x} \otimes \sigma_{x} \otimes \sigma_{x} \otimes  \sigma_{x} \otimes \sigma_{x} \otimes \sigma_{x}\otimes  \sigma_{0} \otimes \sigma_{0}  \otimes \sigma_{0}\\ \vspace{.1 cm}
M_8  & = & \sigma_{0} \otimes \sigma_{0} \otimes \sigma_{0} \otimes  \sigma_{x} \otimes \sigma_{x} \otimes \sigma_{x} \otimes \sigma_{x} \otimes \sigma_x \otimes \sigma_x
\end{array} 
$$

In shorthand notation this would be written in the following way.
$$
\begin{array}{rccccccccc} 
M_I= & Z & Z & I & I & I & I & I & I &I \\ 
M_2 =& I & Z & Z & I & I & I & I & I & I \\ 
M_3= & I &I & I & Z & Z & I & I &  I & I \\
M_4= &  I & I &  I & I & Z & Z & I & I & I \\
M_5 =&   I & I & I & I & I & I & Z & Z & I \\ 
 M_6 =&  I & I & I & I & I & I & I & Z & Z \\
M_7=&    X & X & X & X & X & X & I & I & I \\
 M_8 =&    I & I& I & X & X & X & X &  X & X \end{array}
$$

One can check that $M_i$ and $M_j$ commute for any $i$ and $j$. 

Suppose that $E$ is an error of weight at most $2$. We want to prove that $E \in S$ or $E$ does not commute with some $M_i$. 

We proceed with a case-by-case analysis.

If $E$ has weight one and a single $X$ or $Y$ then it does not commute with one of $M_1,\ldots,M_6$. 
If $E$ has weight one and a single $Z$ then it does not commute with one of $M_7,M_8$. 

If $E$ has weight two which are both $X$ then, without loss of generality, suppose there is a $X$ in the first system. Then $E$ must have a $X$ or $Y$ in the second system so that it commutes with $M_1$. But then it must also have a $X$ or $Z$ in the third system so that it commutes with $M_2$, contradicting the fact that it has weight two.

We leave the case-by-case analysis as an exercise but conclude that the only errors of weight two which commute with all the $M_i$ are precisely those which are in $S$, i.e. $M_1,\ldots,M_6,M_1M_2,M_3M_4,M_5M_6$. 

We will prove that the minimum distance of this code is $3$ in a very simple manner once we have determined its geometry.

An important observation here is that the Shor code is {\em impure} since $S$ contains errors of weight $2$, whereas the minimum distance is $3$.
\end{example}

We can store the same amount of information on fewer qubits with the following code.

\begin{example}
\label{spread1} (A $[\![5,1,3]\!]$ code)
Let $S$ be the subgroup generated by the following elements of $\mathcal{P}_5$.
$$
\begin{array}{rccccc} 
M_I= & X & Z & Z & I & X  \\ 
M_2= & Z & X & I & Z & X  \\ 
M_3 = & I & Z & X & Z & Y  \\
M_4 = & Z & I &  Z & X & Y  \\
\end{array}
$$
This matrix makes the task of checking that $M_iM_j=M_jM_i$ fairly quick.  We will prove that the minimum distance is $3$ by considering its geometry in Example~\ref{5zero3code2}.

Let us see how we can use this example to correct errors of weight one. We perform measurements $\hat{M_i}$ on $E\ket{\phi}$. This will return a value $\pm 1$ (the eigenvalues of $M_i$). This gives us a ``syndrome'', a $4$-tuple of signs for each error $E$. These are given in the following tables. 

$$\begin{array}{c|cccc} 
 & M_1 & M_2 & M_3 & M_4 \\ \hline
XIIII & + & - & + & - \\ 
IXIII & - & + & - & + \\ 
IIXII & - & + & + & - \\ 
IIIXI & + & - & - & + \\ 
IIIIX & + & + & - & - \\ 
\end{array}
\begin{array}{c|cccc} 
 & M_1 & M_2 & M_3 & M_4 \\ \hline
ZIIII & - & + & + & + \\ 
IZIII & + & - & + & + \\ 
IIZII & + & + & - & + \\ 
IIIZI & + & + & + & - \\ 
IIIIZ & - & - & - & - \\ 
\end{array}
$$
$$
\begin{array}{c|cccc} 
 & M_1 & M_2 & M_3 & M_4 \\ \hline
YIIII & - & - & + & - \\ 
IYIII & - & - & - & + \\ 
IIYII & - & + & - & - \\ 
IIIYI & + & - & - & - \\ 
IIIIY & - & - & + & + \\ 
\end{array}
$$
Since each syndrome is distinct we can use this look-up table to identify the error and correct it. An important observation here is that when we perform the measurement $\hat{M_i}$, only the sign of the state can possibly change. Since 
$$
M_iE\ket{\phi}=\pm EM_i \ket{\phi}=\pm E\ket{\phi},
$$
$E\ket{\phi}$ is an eigenvector of $M_i$, so after measuring we will be in the state $\pm E\ket{\phi}$.
Thus, we can measure consecutively each measurement $\hat{M_i}$, for $i=1,\ldots,n-k$.
\end{example}

\subsection{Qubit stabilizer codes as binary linear codes}

In this section we introduce a connection between qubit stabilizer codes and classical binary linear codes. We will go on to exploit this connection to construct qubit quantum codes and then to realise a more general connection between stabilizer codes and classical codes.

Let ${\mathbb F}_q$ denote the finite field with $q$ elements.
Consider the map 
$$
\tau : \{\sigma_0,\sigma_x,\sigma_y,\sigma_z\} \rightarrow {\mathbb F}_2^2
$$ 
defined by the following table.
$$
\tau : \begin{cases}
\begin{array}{rcl} \vspace{0.1 cm}
\sigma_0 & \mapsto & (0 | 0) \\ \vspace{0.1 cm}
\sigma_x & \mapsto & (1 | 0) \\ \vspace{0.1 cm}
\sigma_z & \mapsto & (0 | 1) \\ \vspace{0.1 cm}
\sigma_y & \mapsto & (1 | 1) \\ 
\end{array}
\end{cases}
$$
We extend the map $\tau$ to $\mathcal{P}_n$ by applying $\tau$ to an element of $\mathcal{P}_n$ coordinatewise, where the image of the $j$-th position of $M$ is the $j$ and $(j+n)$-th coordinate in $\tau(M)$. For example,
$$
\tau(\sigma_x \otimes \sigma_y \otimes \sigma_0 \otimes \sigma_x \otimes \sigma_z)=(1 1 0 1 0 \ |\ 0 1 0 0 1).
$$
We draw the line between the $n$ and $(n+1)$-st coordinate, for readability sake.
We ignore the phase, so $\tau(\lambda M)=\tau(M)$ for all $\lambda \in \{ \pm 1, \pm i\}$. Effectively, this defines the domain of the map $\tau$ as $\mathcal{P}_n/ \{ \pm 1, \pm i\}$.

\begin{lemma} \label{tauadd}
For all $M,N \in \mathcal{P}_n/ \{ \pm 1, \pm i\}$,
$$
\tau(MN)=\tau(M)+\tau(N).
$$
\end{lemma}

\begin{proof}
Observe that the multiplicative structure up to a phase factor (for example we ignore the $i$ in $\sigma_y=i\sigma_x \sigma_z$) is isomorphic to the additive structure of 
${\mathbb F}_2^2$.
\end{proof}

We have established a bijection between the elements of $\mathcal{P}_n/ \{ \pm 1, \pm i\}$ and ${\mathbb F}_2^{2n}$. 
The above lemma implies that a subgroup $S$ of $\mathcal{P}_n$ is in bijective correspondence with a subspace of ${\mathbb F}_2^{2n}$. 
We now wish to ascertain what property this subspace has if $S$ is a subgroup generated by commuting elements of $\mathcal{P}_n$.

To this end, we define an alternating form for $u,w \in {\mathbb F}_2^{2n}$,
$$
(u,w)_a=\sum_{j=1}^n (u_{j}w_{j+n}-u_{j+n}w_{j}).
$$

\begin{lemma} \label{taucommutes}
For $M,N \in \mathcal{P}_n/ \{ \pm 1, \pm i\}$,
$$
MN = NM \qquad \mathrm{if}\ \mathrm{and} \ \mathrm{only} \ \mathrm{if} \ \qquad (\tau(M),\tau(N))_a=0.
$$
\end{lemma}

\begin{proof} 
Suppose $u=\tau(M)$ and $w=\tau(N)$. One can check directly that 
$$
u_{j}w_{j+n}-w_{j}u_{j+n}=0
$$ 
if and only if the Pauli matrices in the $j$-th position of $M$ and $N$ commute and is $\pm 1$ otherwise. 

The operators $M$ and $N$ commute if and only if there are an even number of positions where the Pauli-matrices do not commute. This is the case if and only if there are an even number of coordinates $j$ for which 
$$
u_{j}w_{j+n}-w_{j}u_{j+n}=1,
$$ 
a condition equivalent to $(\tau(M),\tau(N))_a=0$. 
\end{proof}

The symplectic weight of a vector $v \in {\mathbb F}_2^{2n}$ is defined as
$$
|\{ i \in \{1,\ldots,n\} \ | \ (v_i,v_{i+n}) \neq (0,0) \}|.
$$

\begin{lemma}
The weight of $M \in \mathcal{P}_n$ is equal to the symplectic weight of $\tau(M)$.
\end{lemma}

\begin{proof}
We have that $n-\mathrm{wt}(M)$ is equal to the number of $\sigma_0$'s in $M$ which is equal to $n$ minus the symplectic weight of $\tau(M)$.
\end{proof}

For a subspace $C \leqslant {\mathbb F}_2^{2n}$, we define $\perp_a$ as
$$
C^{\perp_a}=\{ u \in  {\mathbb F}_2^{2n} \ | \ (u,w)_a=0, \ \mathrm{for} \ \mathrm{all} \ w \in C\}.
$$

\begin{theorem} \label{stabtheorem}
$S$ is a subgroup of $\mathcal{P}_n$ generated by $n-k$ independent mutually commuting elements if and only if $C=\tau(S)$ is a $(n-k)$-dimensional subspace of ${\mathbb F}_2^{2n}$ for which $C \leqslant C^{\perp_a}$. If $k\neq 0$ then the minimum distance of $Q(S)$ is equal to the minimum symplectic weight of the elements of $C^{\perp_a} \setminus C$. If $k=0$ then the minimum distance of $Q(S)$ is equal to the minimum symplectic weight of the non-zero elements of $C=C^{\perp_a}$.
\end{theorem}

\begin{proof}
The fact that $C=\tau(S)$ is contained in $C^{\perp_a}$ follows from Lemma~\ref{tauadd} and Lemma~\ref{taucommutes}.

By Theorem~\ref{stabdistance}, for $k \neq 0$, the minimum distance is equal to the minimum weight of the images of the elements of 
$\mathrm{Centraliser}(S)$ under $\tau$, which are not elements of the image of $S$. 
Since $C=\tau(S)$ and $C^{\perp_a}=\tau(\mathrm{Centraliser}(S))$, the theorem follows for $k \neq 0$. 

For $k=0$, by definition, the minimum distance is equal to the minimum weight of the images of the elements of $S$ under $\tau$, which are the non-zero elements of $C$.
\end{proof}

We can construct a generator matrix $\mathrm{G}(S)$ for $C=\tau(S)$ by taking the $(n-k) \times 2n$ matrix 
whose $i$-th row is $\tau(M_i)$. 

\begin{lemma} \label{fullrankus}
$S$ is a subgroup of $\mathcal{P}_n$ generated by $n-k$ independent elements if and only if the matrix $\mathrm{G}(S)$ has rank $n-k$.
\end{lemma}

\begin{proof}
There is a there is a proper subset $J \subseteq \{1,\ldots,n-k\}$ such that
$$
\sum_{j \in J} \tau(M_j)=0,
$$
if and only if the rank of $\mathrm{G}(S)$ is not $n-k$.
By Lemma~\ref{tauadd}, this is if and only if
$$
\prod_{j \in J} M_j=\one.
$$
\end{proof}

The following table makes for a useful reference.

\begin{small}
\begin{tabular}{@{}ll@{}}
\toprule
 $\mathcal{P}_n$ & the Pauli group, given by $n$-fold tensor products of Pauli matrices \\ & $\sigma_0, \sigma_x, \sigma_y, \sigma_z$ with phases $\{\pm i, \pm 1\}$.\\
 $M_1,\ldots,M_{n-k}$ & the generators, a set of independent elements of $\mathcal{P}_n$ that generate $S$. \\
 $S$ & the stabilizer, an abelian subgroup of $\mathcal{P}_n$. \\
 $Q(S)$ & the quantum code obtained as the joint intersection \\
 &  of the eigenspaces of eigenvalue $1$ of the operators in $S$. \\ 
 $ \mathrm{Centraliser}(S)$ &  the centraliser of $S$ in  $\mathcal{P}_n$  \\
 $C$ & the subspace of ${\mathbb F}_2^{2n}$ obtained from the image of $S$ under $\tau$.\\
 $C^{\perp_a}$ & the subspace of ${\mathbb F}_2^{2n}$ obtained as the image of $\mathrm{Centraliser}(S)$ under $\tau$.\\  
$\mathrm{G}(S)$ & the $(n-k) \times 2n$ generator matrix for $C$ whose $i$-th row is $\tau(M_i)$.\\
\bottomrule
\end{tabular}
\end{small}

\begin{example}\label{5zero3code} (A $[\![ 5,0,3] \!]$ stabilizer code).

Let $S$ be the subgroup of $\mathcal{P}_5$ generated by the following pairwise commuting elements.
$$
\begin{array}{rccccccccc} 
M_1 & =X & Z & I & I & Z  \\ 
M_2 & =Z & X & Z & I & I  \\ 
M_3 & =I & Z & X & Z & I  \\
M_4 & =I & I &  Z & X & Z  \\
M_5 & =Z & I & I &  Z & X  \\
\end{array}
$$
The matrix $\mathrm{G}(S)$ for this code is
$$
\left(
\begin{array}{ccccc|ccccc}
1 & 0 & 0 & 0 & 0 & 0 & 1 & 0 & 0 & 1\\
0 & 1 & 0 & 0 & 0 & 1 & 0 & 1 & 0 & 0 \\
0 & 0 & 1 & 0 & 0 & 0 & 1 & 0 & 1 & 0  \\
0 & 0 & 0 & 1 & 0 & 0 & 0 & 1 & 0 & 1 \\
0 & 0 & 0 & 0 & 1 & 1 & 0 & 0 & 1 & 0  \\
\end{array}
\right)
$$
One can check directly that $(u,v)_a=0$ for any two rows $u,v$ of $\mathrm{G}(S)$. Alternatively, it is enough to observe that $A$ is symmetric and that 
$$
(I \ | A)(\frac{A^t}{I})=A^t+A=A+A=0.
$$

We will prove in Example~\ref{5zero3code2} that the minimum distance of $Q(S)$ is $3$.

Observe that any $n \times n$ symmetric matrix $A$ gives a $[\![n,0,d]\!]$ code, where $G(S)=(I\ | \ A)$. The difficulty lies in choosing $A$ so that the symplectic weight of the code generated by $G$ (and hence $d$)  is large.

\end{example}

\section{The geometry of additive, linear and stabilizer codes} \label{sectiongeometry}

\subsection{Additive and linear codes over a finite field}
We recall that a {\em code} of length $n$ is a subset $C$ of $A^n$, where $A$ is a finite set called the {\em alphabet}. 
An element of $C$ is called a {\em codeword}.

The {\em distance} between any two elements of $A^n$ is the number of coordinates in which they differ. 
The {\em minimum distance} of $C$ is the minimum distance between any two codewords of $C$.

Suppose $A$ is a finite abelian group with identity element $0$. If $u+v \in C$ for all $u,v \in C$ then we say that $C$ is {\em additive}.

The {\em weight} of an element (codeword) $u$ of an additive code is the number of non-zero coordinates that it has.

\begin{lemma} \label{minweight}
If $C$ is an additive code over an alphabet which is a finite abelian group 
then the minimum distance $d$ of $C$ 
is equal to the minimum non-zero weight $w$.
\end{lemma}

\begin{proof}
Summing $u\in C$ enough times will eventually give the $n$-tuple of all zeros, hence $0=(0,\ldots,0) \in C$. Note that this implies $-u \in C$ too.
 
Suppose that~$u$ is a codeword of minimum weight~$w$. Then since $0 \in C$, we have $w \geqslant d$.

Suppose that $u$ and $v$ are two codewords which differ in exactly $d$~coordinates. Then $u-v$ is a codeword in $C$ of weight $d$ and so $d \geqslant w$.
\end{proof}

Suppose that $A={\mathbb F}_q$, the finite field with $q=p^h$ elements, $p$ prime. If $C$ is additive then $\lambda u \in C$ for all $\lambda \in {\mathbb F}_p$, so $C$ is a subspace over ${\mathbb F}_p$. If $C$ has the additional property that $\lambda u \in C$ for all $\lambda$ in ${\mathbb F}_q$ then we say $C$ is {\em linear}. A linear code of length $n$ is a subspace of ${\mathbb F}_q^n$.

We use the notation $(n,K,d)_q$ code to denote a code over an alphabet of size $q$ of length $n$, size $K$ and minimum distance $d$.

The notation $[n,k,d]_q$ code denotes a $k$-dimensional linear code over ${\mathbb F}_q$ of length $n$ and minimum distance $d$.

\subsection{The geometry of linear codes} \label{geomlinear}

We will begin our geometrical study of codes by considering linear codes over ${\mathbb F}_q$.

Let $\mathrm{G}$ be a $k\times n$ matrix. We recall that when $a^t$ is a row vector in ${\mathbb F}_q^k$, the expression $a^t\mathrm{G}$ yields a linear combination of the rows of $\mathrm{G}$. Likewise, when $b$ is a column vector in ${\mathbb F}_q^n$, 
the expression $\mathrm{G}b$ yields a linear combination of the columns of $\mathrm{G}$.

Let $C$ be a $k$-dimensional linear code over ${\mathbb F}_q$ of length $n$, in other words, $C$ is a $k$-dimensional subspace of ${\mathbb F}_q^n$. We describe $C$ by a $k \times n$ matrix $\mathrm{G}$ whose row space is $C$, i.e. the rows of $\mathrm{G}$ are a basis for $C$. Thus, for each $u\in C$, there is an $a^t=(a_1,\ldots,a_k) \in {\mathbb F}_q^k$ such that
$$
u=a^t\mathrm{G}.
$$
In other words, the {\em generator matrix} $\mathrm{G}$ acts as a linear encoding matrix for the message $a$, 
yielding the codeword $u$ ready to be sent over a noisy channel.

The geometry of $C$ is seen by considering the set of columns of the generator matrix~$\mathrm{G}$. Let $\mathcal{X}$ be the set of columns of $\mathrm{G}$, 
so $\mathcal{X}$ is a (possibly multi-)set of $n$ vectors of ${\mathbb F}_q^k$. The codeword $u=a^t\mathrm{G}$ has a zero in its $i$-th coordinate if and only if 
$$
a \cdot z=a_1z_1+\cdots+a_kz_k=0
$$
where $z=(z_1,\ldots,z_k)$ is the $i$-th column of $\mathrm{G}$. This property is unaffected if we replace $z$ by a non-zero scalar multiple of $z$, so it is natural to consider $\mathcal{X}$ as a (possibly multi-)set of $n$ points of $\mathrm{PG}(k-1,q)$, the $(k-1)$-dimensional projective space over ${\mathbb F}_q$.

The projective space $\mathrm{PG}(k-1,q)$ is obtained from the vector space ${\mathbb F}_q^k$ by identifying the vectors which are scalar multiples of each other. In this way, the {\em points} of $\mathrm{PG}(k-1,q)$ are the one-dimensional subspaces of ${\mathbb F}_q^k$ and, more generally, the $(i-1)$-dimensional subspaces of $\mathrm{PG}(k-1,q)$ are the $i$-dimensional subspaces of ${\mathbb F}_q^k$. The {\em lines, planes} and {\em hyperplanes} of $\mathrm{PG}(k-1,q)$ are the $1$-dimensional, $2$-dimensional and co-dimension $1$ subspaces, respectively. Note that in $\mathrm{PG}(k-1,q)$ familiar geometric properties hold. For example, two points are joined by a line; the intersection of two planes in a three-dimensional subspace is a line. If a point $x$ is contained in a subspace $\pi$ we say that $x$ is {\em incident} with $\pi$. If two subspaces $\pi_1$ and $\pi_2$ have an empty intersection (i.e. their corresponding subspaces in ${\mathbb F}_q^k$ intersect in the zero vector), then we say that they are {\em skew}.

A set of points $x_1,\ldots,x_r$ of a projective space are {\em independent} if they span an $(r-1)$-dimensional (projective) subspace. If they are not independent then they are {\em dependent}.

The number of $r$-tuples of linearly independent vectors of ${\mathbb F}_q^k$ is 
$$
(q^k-1)(q^{k-1}-1)\cdots (q^{k-r+1}-1).
$$
Hence, the number of $r$-dimensional subspaces of ${\mathbb F}_q^k$ is
$$
\left[ \begin{array}{c} k \\ r \end{array}\right]_q:=\frac{(q^k-1)(q^{k-1}-1)\cdots (q^{k-r+1}-1)}{(q^r-1)(q^{r-1}-1)\cdots(q-1)}.
$$

Thus, the number of points of $\mathrm{PG}(k-1,q)$ is
$$
\frac{q^k-1}{q-1}=q^{k-1}+q^{k-2}+\cdots+q+1.
$$

There is a natural duality between the points of $\mathrm{PG}(k-1,q)$ and the hyperplanes of $\mathrm{PG}(k-1,q)$. A point $(a_1,\ldots,a_k)$ is mapped to the hyperplane defined as the kernel as the linear form
$$
a_1X_1+\cdots+a_kX_k.
$$
For example, the point $(1,-1,0)$ is mapped to the hyperplane $X_1-X_2=0$,

Thus, the number of hyperplanes of $\mathrm{PG}(k-1,q)$ is also
$$
q^{k-1}+q^{k-2}+\cdots+q+1,
$$
which can be checked directly by calculating $\left[ \begin{array}{c} k \\ k-1 \end{array}\right]_q$.

The number of lines of $\mathrm{PG}(3,q)$ is
$$
\frac{(q^4-1)(q^3-1)}{(q^2-1)(q-1)}=(q^2+1)(q^2+q+1).
$$

The number of points in $\mathrm{PG}(k-1,2)$ is $2^k-1$ and the number of lines of 
 $\mathrm{PG}(k-1,2)$ is $(2^k-1)(2^{k-1}-1)/3$.

\begin{lemma} \label{numberofsubspaces}
The number of $(r-1)$-dimensional subspaces of $\mathrm{PG}(k-1,q)$ containing a fixed $(s-1)$-dimensional subspace is
$$
\left[ \begin{array}{c} k-s \\ r -s \end{array}\right]_q.
$$
\end{lemma}

\begin{proof}
For any $s$-dimensional subspace $U$ of ${\mathbb F}_q^k$, the quotient space ${\mathbb F}_q^k/U$ is a $(k-s)$-dimensional vector space. An $r$-dimensional subspace containing $U$ is a $(r-s)$-dimensional subspace in the quotient space. Thus, the lemma holds, taking into account the dimension shift when considering the projective space.
\end{proof}

The following theorem explains what the minimum distance $d$ of a linear code implies for the set of points $\mathcal{X}$.

\begin{theorem} \label{lineargeomthm}
An $[n,k,d]$ linear code over ${\mathbb F}_q$ is equivalent to a (possibly multi-)set of points $\mathcal{X}$ in $\mathrm{PG}(k-1,q)$ in which every hyperplane of $\mathrm{PG}(k-1,q)$ contains at most $n-d$ points of $\mathcal{X}$ and some hyperplane contains exactly $n-d$ points of $\mathcal{X}$.
\end{theorem}

\begin{proof}
Let $\mathrm{G}$ be a $k \times n$ matrix whose row space is a $[n,k,d]$ linear code $C$. Let $\mathcal{X}$ be the set of columns of $\mathrm{G}$ viewed as points of $\mathrm{PG}(k-1,q)$. 

Recall that the codeword $u=a^t\mathrm{G}$ has a zero in its $i$-th coordinate if and only if 
$$
a \cdot z=a_1z_1+\cdots+a_kz_k=0
$$
where $z=(z_1,\ldots,z_k)$ is the $i$-th column of $\mathrm{G}$. 

The kernel of the linear form
$$
a_1X_1+\cdots+a_kX_k
$$
defines a hyperplane $\pi_a$ of $\mathrm{PG}(k-1,q)$.  The codeword $u=a^t\mathrm{G}$ has weight $w$ if and only if $u$ has exactly $n-w$ zero coordinates. This is the case if and only if $\pi_a$ is incident with $n-w$ points of $\mathcal{X}$. 

By Lemma~\ref{minweight}, the minimum distance of a linear code is equal to its minimum weight. Hence, the maximum number of points of $\mathcal{X}$ on a hyperplane of $\mathrm{PG}(k-1,q)$ is $n-d$, where $d$ is the minimum distance of $C$.
\end{proof}

\subsection{The geometry of additive codes} \label{additivecodes}

An additive code $C$ over ${\mathbb F}_{q}$ is linear over ${\mathbb F}_p$, where $q=p^h$ for some prime $p$. Therefore, $|C|=p^r$ for some $r$. The following theorem is the additive version of Theorem~\ref{lineargeomthm}; the set of points $\mathcal{X}$ is replaced by a set of subspaces.
 
\begin{theorem} \label{additivegeom}
An $(n,p^{r},d)$ additive code over ${\mathbb F}_q$ with $q = p^h$ is equivalent to a (possibly multi-)set $\mathcal{X}$ of $\leqslant (h-1)$-dimensional subspaces  in $\mathrm{PG}(r-1,p)$ in which every hyperplane of $\mathrm{PG}(r-1,p)$ contains at most $n-d$ subspaces of $\mathcal{X}$ and some hyperplane contains exactly $n-d$ subspaces of $\mathcal{X}$.
\end{theorem}

\begin{proof}
Let $\mathrm{G}$ be a $r \times n$ matrix which is a basis for $C$ over ${\mathbb F}_p$. As in the case of linear codes, we consider the (possibly multi-)set $\mathcal{X}$ of columns of $\mathrm{G}$. However, we shouldn't consider the elements of $\mathcal{X}$ as points of $\mathrm{PG}(r-1,q)$, since we obtain $C$ from $\mathrm{G}$ by taking the row span over ${\mathbb F}_p$ and not over ${\mathbb F}_q$. Thus, we consider the elements of $\mathcal{X}$ as subspaces of $\mathrm{PG}(r-1,p)$. Suppose that $e \in {\mathbb F}_q$, is such that $\{1,e,e^2,\ldots,e^{{h-1}}\}$ is a basis for ${\mathbb F}_q$ over ${\mathbb F}_p$. Then, up to scalar factor, we can write $x \in \mathcal{X}$ as
$$
\sum_{j=0}^{h-1} e^{j} x_j,
$$
where $x_j \in {\mathbb F}_p^r$. We associate $x$ with the subspace spanned by $x_0,\ldots,x_{h-1}$ in $\mathrm{PG}(r-1,p)$, which we denote by $\ell_x$. The subspace $\ell_x$ has dimension at most $h-1$. 

Suppose that $x$ is the $i$-th column of $\mathrm{G}$, so $x \in \mathcal{X}$. The non-zero codeword $u=a^t\mathrm{G}$, where $a \in {\mathbb F}_p^r$, has a zero in its $i$-th coordinate if and only if the hyperplane of $\mathrm{PG}(r-1,p)$, which is the kernel of linear form
$$
a_1X_1+\cdots+a_rX_r,
$$
contains the subspace $\ell_x$.
\end{proof}

Observe that a linear code over ${\mathbb F}_q$ necessarily has size $q^k$, so if we wish to obtain an additive code with the same parameters as a linear code, then $r=kh$ for some $k$.

\subsection{The geometry of qubit quantum codes}

For the moment, we restrict to the case $q=2$ and consider the geometrical consequences of Theorem~\ref{stabtheorem}, which describes the connection between stabilizer codes and binary linear codes.

A qubit stabilizer code $Q(S)$ is equivalent to a binary linear code $C=\tau(S)$ of length $2n$ which is contained in its alternating dual $C^{\perp_a}$. According to Theorem~\ref{stabtheorem}, the minimum distance of $Q(S)$ is the minimum symplectic weight of $C^{\perp_a} \backslash C$.
 
Consider once again the Shor code from Example~\ref{shorcode1}.

\begin{example} \label{Gmatrixshor} (Shor code)
Applying the map $\tau$ to the elements in Example~\ref{shorcode1} we have that $C=\tau(S)$ is the ${\mathbb F}_2$ row span of the matrix
$$
G(S)=\left( \begin{array}{ccccccccc|ccccccccc} 
0 & 0 & 0 & 0 & 0 & 0 & 0 & 0 & 0 & 1 & 1 & 0 & 0 & 0 & 0 & 0 & 0 & 0 \\
0 & 0 & 0 & 0 & 0 & 0 & 0 & 0 & 0 & 0 & 1 & 1 & 0 & 0 & 0 & 0 & 0 & 0\\
0 & 0 & 0 & 0 & 0 & 0 & 0 & 0 & 0 & 0 & 0 & 0 & 1 & 1 & 0 & 0 & 0 & 0 \\
0 & 0 & 0 & 0 & 0 & 0 & 0 & 0 & 0 & 0 & 0 & 0 & 0 & 1 & 1 & 0 & 0 & 0\\
0 & 0 & 0 & 0 & 0 & 0 & 0 & 0 & 0  & 0 & 0 & 0 & 0 & 0 & 0 & 1 & 1 & 0\\
0 & 0 & 0 & 0 & 0 & 0 & 0 & 0 & 0 & 0 & 0 & 0 & 0 & 0 & 0 & 0 & 1 & 1\\
1 & 1 & 1 & 1 & 1 & 1 & 0 & 0 & 0 & 0 & 0 & 0 & 0 & 0 & 0 & 0 & 0 & 0 \\
0 & 0 & 0 & 1 & 1 & 1 & 1 & 1 & 1 & 0 & 0 & 0 & 0 & 0 & 0 & 0 & 0 & 0 \\
\end{array} \right).
$$
Since there are two columns which are linearly dependent, there are elements of $C^{\perp_a
}$ of symplectic weight two; these are images under $\tau$ of Pauli operators of $\mathrm{Centraliser}(S)$ of weight two. 

To see this, recall that the alternating form is defined as
$$
(u,w)_a=\sum_{j=1}^n (u_{j}w_{j+n}-u_{j+n}w_{j}),
$$ 
so the dependency of the first two columns implies that
$$
(0,0,0,0,0,0,0,0,0 \ | \ 1,1,0,0,0,0,0,0,0)
$$
is an element of $C^{\perp_a}$. However, this element is an element of $C$, since it's the first row of the matrix. Recall that the minimum distance is equal to the minimum symplectic weight of $C^{\perp_a} \setminus C$. Therefore, although $C^{\perp_a}$ contains elements of symplectic weight $2$, the minimum symplectic weight of $C^{\perp_a} \setminus C$ is in fact $3$. We will prove this in Example~\ref{shorcode3}. 
\end{example}

Given a subgroup $S$, generated by $n-k$ commuting elements $M_1,\ldots,M_{n-k}$ of $\mathcal{P}_n$, we obtain a set $\mathcal{X}$ of $n$ lines or possibly points in $\mathrm{PG}(n-k-1,2)$ in the following way. For each $i \in \{1,\ldots,n\}$, we get a line (or a point) by considering the span of the $i$-th and $(i+n)$-th column of the generator matrix $\mathrm{G}(S)$. Vice versa, given a set of $n$ lines in $\mathrm{PG}(n-k-1,2)$, we construct a $(n-k)\times 2n$ matrix, from which we obtain $M_1,\ldots,M_{n-k}$ by applying $\tau^{-1}$ to the rows of the matrix.

On first sight it may seem that there is a certain amount of freedom when we reconstruct the code from a given quantum set of lines. Each line is incident with three points and we can choose which pair of points on the line to use to construct the $i$-th and the $(i+n)$-th column of $\mathrm{G}$. This choice is equivalent to invoking a permutation of $\{ \sigma_x,\sigma_y,\sigma_z\}$ on the $i$-th position of each of the $M_1,\ldots,M_{n-k}$. This does not affect the property that these elements pairwise commute, so we define all quantum codes that can be obtained from each other in this way to be equivalent.

For example, in Example~\ref{5zero3code}, invoking the permutation $\sigma$ which takes $X\rightarrow Z \rightarrow Y \rightarrow X$ on the $M_i$ in the first, second and fourth positions gives
$$
\begin{array}{rccccccccc} 
\sigma(M_1) & =Z & Y & I & I & Z  \\ 
\sigma(M_2) & =Y & Z & Z & I & I  \\ 
\sigma(M_3) & =I & Y & X & Y & I  \\
\sigma(M_4) & =I & I &  Z & Z & Z  \\
\sigma(M_5) & =Y & I & I &  Y & X  \\
\end{array}.
$$
The matrix whose $i$-th row is $\tau(M_i)$ is
$$
\left(
\begin{array}{ccccc|ccccc}
0 & 1 & 0 & 0 & 0 & 1 & 1 & 0 & 0 & 1\\
1 & 0 & 0 & 0 & 0 & 1 & 1 & 1 & 0 & 0 \\
0 & 1 & 1 & 1 & 0 & 0 & 1 & 0 & 1 & 0  \\
0 & 0 & 0 & 0 & 0 & 0 & 0 & 1 & 1 & 1 \\
1 & 0 & 0 & 1 & 1 & 1 & 0 & 0 & 1 & 0  \\
\end{array}
\right)
$$
Comparing this to the matrix 
$$
\mathrm{G}(S)=
\left(
\begin{array}{ccccc|ccccc}
1 & 0 & 0 & 0 & 0 & 0 & 1 & 0 & 0 & 1\\
0 & 1 & 0 & 0 & 0 & 1 & 0 & 1 & 0 & 0 \\
0 & 0 & 1 & 0 & 0 & 0 & 1 & 0 & 1 & 0  \\
0 & 0 & 0 & 1 & 0 & 0 & 0 & 1 & 0 & 1 \\
0 & 0 & 0 & 0 & 1 & 1 & 0 & 0 & 1 & 0  \\
\end{array}
\right)
$$
from Example~\ref{5zero3code}, we see that the set of lines $\mathcal{X}$ remains unchanged.

There is also a choice between the scalar factor of $M$ when we apply $\tau^{-1}$ to a row of the matrix $\mathrm{G}$. We will always assume that this factor to be $1$. However, changing the sign of some of the generators of a subgroup $S$ can be useful, as we shall see in Section~\ref{unionstabilizers}.

\begin{lemma} \label{theyrelines}
The span of the $i$-th and $(i+n)$-th column of the generator matrix $\mathrm{G}(S)$ is a line of $\mathrm{PG}(n-k-1,2)$ for all $i=1,\ldots,n$ if and only if the minimum non-zero weight of $\mathrm{Centraliser}(S)$ is at least two.
\end{lemma}

\begin{proof}

We fail to obtain a line of $\mathrm{PG}(n-k-1,2)$ if and only if the $i$-th and $(i+n)$-th column of the matrix $\mathrm{G}(S)$ are either the same non-zero vector or one or both of them is the zero vector. This implies that in the $i$-th position of all the Pauli operators in $S$, there is either $\sigma_0$ or a fixed element $\sigma \in \{\sigma_x,\sigma_y,\sigma_z\}$. This occurs if and only if there is an element of $\mathrm{Centraliser}(S)$ of weight $1$. 
\end{proof}

If $Q(S)$ is pure then the condition that the minimum non-zero weight of $\mathrm{Centraliser}(S)$ is at least $2$ can be replaced by the condition that the minimum distance of $Q(S)$ is at least $2$. However, this does not need to hold for impure codes. 
Indeed it could be that there are elements of $\mathrm{Centraliser}(S) \cap S$ of weight one. 
Yet, if the stabilizer of a $[\![n,k,d]\!]$ code $Q(S)$ contains an element of weight one, then it is easy to see that one can construct a $[\![n-1,k,d]\!]$ stabilizer code by deleting that position.

We would like to give a geometrical interpretation of the fact that the code $C=\tau(S)$ is contained in $C^{{\perp}_a}$. 

Recall that we say two subspaces of $\mathrm{PG}(k-1,q)$ are {\em skew} if they do not intersect.

\begin{theorem} \label{quantumlines}
The following are equivalent.
\begin{enumerate}
\item
There is a $[\![ n,k,d ]\!]$ stabilizer code $Q(S)$, where $S$ is a subgroup generated by $n-k$ independent commuting elements of $\mathcal{P}_n$ and whose centraliser contains no element of weight one.
\item
There is a set of $n$ lines $\mathcal{X}$ spanning $\mathrm{PG}(n-k-1,2)$ with the property that every co-dimension $2$ subspace is skew to an even number of the number of lines of~$\mathcal{X}$.
\end{enumerate}
\end{theorem} 

\begin{proof}
($1 \Rightarrow 2$)

Let $C=\tau(S)$ and let $\mathrm{G}=\mathrm{G}(S)$ be a $(n-k) \times 2n$ generator matrix for $C$.

From Lemma~\ref{fullrankus}, the matrix~$G$ has rank $n-k$. Thus, its columns span~$\mathrm{PG(n-k-1,2)}$.

Let $\mathcal{X}$ be the set of $n$ lines obtained for $i=1,\ldots,n$ as the span of the $i$-th and $(i+n)$-th column of $\mathrm{G}(S)$.

Let $u,w \in C$, so $u=(a_1,\ldots,a_{n-k})\mathrm{G}$ and $w=(b_1,\ldots,b_{n-k})\mathrm{G}$ for some $a=(a_1,\ldots,a_{n-k}) \in {\mathbb F}_2^{n-k}$ and $b=(b_1,\ldots,b_{n-k}) \in {\mathbb F}_2^{n-k}$.

One has $C \subseteq C^{\perp_a}$ if and only if 
$$
(u,w)_a=\sum_{j=1}^n (u_jw_{n+j}-w_ju_{n+j})=0,
$$
for all $u,w \in C$. 

We want to deduce the geometrical meaning of $(u,w)_a=0$. 

Consider a single term in the sum first. Let $x$ and $y$ be the $j$-th and the $(n+j)$-th column of $\mathrm{G}$ respectively. 
Then
$$
u_j w_{n+j} - u_{n+j} w_j = (a \cdot x)(b \cdot y) - (a \cdot y)(b \cdot x).
$$
The right-hand side is zero if and only if the matrix
$$
\left(\begin{array}{cc}
a\cdot  x & a \cdot y\\
b \cdot x & b\cdot  y
\end{array} \right)
$$
has zero determinant, i.e. it has rank $1$. 

This is if and only if there exists $\lambda,\mu \in {\mathbb F}_2$ such that 
$$
a \cdot( \lambda x+\mu y)=0
$$
and
$$
b \cdot( \lambda x+\mu y)=0.
$$
Recall that we define $\pi_a$ as the hyperplane which is the kernel of the linear form 
$$
a \cdot X=a_1X_1+\cdots+ a_{n-k}X_{n-k}.
$$

\begin{figure}[tbp]
\begin{center}
\includegraphics[width = 0.7\textwidth]{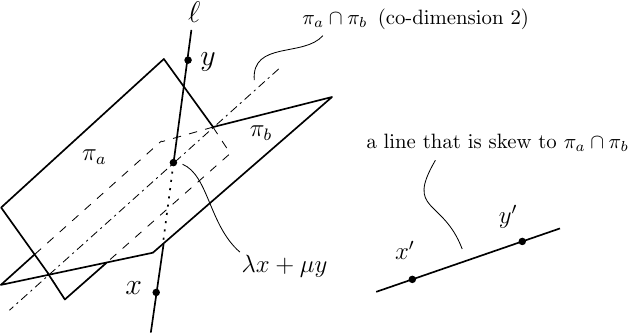} 
 \caption{\label{fig:intersect}}
A point $\lambda x + \mu y$ on the intersection of the hyperplanes $\pi_a$ and $\pi_b$.
\end{center}
\end{figure}

We can thus rewrite the above conditions as the requirement 
that the point $\lambda x + \mu y$ is contained in both $\pi_a$ and $\pi_b$.
In other words, there is a point on the line $\ell$, spanned by $x$ and $y$, 
which is incident with the intersection of the two hyperplanes $\pi_a$ and $\pi_b$.

Returning to the condition $(u,v)_a = 0$, 
we must therefore get an even number of ones in the sum 
$$ 
\sum_{j=1}^n (u_j w_{n+j}-u_{n+j} w_j )\,.
$$
All lines of $\mathcal{X}$ that are skew to
$\pi_a \cap \pi_b = \ker(a \cdot X) \,\cap\, \ker (b \,\cdot X)$ 
contribute; for any given $a$ and $b$ there must in total be an even number of such lines.

We note that {\em every} co-dimension $2$ subspace of $\mathrm{PG}(n-k-1,2)$ 
can be realised in this way (as the intersection of some $a \cdot X=0$ and $b \cdot X=0$). 
This proves the forward implication.

\noindent ($1 \Leftarrow 2$) 

Let $\mathcal{X}$ be a set of lines spanning $\mathrm{PG}(n-k-1,2)$ with the property that every co-dimension $2$ subspace of $\mathrm{PG}(n-k-1,2)$ is skew to an even number of lines of $\mathcal{X}$. Let $\mathrm{G}$ be the matrix whose $i$-th and $(i+n)$-th column are points which span the $i$-th line of $\mathcal{X}$. Let $C$ be the code generated by $\mathrm{G}$. Since $\mathcal{X}$ spans $\mathrm{PG}(n-k-1,2)$, the code $C$ is $(n-k)$-dimensional.
As we proved in the forward implication, the property that every co-dimension 2 subspace is skew to an even number of lines of $\mathcal{X}$ implies that for any two codewords $u$ and $v$ of $C$, $(u,v)_a=0$ holds. By Lemma~\ref{taucommutes}, the image under $\tau^{-1}$ of $C$ is an abelian subgroup $S$ of $\mathcal{P}_n$ and by Lemma~\ref{fullrankus}, it is generated by $n-k$ pairwise commuting elements of $\mathcal{P}_n$.
\end{proof}

Let $\mathcal{X}$ be a set of lines and let $\Theta(\mathcal{X})$ be the space spanned by the lines of $\mathcal{X}$.

We say that $\mathcal{X}$ is a {\em quantum set of lines} if it has the property that every co-dimension~$2$ subspace of $\Theta(\mathcal{X})$ is skew to an even number of lines of $\mathcal{X}$. 
To deduce the minimum distance of the corresponding stabilizer code, 
we introduce the parameter~$d(\mathcal{X})$. 

Recall that $r$ points are {\em independent} if they span an $(r-1)$-dimensional subspace; 
they are {\em dependent} otherwise.

Consider first the case in which $\dim \Theta(\mathcal{X}) \neq |\mathcal{X}|-1$. By Theorem~\ref{quantumlines}, $\mathcal{X}$ will give a quantum $[\![n,k,d]\!]$ code with $k \neq 0$.
We define the parameter $d(\mathcal{X})$ as the minimum number of dependent points 
that can be found on distinct lines of $\mathcal{X}$; 
not including the dependencies for which there is a hyperplane of $\Theta(\mathcal{X})$ which both
\begin{enumerate}
\item[a)]
contains all the lines of $\mathcal{X}$ which do not contain the dependent points\,,
\item[b)]
contains all the dependent points.\footnote{In the original definition of Glynn et al \cite{GGMG}, the condition b) does not appear.} 
\end{enumerate}

Thus, $d(\mathcal{X})=r$, where $r$ is minimal such that there exists a set of dependent points $\{x_1,\ldots,x_r\}$, where each $x_i$ is incident with a line $\ell_i \in \mathcal{X}$ and the lines $\ell_1,\ldots,\ell_r$ are distinct, but for which there is no hyperplane containing the lines $\mathcal X \setminus \{\ell_1,\ldots,\ell_r\}$  and the points $\{x_1,\ldots,x_r\}$.

In the case in which $\dim \Theta(\mathcal{X}) =|\mathcal{X}|-1$, Theorem~\ref{quantumlines} implies that $\mathcal{X}$ will give a quantum $[\![n,k,d]\!]$ code with $k=0$. We define the parameter $d(\mathcal{X})$ as the minimum $d$ for which there is a hyperplane of $\Theta(\mathcal{X})$ containing $|\mathcal{X}|-d$ lines of $\mathcal{X}$. Equivalently. it is the minimum number of dependent points that can be found on distinct lines of $\mathcal{X}$. This definition and the equivalence will be justified in the proof of Theorem~\ref{thm:stab_geo}.

From now on we assume that the centraliser of the stabilizer $S$ contains no elements of weight one.  
By Lemma~\ref{theyrelines}, this assumption guarantees that there is a quantum set of lines associated with the stabilizer code. 
As mentioned before, this is equivalent to assuming that the minimum distance is at least $2$ in the case of pure codes.

\begin{theorem}\label{thm:stab_geo}
There is a $[\![n,k,d]\!]$ stabilizer code if and only if there is a quantum set of lines $\mathcal{X}$ for which $d(\mathcal{X})=d$ and $\Theta(\mathcal X)=\mathrm{PG}(n-k-1,2)$.
\end{theorem}

\begin{proof}
We only have to prove the part about the minimum distance since Theorem~\ref{quantumlines} covers the rest.

($\Rightarrow$) Let $Q(S)$ be a $[\![n,k,d]\!]$ stabilizer code given by some stabilizer $S$. Let $C=\tau(S)$.

As in the proof of Theorem ~\ref{quantumlines}, let $\mathrm{G}=\mathrm{G}(S)$ be the $(n-k) \times 2n$ generator 
matrix with entries from ${\mathbb F}_2$ whose row space forms the code $C$. Define a set of lines 
$$
\mathcal{X}=\{ \ell_j \ | \ j=1,\ldots,n \},
$$ 
where $\ell_j$ is the line that corresponds to 
the span of the $j$-th and $(j+n)$-th column of $\mathrm{G}$.

Consider the case $k \neq 0$. 

By Theorem~\ref{stabtheorem}, the parameter $d$ is the minimum symplectic weight of $C^{\perp_a} \setminus C$.

Suppose now that $v \in C^{\perp_a}$ has symplectic weight $w$ and let $W$ denote the set of positions that contribute to the weight,
$$
W=\{ j \in \{1,\ldots,n\} \ | \ (v_j,v_{n+j}) \neq (0,0) \}.
$$
Clearly, $|W|=w$. 

Denote by $x_j$ the $j$-th column of $\mathrm{G}$.
Since $v=(v_1,\ldots,v_{2n})$ is in $C^{\perp_a}$, one has
\begin{equation}\label{eq:sum}
\sum_{j \in W} (v_{n+j} x_j-x_{n+j}v_j)=0. 
\end{equation} 

Each summand corresponds to some point of $\ell_{j}$.
Thus, there are $w = |W|$ points on distinct lines $\{ \ell_{j} \ | \ j \in W\}$ which are dependent.

However, since the minimum distance $d$ is the minimum symplectic weight of $C^{\perp_a} \setminus C$, we have to disregard this dependency if $v \in C$.

A vector $v$ is in $C$ if and only if $v=a\mathrm{G}$ for some $a \in {\mathbb F}_2^{n-k}$. 
As a consequence, $v_j=a\cdot x_j$ for all $j=1,\ldots,2n$. 

First, consider those positions $j$ of $v$ that {\em do not} contribute to its symplectic weight, 
that is, $j\notin W$. 
For each $j\notin W$, one has that $v_j=a\cdot x_j=0$ and $v_{n+j}=a\cdot x_{n+j}=0$ if and only if the line $l_j$ is contained in the hyperplane $\pi_a$ described by $a\cdot X = 0$.
So the lines of $\{ \ell_j \ | \ j \in \{1,\ldots,n\} \setminus W\}$ are contained in $\pi_a$. 

Second, consider those positions $j$ of $v$ that contribute to its symplectic weight, $j\in W$. Then
$$
a \cdot (v_{n+j} x_j-x_{n+j}v_j)= v_{n+j}(a \cdot x_j)-(a \cdot x_{n+j})v_j= v_{n+j} v_j-v_{n+j}v_j=0,
$$
since $v_j=a\cdot x_j$ and $v_{n+j}=a\cdot x_{n+j}$. Hence, the dependent points are also contained in the hyperplane $a \cdot X=0$. 

This exactly coincides with our definition of $d(\mathcal X)$.

Now, consider the case $k=0$. 

By Theorem~\ref{stabtheorem}, the parameter $d$ is the minimum non-zero symplectic weight of $C$.

Let $v \in C$ be of minimum non-zero symplectic weight. Since $v \in C$, $v=a\mathrm{G}$ for some $a \in {\mathbb F}_2^{n-k}$. Thus, $v_j=a\cdot x_j$ for all $j=1,\ldots,2n$. 

Let $W$ denote the set of positions that contribute to the symplectic weight of $v$, i.e.
$$
W=\{ j \in \{1,\ldots,n\} \ | \ (v_j,v_{n+j}) \neq (0,0) \}.
$$
Then, for $j \in W$, $a\cdot x_j=a\cdot x_{n+j}=0$ which is equivalent to the line $\ell_j \in \mathcal{X}$ being contained in the hyperplane $a\cdot X=0$. Therefore, there is a hyperplane of $\Theta(\mathcal{X})$ containing $|\mathcal{X}|-d$ lines of $\mathcal{X}$ which coincides with our definition of $d(\mathcal X)$ in this case.

Alternatively, since $C=C^{\perp_a}$, the parameter $d$ is the minimum non-zero symplectic weight of $C^{\perp_a}$. As in the case $k\neq 0$, a vector $v=(v_1,\ldots,v_{2n}) \in C^{\perp_a}$ of symplectic weight $d$, will give a dependency of $d$ points of $\mathcal{X}$, which coincides with our alternative definition of $d(\mathcal X)$ in this case.

($\Leftarrow$) Vice-versa, suppose that $\mathcal{X}$ is a quantum set of lines for which $d(\mathcal{X})=d$ and $\Theta(\mathcal X)=\mathrm{PG}(n-k-1,2)$. 

Let $\mathrm{G}=\mathrm{G}(S)$ be the $(n-k) \times 2n$ generator 
matrix for a code $C$, whose $i$-th and $(i+n)$-th column span the $i$-th line of $\mathcal X$. Let $S=\tau^{-1}(C)$ and let $Q(S)$ be the stabiliser code. By Theorem~\ref{quantumlines} and the fact that $\Theta(\mathcal X)=\mathrm{PG}(n-k-1,2)$, $Q(S)$ is a $[\![n,k,d]\!]$ stabilizer code for some $d$. The fact that $d=d(\mathcal X)$ follows from the same arguments as in the forward implication, observing that if
$$
a \cdot (v_{n+j} x_j-x_{n+j}v_j)=0
$$
then 
$$
v_{n+j}(a \cdot x_j)-(a \cdot x_{n+j})v_j=0
$$
 which implies  $v_j=a\cdot x_j$ and $v_{n+j}=a\cdot x_{n+j}$, assuming $(a\cdot x_j,a\cdot x_{n+j}) \neq (0,0)$. This is precisely the assumption that $\ell_j$ is not contained in the hyperplane $\pi_a$.
\end{proof}

\begin{example} (Shor code) \label{shorcode3}
As we saw in Example~\ref{Gmatrixshor}, the Shor code has the generator matrix
$$
\mathrm{G}(S)=\left( \begin{array}{ccccccccc|ccccccccc} 
0 & 0 & 0 & 0 & 0 & 0 & 0 & 0 & 0 & 1 & 1 & 0 & 0 & 0 & 0 & 0 & 0 & 0 \\
0 & 0 & 0 & 0 & 0 & 0 & 0 & 0 & 0 & 0 & 1 & 1 & 0 & 0 & 0 & 0 & 0 & 0\\
0 & 0 & 0 & 0 & 0 & 0 & 0 & 0 & 0 & 0 & 0 & 0 & 1 & 1 & 0 & 0 & 0 & 0 \\
0 & 0 & 0 & 0 & 0 & 0 & 0 & 0 & 0 & 0 & 0 & 0 & 0 & 1 & 1 & 0 & 0 & 0\\
0 & 0 & 0 & 0 & 0 & 0 & 0 & 0 & 0  & 0 & 0 & 0 & 0 & 0 & 0 & 1 & 1 & 0\\
0 & 0 & 0 & 0 & 0 & 0 & 0 & 0 & 0 & 0 & 0 & 0 & 0 & 0 & 0 & 0 & 1 & 1\\
1 & 1 & 1 & 1 & 1 & 1 & 0 & 0 & 0 & 0 & 0 & 0 & 0 & 0 & 0 & 0 & 0 & 0 \\
0 & 0 & 0 & 1 & 1 & 1 & 1 & 1 & 1 & 0 & 0 & 0 & 0 & 0 & 0 & 0 & 0 & 0 \\
\end{array} \right).
$$

Let $e_i$ denote the $i$-th vector in the canonical basis of ${\mathbb F}_2^8$.

The quantum set of lines $\mathcal{X}$ is
$$
\{ 
\langle e_1, e_7\rangle, \langle e_1+e_2,e_7 \rangle, \langle e_2, e_7\rangle,
\langle e_3, e_7+e_8\rangle,
$$
$$
 \langle e_3+e_4,e_7+e_8 \rangle, \langle e_4, e_7+e_8\rangle,
\langle e_5, e_8\rangle, \langle e_5+e_6,e_8 \rangle, \langle e_6, e_8\rangle
\}.
$$
which is drawn in Figure~\ref{shorcodefig}. Here, $\langle e_i, e_j\rangle$ denotes the
line spanned by points $e_i$ and $e_j$.

Note that the point $e_7$ is on the two lines
$\langle e_1, e_7\rangle$ and $\langle e_1+e_2,e_7 \rangle$, 
and thus $e_7$ is ``dependent with itself''. 
So at first sight it seems that $d(\mathcal{X})=2$. However, the remaining seven lines span a six dimensional subspace since the two planes $\langle e_3,e_4,e_7+e_8\rangle$ and $\langle e_5,e_6,e_8\rangle$  span a five dimensional subspace,
while the line $\langle e_2, e_7\rangle$ extends this to a six dimensional subspace that also contains the point $e_7$ (i.e. contains all dependent points). Following Theorem~\ref{thm:stab_geo}, we do not count this dependency and conclude that $d(\mathcal{X}) \geqslant 3$. The dependency of $e_7$ with itself implies that the Shor code is impure. The dependent points $\{e_1, e_2,e_1+e_2\}$ imply that $d(\mathcal{X})=3$. Although the six lines not containing these points are contained in a hyperplane, there is no hyperplane containing the six lines and the dependent points, thus we do not disregard this dependency. Thus, we see that condition b) is essential in the definition of $d(\mathcal X)$. 
\end{example}

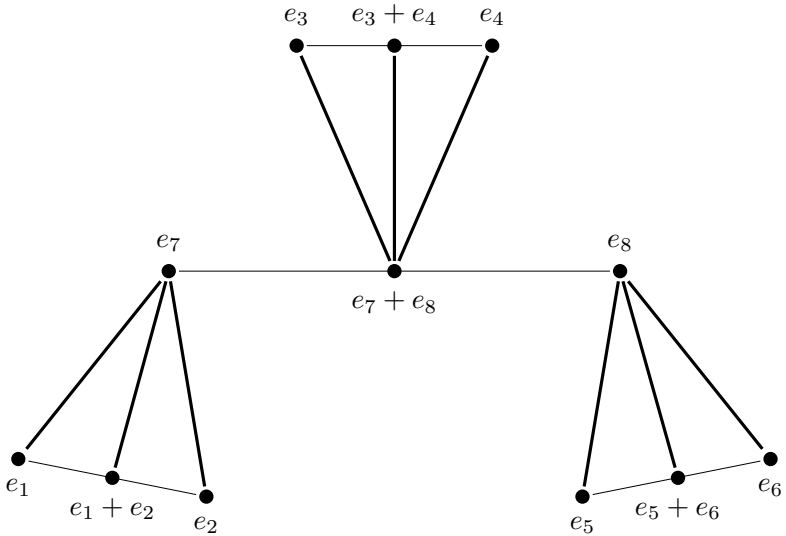
\begin{figure}
\begin{center}
\begin{tikzpicture}
       \node (e8) at (0,0) [label=below:$e_7+e_8$]{};
       \node (e7) at (-3,0) [label=above:$e_7$]{};
       \node (e9) at (3,0) [label=above:$e_8$]{}; 
       \node (e1) at (-5,-2.5) [label=below:$e_1$]{};
       \node (e12) at (-3.75,-2.75) [label=below:$e_1+e_2$]{};
       \node (e2) at (-2.5,-3) [label=below:$e_2$]{}; 
       \node (e3) at (-1.3,3) [label=above:$e_3$]{}; 
       \node (e4) at (1.3,3) [label=above:$e_4$]{};
       \node (e34) at (0,3) [label=above:$e_3+e_4$]{};  
       \node (e5) at (2.5,-3) [label=below:$e_5$]{};
       \node (e6) at (5,-2.5) [label=below:$e_6$]{};
       \node (e56) at (3.77,-2.75) [label=below:$e_5+e_6$]{};  
        \draw[very thick] (e7) -- (e1);
        \draw[very thick] (e7) -- (e2);  
        \draw[very thick] (e7) -- (e12);  
	\draw[very thick] (e8) -- (e3);
       	\draw[very thick] (e8) -- (e4);  
        \draw[very thick] (e8) -- (e34);
        \draw[very thick] (e9) -- (e5);
       	\draw[very thick] (e9) -- (e6);  
        \draw[very thick] (e9) -- (e56);  
        \draw[ultra thin] (e7) -- (e9); 
        \draw[ultra thin] (e1) -- (e2); 
        \draw[ultra thin] (e3) -- (e4);
        \draw[ultra thin] (e5) -- (e6);   
       \fill (e1) circle (2.7pt);
       \fill (e2) circle (2.7pt);
       \fill (e12) circle (2.7pt);
       \fill (e34) circle (2.7pt);
       \fill (e56) circle (2.7pt);
       \fill (e3) circle (2.7pt);
       \fill (e4) circle (2.7pt);
       \fill (e5) circle (2.7pt);
       \fill (e6) circle (2.7pt);
       \fill (e7) circle (2.7pt);
       \fill (e8) circle (2.7pt);
       \fill (e9) circle (2.7pt);

     \end{tikzpicture} 
    
\end{center}
\caption{The set of nine (thick) lines describing the geometry of the Shor code.}  \label{shorcodefig}
\end{figure}

Let us generalize one feature of the Shor code further:
a {\em planar pencil of lines} in a projective space is a set of lines which are all contained in some 
plane and are all the lines incident with a point in that plane. 
As illustrated in Figure~\ref{shorcodefig}, the Shor code is the union of three planar pencils.

Observe that a planar pencil of lines is itself a quantum set of lines. 
Our aim is to show that a quantum set of lines is nothing more than the union modulo two of planar pencils of lines. We first prove a few lemmas.

\begin{lemma} \label{qlunion}
The union modulo two of two quantum sets of lines is a quantum set of lines.
\end{lemma}

\begin{proof}
Let $\mathcal{X}$ and $\mathcal{Y}$ be two quantum sets of lines. Recall that $\Theta(\XX)$, $\Theta(\YY)$, and $\Theta({\XX \cup \YY})$ are the spaces spanned by $\XX$, $\YY$, and both sets of lines respectively. A co-dimension $2$ subspace $\pi$ intersects $\Theta(\mathcal{X})$ in either a co-dimension $2$ subspace, in a hyperplane, or in $\Theta(\mathcal{X})$. In the first case it is skew to an even number of the lines of $\mathcal{X}$; 
in the latter two cases it is skew to none (which is even).  

Let $\mathcal{\overbar X}$ be the subset of $\mathcal{X}$ of lines skew to $\pi$.
Likewise, let $\mathcal{\overbar{Y}}$ be the subset of $\mathcal{Y}$ of lines skew to $\pi$. 
Then $\pi$ is skew to
$
|\mathcal{\overbar X}|+|\mathcal{\overbar{Y}}|-2|\mathcal{\overbar X} \cap \mathcal{\overbar{Y}}|
$
lines of the union modulo two of $\mathcal{X}$ and $\mathcal{Y}$.

Since both $|\mathcal{\overbar X}|$ and $|\mathcal{\overbar{Y}}|$ are even,
every co-dimension $2$ subspace is skew to an even number of lines of $\XX \cup \YY$. This proves the lemma.
\end{proof}

An {\em $r$-sputnik} is a set of $(r+1)$ concurrent lines (they are all incident with some point) in an $r$-dimensional subspace $\pi$ with the property that any $r$ of them span $\pi$. In Figure~\ref{fig:sputnik} a $3$-sputnik is illustrated.
\begin{figure}[tbh]
\centering
 \includegraphics[width=0.3\textwidth]{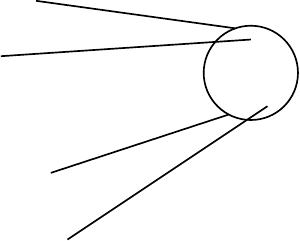}
 \caption{\label{fig:sputnik} A $3$-sputnik looks quite like a Soviet radio satellite from 1957.}
\end{figure}

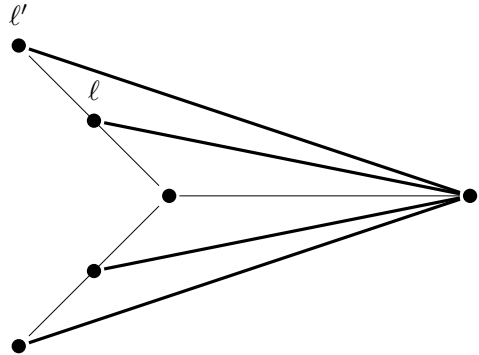
\begin{figure}
\begin{center}
\begin{tikzpicture} 
       \node (e8) at (0,0) {};
       \node (e7) at (-4,0){};
       \node (e1) at (-5,-1) {};
       \node (e2) at (-6,-2){}; 
       \node (e3) at (-5,1) [label=above:$\ell$]{}; 
       \node (e4) at (-6,2) [label=above:$\ell'$]{};
        \draw[very thick] (e8) -- (e1) {};
        \draw[very thick] (e8) -- (e2);  
	\draw[very thick] (e8) -- (e3);
       	\draw[very thick] (e8) -- (e4);    
        \draw[ultra thin] (e7) -- (e8); 
        \draw[ultra thin] (e7) -- (e2); 
        \draw[ultra thin] (e7) -- (e4);
       \fill (e1) circle (2.7pt);
       \fill (e2) circle (2.7pt);
       \fill (e3) circle (2.7pt);
       \fill (e4) circle (2.7pt);
       \fill (e7) circle (2.7pt);
       \fill (e8) circle (2.7pt);
     \end{tikzpicture} 
     
\end{center}
\caption{A $3$-sputnik seen as the union modulo two of two planar pencils of lines.} \label{3sputnik}
\end{figure}

Our aim will be to prove that a quantum set of lines is the union modulo two of planar pencils of lines. Firstly we will prove that this claim is true for an $r$-sputnik.

\begin{lemma} \label{sputnikisaqsol}
An $r$-sputnik is the union modulo two of planar pencils of lines. In particular, an $r$-sputnik is a quantum set of lines.
\end{lemma}

\begin{proof}
Let $\mathcal{X}$ be an $r$-sputnik and take any two lines $\ell$ and $\ell' \in \mathcal{X}$. The $r-1$ lines of $\mathcal{X} \setminus \{\ell,\ell'\}$ span a $(r-1)$-dimensional subspace which intersects the plane spanned by $\ell$ and $\ell'$ in a line $\ell''$. The line $\ell''$ is the third line in the planar pencil of lines spanned by $\ell$ and $\ell'$. Thus, adding (modulo $2$) this pencil of lines to $\mathcal{X}$ we get an $(r-1)$-sputnik. Now continue adding planar pencils of lines in this way until we get a $2$-sputnik. Since a $2$-sputnik is a planar pencil of lines, it is a quantum set of lines. We can then reverse the process adding planar pencils of lines to recover the $r$-sputnik which, by Lemma~\ref{qlunion}, is also a quantum set of lines.
\end{proof}

\begin{lemma} \label{dependqsol}
Let $\mathcal{X}$ be a quantum set of lines. There is a set $D$ of dependent points such that each point of $D$ is incident with a different line of $\mathcal{X}$.
\end{lemma}

\begin{proof}
Let $\pi=\Theta(\XX)$ be the subspace spanned by the lines of $\mathcal{X}$ and let $\ell \in \mathcal{X}$. Let $\pi'=\Theta(\XX \setminus \{\ell\})$ be the subspace spanned by the lines of $\mathcal{X} \setminus \{ \ell\}$. The subspace $\pi'$ is either a co-dimension $2$ subspace of $\pi$, a hyperplane of $\pi$, or $\pi$ itself. The first case is ruled out since $\mathcal{X}$ is a quantum set of lines and, by definition, any co-dimension $2$ subspace is skew to an even number of lines of $\mathcal{X}$. Therefore, there is a point of $x$ of $\ell$ incident with $\pi'$. Any point of $\pi'$ is the sum of points incident with the lines of $\mathcal{X} \setminus \{ \ell\}$. Thus, we obtain a set of dependent points each incident with a line of $\mathcal{X}$. If in this set there are two points $y$ and $z$ incident with same line $\ell'$ of $\mathcal{X}$, then we can replace $y$ and $z$ by $\ell' \setminus \{y,z\}$. Hence, we obtain a set of dependent points each incident with a distinct line of $\mathcal{X}$.
\end{proof}

\begin{lemma} \label{qsol3}
A quantum set of three lines is a planar pencil of lines.
\end{lemma}

\begin{proof}
Suppose that the quantum set of three lines $\mathcal{X}=\{\ell_1,\ell_2,\ell_3\}$ span $\mathrm{PG}(4,2)$ or $\mathrm{PG}(5,2)$ respectively. Then there is a point $x \in \ell_2$ such that the co-dimension $2$ subspace spanned by $\ell_1$ and $x$ (resp. $\ell_1$ and $\ell_2$) is skew to $\ell_3$. This contradicts the definition of a quantum set of lines.

Suppose that the quantum set of three lines $\mathcal{X}=\{\ell_1,\ell_2,\ell_3\}$ span $\mathrm{PG}(3,2)$. If $\ell_1$ and $\ell_2$ intersect then the co-dimension $2$ subspace $\ell_1$ (and also $\ell_2$) must also intersect $\ell_3$. Since they span $\mathrm{PG}(3,2)$ the three lines must be concurrent (and not co-planar). Taking the union modulo $2$ of the planar pencil of lines spanned by $\ell_2$ and $\ell_3$ we obtain, by Lemma~\ref{qlunion}, a quantum set of two lines, which does not exist. Thus we have three pairwise skew lines $\ell_1,\ell_2,\ell_3$ with the property that any line incident with two of them is incident with the third. This implies there are nine lines which are all incident with exactly one point of each of $\ell_1,\ell_2,\ell_3$, see Figure~\ref{32}. By Lemma~\ref{numberofsubspaces}, a point of $\mathrm{PG}(3,2)$ is incident with seven lines of $\mathrm{PG}(3,2)$, so in all we have that there are (at least)
$$
9(7-4)+3+9=39
$$
lines of $\mathrm{PG}(3,2)$, when in fact, by Lemma~\ref{numberofsubspaces}, there are $35$.

Therefore, the quantum set of three lines span a $\mathrm{PG}(2,2)$. A co-dimension $2$ subspace is just a point, so a quantum set of lines must be incident with every point of the plane. Hence, $\mathcal{X}$ is a planar pencil of lines.
\end{proof}

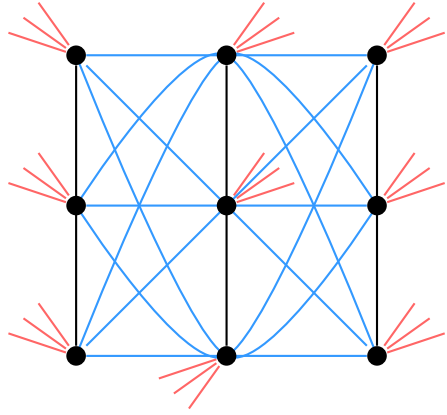
\begin{figure}
\centering
\begin{tikzpicture}
\definecolor{blau}{RGB}{51, 153, 255}
\definecolor{v}{RGB}{255, 102, 102}
	\node (A) at (0,4){};
	\node (B) at (0,2){};
	\node (C) at (0,0){};
	\node (D) at (2,4){};
	\node (E) at (2,2){};
	\node (F) at (2,0){};
	\node (G) at (4,4){};
	\node (H) at (4,2){};
	\node (I) at (4,0){};
	
	\draw[v, thick] (G) -- (4.5, 4.7);
	\draw[v, thick] (G) -- (4.7, 4.5);
	\draw[v, thick] (G) -- (4.9, 4.3);
	
	\draw[v, thick] (H) -- (4.5, 2.7);
	\draw[v, thick] (H) -- (4.7, 2.5);
	\draw[v, thick] (H) -- (4.9, 2.3);
	
	\draw[v, thick] (I) -- (4.5, 0.7);
	\draw[v, thick] (I) -- (4.7, 0.5);
	\draw[v, thick] (I) -- (4.9, 0.3);
		
	\draw[v, thick] (A) -- (-0.5, 4.7);
	\draw[v, thick] (A) -- (-0.7, 4.5);
	\draw[v, thick] (A) -- (-0.9, 4.3);
	
	\draw[v, thick] (B) -- (-0.5, 2.7);
	\draw[v, thick] (B) -- (-0.7, 2.5);
	\draw[v, thick] (B) -- (-0.9, 2.3);
	
	\draw[v, thick] (C) -- (-0.5, 0.7);
	\draw[v, thick] (C) -- (-0.7, 0.5);
	\draw[v, thick] (C) -- (-0.9, 0.3);
	
	\draw[v, thick] (D) -- (2.5, 4.7);
	\draw[v, thick] (D) -- (2.7, 4.5);
	\draw[v, thick] (D) -- (2.9, 4.3);
	
	\draw[v, thick] (E) -- (2.5, 2.7);
	\draw[v, thick] (E) -- (2.7, 2.5);
	\draw[v, thick] (E) -- (2.9, 2.3);
	
	\draw[v, thick] (F) -- (1.5, -0.7);
	\draw[v, thick] (F) -- (1.3, -0.5);
	\draw[v, thick] (F) -- (1.1, -0.3);

	\draw[thick] (A) -- (C);
	\draw[thick] (D) -- (F);
	\draw[thick] (G) -- (I);
	
	\draw[blau, thick] (A) -- (G);
	\draw[blau, thick] (A) -- (I);
	\draw[blau, thick] plot[smooth] coordinates {(A) (F) (H)};
	
	\draw[blau, thick] (B) -- (H);
	\draw[blau, thick] plot[smooth] coordinates {(B) (D) (I)};
	\draw[blau, thick] plot[smooth] coordinates {(B) (F) (G)};
	
	\draw[blau, thick] (C) -- (I);
	\draw[blau, thick] (C) -- (G);
	\draw[blau, thick] plot[smooth] coordinates {(C) (D) (H)};

	\fill (A) circle (3.7pt);
    \fill (B) circle (3.7pt);     
    \fill (C) circle (3.7pt);
    \fill (D) circle (3.7pt);
    \fill (E) circle (3.7pt);
    \fill (F) circle (3.7pt);
    \fill (G) circle (3.7pt);
    \fill (H) circle (3.7pt);
    \fill (I) circle (3.7pt);

\end{tikzpicture}
\caption{Configuration of the lines in $PG(3,2)$.}\label{32}
\end{figure}

The following theorem is due to Glynn, Gulliver, Maks and Gupta  \cite{GGMG}. It is important to note that if the qubit stabilizer code has minimum distance $2$ then it is possible that the quantum set of lines $\mathcal{X}$ contains repeated lines. This occurs, for example, in the $[\![5,2,2]\!]$ code. 

\begin{theorem} \label{3GMakx}
A qubit stabilizer code with minimum distance at least three is equivalent to a quantum set of lines which is generated by the union modulo two of planar pencils of lines.
\end{theorem}

\begin{proof}
Let $\mathcal{X}$ be a quantum set of lines. We will prove that there is an $r$-sputnik $\mathcal{X}'$ such that the union modulo $2$ of $\mathcal{X}$, $\mathcal{X}'$ and $r-1$ planar pencils of lines is a quantum set of $|\mathcal{X}|-1$ lines. Since, by Lemma~\ref{sputnikisaqsol}, $\mathcal{X}'$ is the union modulo~$2$ of planar pencils of lines, this implies that, by iteration, we can take the union modulo $2$ of $\mathcal{X}$ and some planar pencils of lines and obtain a quantum set of three lines, by Lemma~\ref{qlunion}. By Lemma~\ref{qsol3}, this set of three lines is a planar pencil of lines and we are done.

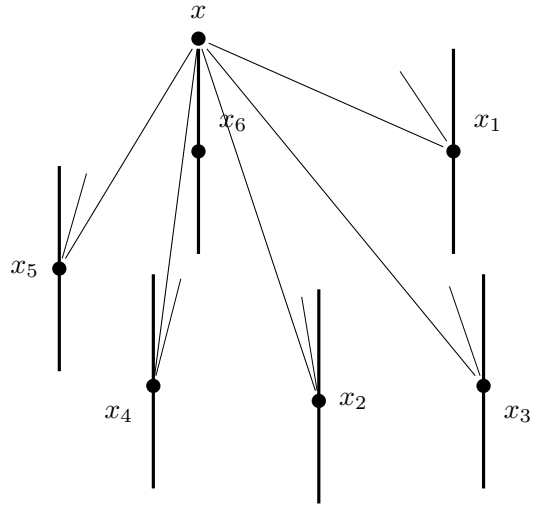
\begin{figure}
\begin{center}
\begin{tikzpicture} 
       \node (e0) at (1.85,2.56) [label=above right:$x_6$]{};
       \node (e1) at (5.24,2.56) [label=above right:$x_1$]{};
       \node (e2) at (3.45,-0.76) [label=right:$x_2$]{}; 
       \node (e3) at (5.64,-0.56) [label=below right:$x_3$]{}; 
       \node (e4) at (1.25,-0.56) [label=below left:$x_4$]{}; 
       \node (e5) at (0,1.0) [label=left:$x_5$]{}; 
       \node (e0u) at (1.85,4.06) [label=above:$x$]{};
       \node (e0b) at (1.85,1.06) [label=below:$$]{};
       \node (e1u) at (5.24,4.06) [label=below:$$]{};
       \node (e1b) at (5.24,1.06) [label=below:$$]{};
       \node (e1r) at (4.44,3.76) [label=below:$$]{};
       \node (e2u) at (3.45,0.86) [label=below:$$]{};
       \node (e2b) at (3.45,-2.26) [label=below:$$]{};
       \node (e2r) at (3.20,0.76) [label=below:$$]{};
       \node (e3u) at (5.64,1.06) [label=below:$$]{};
       \node (e3b) at (5.64,-2.06) [label=below:$$]{};
       \node (e3r) at (5.14,0.90) [label=below:$$]{};
       \node (e4u) at (1.25,1.06) [label=below:$$]{};
       \node (e4b) at (1.25,-2.06) [label=below:$$]{};
       \node (e4r) at (1.65,1.00) [label=below:$$]{};
       \node (e5u) at (0,2.5) [label=below:$$]{};
       \node (e5b) at (0,-0.5) [label=below:$$]{};
       \node (e5r) at (0.4,2.4) [label=below:$$]{};
       \draw (e0u) -- (e1); 
       \draw (e0u) -- (e2); 
       \draw (e0u) -- (e3); 
       \draw (e0u) -- (e4); 
       \draw (e0u) -- (e5); 
       \draw[ultra thin] (e5) -- (e5r); 
       \draw[ultra thin] (e4) -- (e4r); 
       \draw[ultra thin] (e3) -- (e3r); 
       \draw[ultra thin] (e2) -- (e2r); 
       \draw[ultra thin] (e1) -- (e1r);
       \draw[very thick] (e0b) -- (e0u); 
       \draw[very thick] (e1b) -- (e1u); 
       \draw[very thick] (e2b) -- (e2u); 
       \draw[very thick] (e3b) -- (e3u); 
       \draw[very thick] (e4b) -- (e4u);  
       \draw[very thick] (e5b) -- (e5u); 
       \fill (e0) circle (2.7pt);
       \fill (e1) circle (2.7pt);
       \fill (e2) circle (2.7pt);
       \fill (e3) circle (2.7pt);
       \fill (e4) circle (2.7pt);
       \fill (e5) circle (2.7pt);
        \fill (e0u) circle (2.7pt);

     \end{tikzpicture} 
    
\end{center}
\caption{The thick lines are in $\mathcal{X}$, the medium-thick lines are in $\mathcal{X}'$ and the thin lines make up the planar pencils at each point $x_1,\ldots,x_r$.} \label{addthesputnik}
\end{figure}

By Lemma~\ref{dependqsol}, there is a set $x_1,\ldots,x_{r+1}$ of minimally dependent points incident with the lines $\ell_1,\ldots,\ell_{r+1}$ of $\mathcal{X}$, respectively. Let $x \in \ell_{r+1}\setminus \{x_{r+1}\}$. Let $\ell_j'$ be the line spanned by the points $x$ and $x_j$, for $j=1,\ldots,r$. Let $\mathcal{X}'$ be the $r$-sputnik, 
$$
\mathcal{X}'=\{\ell_j' \ | \ j=1,\ldots,r\} \cup \{\ell_{r+1}\}.
$$
Let $\mathcal{L}_j$ be the planar pencil of lines spanned by $\ell_j$ and $\ell_j'$. In Figure~\ref{addthesputnik}, $r=5$, the lines $\ell_j$ are the thick lines, the $\ell_j'$ are the medium thickness lines and the thin lines are the third line in the planar pencil of lines spanned by $\ell_j$ and $\ell_j'$.

By Lemma~\ref{qlunion}, the union modulo two of
$$
(\cup_{j=1}^r \mathcal{L}_j) \cup \mathcal{X}\cup \mathcal{X}'
$$
is a quantum set of lines and, on inspection, it is a set of $|\mathcal{X}|-1$ lines.
\end{proof}

\begin{example} \label{5zero3code2}
Consider again the $[\![5,0,3]\!]$ code constructed in Example~\ref{5zero3code}. As a quantum set of lines $\mathcal{X}$, this is the union modulo two of pencils of lines drawn in Figure~\ref{503figure}.

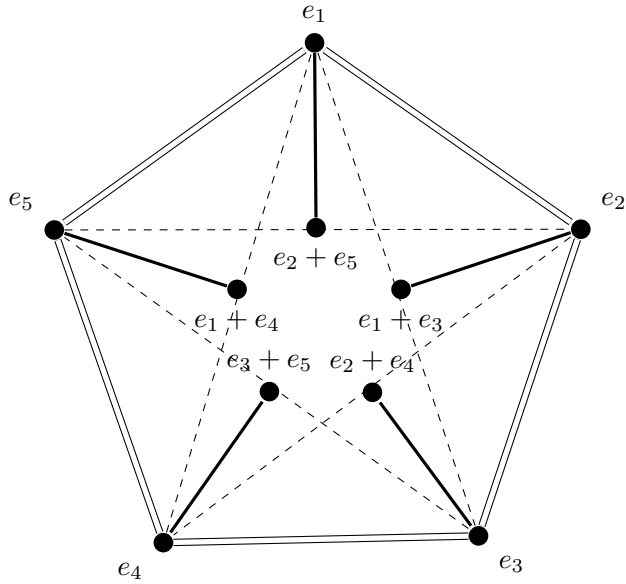
\begin{figure}
\begin{center}
\begin{tikzpicture} 
       \node (e1) at (3.46,5) [label=$e_1$]{};
       \node (e2) at (7,2.52) [label=above right:$e_2$]{}; 
       \node (e3) at (5.64,-1.56) [label=below right:$e_3$]{}; 
       \node (e4) at (1.45,-1.65) [label=below left:$e_4$]{}; 
       \node (e5) at (0,2.51) [label=above left:$e_5$]{}; 
       \node (e13) at (4.61,1.72) [label=below:$e_1+e_3$]{};
       \node (e24) at (4.23,0.35) [label=above :$e_2+e_4$]{}; 
       \node (e35) at (2.86,0.36) [label=above :$e_3+e_5$]{}; 
       \node (e14) at (2.43,1.72) [label=below :$e_1+e_4$]{}; 
       \node (e25) at (3.48,2.54) [label=below:$e_2+e_5$]{}; 
       \draw[double distance=2pt] (e1) -- (e2); 
       \draw[double distance=2pt] (e2) -- (e3); 
       \draw[double distance=2pt] (e3) -- (e4); 
       \draw[double distance=2pt] (e4) -- (e5); 
       \draw[double distance=2pt] (e5) -- (e1); 
       \draw[dashed] (e4) -- (e1); 
       \draw[dashed] (e1) -- (e3); 
       \draw[dashed] (e3) -- (e5); 
       \draw[dashed] (e5) -- (e2);
       \draw[dashed] (e2) -- (e4); 
       \draw[very thick] (e1) -- (e25); 
       \draw[very thick] (e2) -- (e13); 
       \draw[very thick] (e3) -- (e24); 
       \draw[very thick] (e4) -- (e35); 
       \draw[very thick] (e5) -- (e14);
        
       \fill (e1) circle (3.7pt);
       \fill (e2) circle (3.7pt);
       \fill (e3) circle (3.7pt);
       \fill (e4) circle (3.7pt);
       \fill (e5) circle (3.7pt);
       \fill (e13) circle (3.7pt);
       \fill (e24) circle (3.7pt);
       \fill (e35) circle (3.7pt);
       \fill (e14) circle (3.7pt);
       \fill (e25) circle (3.7pt);

     \end{tikzpicture} 
\end{center}
\caption{The $[\![5,0,3]\!]$ code as the union modulo two of planar pencils of lines.} \label{503figure}
\end{figure}

Since $k=0$, $d(\mathcal X)$ is the minimum $d$ for which there is a hyperplane of $\mathrm{PG}(4,q)$ containing $|\mathcal{X}|-d=5-d$ lines of $\mathcal{X}$. Since any three lines span the whole space, we have that $d=3$. Thus, this is a $[\![5,0,3]\!]$ code.

We can also construct the $[\![5,1,3]\!]$ code from Figure~\ref{503figure}. We only have to replace $e_5$ with $e_1+e_2+e_3+e_4$ and check that the five (thick) lines are then pairwise skew. This can be done by writing down the $15$ points and checking we get every point of $\mathrm{PG}(3,2)$. Then, since any two of the thick lines are pairwise skew, we have that the minimum distance is $3$.
\end{example}

\begin{example}

The $[\![6,0,4]\!]$ code is the sum modulo $2$ of 16 planar pencils of lines, see Figure~\ref{64}. The cyclic structure allows one to check quickly that there are no three collinear points intersecting distinct lines of the six lines of the quantum set of lines. Indeed, the points of weight two obtained by summing two points incident with the quantum lines are cyclic shifts of $26, 36, 46$ and the points of weight three obtained by summing two points incident with the quantum lines are cyclic shifts of $134$ and $146$. Therefore, the minimum distance of the code is at least $4$. The points $e_{126},e_{34}, e_{16}, e_{234}$ are four dependent points, implying that the minimum distance of the code is $4$.

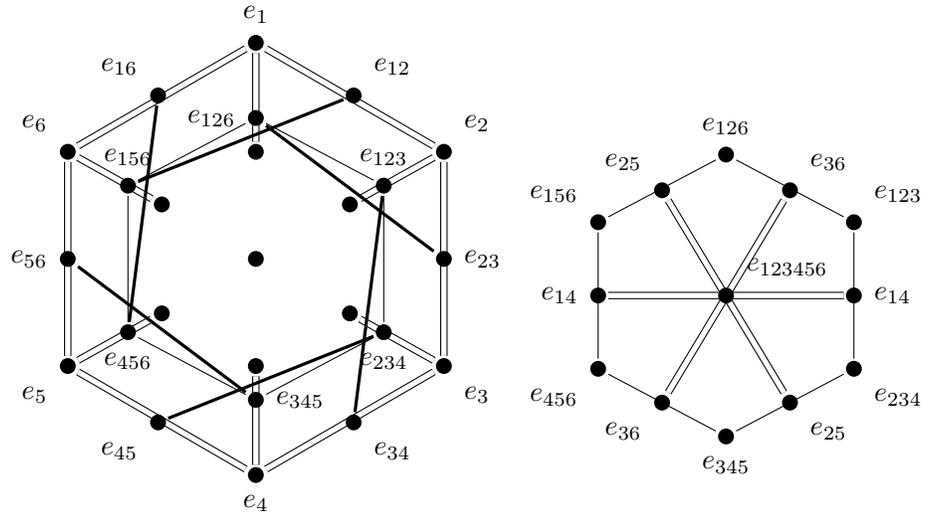
\begin{figure}
\centering
\begin{tikzpicture}[scale = 0.5]
\node (e1) at (5, 11.5)[label = above:$e_1$]{};
\node (e2) at (10, 8.6)[label = above right:$e_2$]{};
\node (e3) at (10, 2.9)[label = below right:$e_3$]{};
\node (e4) at (5, 0)[label = below:$e_4$]{};
\node (e5) at (0, 2.9)[label = below left:$e_5$]{};
\node (e6) at (0, 8.6)[label = above left:$e_6$]{};
\node (e12) at (7.6, 10.1)[label = above right:$e_{12}$]{};
\node (e23) at (10, 5.75)[label = right:$e_{23}$]{};
\node (e34) at (7.6, 1.4)[label = below right:$e_{34}$]{};
\node (e45) at (2.4, 1.4)[label = below left:$e_{45}$]{};
\node (e56) at (0, 5.75)[label = left:$e_{56}$]{};
\node (e16) at (2.4, 10.1)[label = above left:$e_{16}$]{};
\node (e123) at (8.4, 7.7)[label = above:$e_{123}$]{};
\node (e234) at (8.4, 3.8)[label = below:$e_{234}$]{};
\node (e345) at (5, 2)[label = right:$e_{345}$]{};
\node (e456) at (1.6, 3.8)[label = below :$e_{456}$]{};
\node (e156) at (1.6, 7.7)[label = above:$e_{156}$]{};
\node (e126) at (5, 9.5)[label = left:$e_{126}$]{};
\node (e26) at (5,8.6)[]{};
\node (e35) at (5, 2.9)[]{};
\node (e46) at (2.5, 4.3)[]{};
\node (e24) at (7.5, 4.3)[]{};
\node (e123456) at (5,5.75)[]{};
\node (e13) at (7.5, 7.2)[]{};
\node (e15) at (2.5, 7.2)[]{};

\draw[double distance=2pt] (e1) -- (e2);
\draw[double distance=2pt] (e2) -- (e3);
\draw[double distance=2pt] (e3) -- (e4);
\draw[double distance=2pt] (e4) -- (e5);
\draw[double distance=2pt] (e5) -- (e6);
\draw[double distance=2pt] (e6) -- (e1);

\draw[double distance=2pt] (e1) -- (e26);
\draw[double distance=2pt] (e4) -- (e35);
\draw[double distance=2pt] (e2) -- (e13);
\draw[double distance=2pt] (e6) -- (e15);
\draw[double distance=2pt] (e5) -- (e46);
\draw[double distance=2pt] (e3) -- (e24);

\draw[very thick] (e12)--(e156);
\draw[very thick] (e23) -- (e126);
\draw[very thick] (e34) -- (e123);
\draw[very thick] (e45) -- (e234);
\draw[very thick] (e345) -- (e56);
\draw[very thick] (e456) -- (e16);

\draw (e156) -- (e126);
\draw (e126) -- (e123);
\draw (e123) -- (e234);
\draw (e234) -- (e345);
\draw (e456) -- (e345);
\draw (e156) -- (e456);

\fill (e1) circle (6pt);
\fill (e2) circle (6pt);
\fill (e3) circle (6pt);
\fill (e4) circle (6pt);
\fill (e5) circle (6pt);
\fill (e6) circle (6pt);
\fill (e12) circle (6pt);
\fill (e23) circle (6pt);
\fill (e34) circle (6pt);
\fill (e45) circle (6pt);
\fill (e56) circle (6pt);
\fill (e16) circle (6pt);
\fill (e123) circle (6pt);
\fill (e234) circle (6pt);
\fill (e345) circle (6pt);
\fill (e456) circle (6pt);
\fill (e156) circle (6pt);
\fill (e126) circle (6pt);
\fill (e26) circle (6pt);
\fill (e35) circle (6pt);
\fill (e123456) circle (6pt);
\fill (e46) circle (6pt);
\fill (e24) circle (6pt);
\fill (e13) circle (6pt);
\fill (e15) circle (6pt);
\end{tikzpicture}
\begin{tikzpicture}[scale = 0.5]
\node (e1) at (5, 11.5)[]{};
\node (e2) at (10, 8.6)[]{};
\node (e3) at (10, 2.9)[]{};
\node (e4) at (5, 0)[]{};
\node (e5) at (0, 2.9)[]{};
\node (e6) at (0, 8.6)[]{};
\node (e12) at (7.6, 10.1)[]{};
\node (e23) at (10, 5.75)[]{};
\node (e34) at (7.6, 1.4)[]{};
\node (e45) at (2.4, 1.4)[]{};
\node (e56) at (0, 5.75)[]{};
\node (e16) at (2.4, 10.1)[]{};
\node (e123) at (8.4, 7.7)[label = above right:$e_{123}$]{};
\node (e234) at (8.4, 3.8)[label = below right:$e_{234}$]{};
\node (e345) at (5, 2)[label = below:$e_{345}$]{};
\node (e456) at (1.6, 3.8)[label = below left:$e_{456}$]{};
\node (e156) at (1.6, 7.7)[label = above left:$e_{156}$]{};
\node (e126) at (5, 9.5)[label = above:$e_{126}$]{};
\node (e26) at (5,8.6)[]{};
\node (e35) at (5, 2.9)[]{};
\node (e46) at (2.5, 4.3)[]{};
\node (e24) at (7.5, 4.3)[]{};
\node (e123456) at (5,5.75)[label=above right:$e_{123456}$]{};
\node (e13) at (7.5, 7.2)[]{};
\node (e15) at (2.5, 7.2)[]{};
\node (e36) at (3.3, 2.9)[label = below left:$e_{36}$]{};
\node (e25) at (3.3, 8.55)[label = above left:$e_{25}$]{};
\node (e14) at (1.6, 5.75)[label = left:$e_{14}$]{};
\node (e14*) at (8.4,5.75)[label = right:$e_{14}$]{};
\node (e25*) at (6.7, 2.9)[label = below right:$e_{25}$]{};
\node (e36*) at (6.7, 8.55)[label = above right:$e_{36}$]{};

\draw (e156) -- (e126);
\draw (e126) -- (e123);
\draw (e123) -- (e234);
\draw (e234) -- (e345);
\draw (e456) -- (e345);
\draw (e156) -- (e456);
\draw[double distance = 2pt] (e36) -- (e36*);
\draw[double distance = 2pt] (e25) -- (e25*);
\draw[double distance = 2pt] (e14) -- (e14*);

\fill (e123) circle (6pt);
\fill (e234) circle (6pt);
\fill (e345) circle (6pt);
\fill (e456) circle (6pt);
\fill (e156) circle (6pt);
\fill (e126) circle (6pt);
\fill (e123456) circle (6pt);
\fill (e36) circle (6pt);
\fill (e25) circle (6pt);
\fill (e14) circle (6pt);
\fill (e36*) circle (6pt);
\fill (e14*) circle (6pt);
\fill (e25*) circle (6pt);
\end{tikzpicture}
\caption{The quantum set of lines (the thicker lines) giving a $[\![6,0,4]\!]$ code.} \label{64}
\end{figure}

\end{example}

\begin{rprob}
The parameters $[\![14,3,5]\!]$ are the smallest for which it is unknown whether there exists a qubit stabilizer code or not~\cite{codetables}. To construct such a code one should look for a union modulo two of planar pencils of lines that give $14$ lines in $\mathrm{PG}(10,2)$, such that any four points on $4$ of the $14$ lines
that also lie on a common plane,
the remaining $10$ lines are contained in a hyperplane which also contains those four dependent points. 
\end{rprob}

Theorem~\ref{3GMakx} can also be used to rule out the existence of quantum codes with certain parameters sets. For example, were a $[\![4,0,3]\!]$ stabilizer code to exist then $\mathcal{X}$ would be a set of four skew lines in $\mathrm{PG}(3,2)$ with the property that any line is skew to an even number of lines of $\mathcal{X}$. However, the lines of $\mathcal X$ themselves are skew to the other three lines of $\mathcal X$, which is an odd number. A more interesting exercise is to prove that a $[\![7,0,4]\!]$ code does not exist. To prove this, show that there are at least five three dimensional subspaces which intersect all of the $7$ lines of $\mathrm{PG}(6,2)$ in the quantum set of lines and prove that these pairwise intersect in a point.

\section{Non-additive qubit quantum codes} \label{unionstabilizers}

\subsection{Direct sum of stabilizer codes} 

As discussed in the previous sections, 
a stabilizer code is defined as the common $(+1)$-eigenspace of an independent set of pairwise commuting Pauli operators $M_1, \dots, M_{n-k}$; this is the generator of the code. In other words, these codes are completely characterized by an abelian subgroup~$S = \langle M_1, \dots, M_{n-k}\rangle \subset \mathcal{P}_n$.
The aim of this section is to construct quantum codes that are the {\em direct sum} of stabilizer codes.
Technically speaking, any subspace can be regarded as a quantum code,
and naturally we want to make sure to obtain a large mininum distance when taking this direct sum of subspaces. 
Thus, we seek for some additional structure amongst them.
While each individual subspace will again be defined by a set of generators $M_1,\ldots,M_{n-k}$, 
we will now not simply take the joint eigenspace with eigenvalue~$1$ as our code space.

We have already observed that to avoid constructing a trivial code, 
one restricts the stabilizer not to contain a non-trivial multiple of the identity, $-\one \not\in S$. 
This implies that each generator can only have an overall phase of $+1$ or $-1$
and they are of the form
$$
M_j= \pm \sigma_{1} \otimes \cdots \otimes \sigma_{n}
$$
for some $\sigma_{1}, \dots, \sigma_{n} \in \mathcal{P}_1$. 
Now observe that when $M_1,\ldots,M_{n-k}$ commute, then so do
$$
\pm M_1,\ldots,\pm M_{n-k}\,.
$$

Thus for all $t=(t_1,\ldots,t_{n-k}) \in \{0,1\}^{n-k}$, one can define a corresponding stabilizer code $Q(S_t)$ 
as the joint $(+1)$-eigenspace of
$$
(-1)^{t_1} M_1,\ldots,(-1)^{t_{n-k}} M_{n-k}.
$$

For distinct $t$ and $t' \in T$, there is a $j$ such that $t_j \neq t'_j$. Without loss of generality, suppose that $t_j=1$. For all $\ket{v} \in Q(S_t)$ and $\ket{w}\in Q(S_{t'})$, one has $\braket{v}{w} = \bra{v}\ket{M_j w} = \bra{M_jv}\ket{w} = -\braket{v}{w} = 0$. Consequently, $Q(S_t)$ and $Q(S_{t'})$ are orthogonal.

For any $T \subset  \{0,1\}^m$, we define a {\em direct sum stabilizer code} (confusingly also known as a {\em union stabilizer code}) as
$$
Q(S_T)= \bigoplus_{t \in T} Q(S_t).
$$

To be able to determine the minimum distance of this quantum code, we first determine the errors which are not detectable.

As before, let $\mathrm{G}$ be the generator matrix whose row space is $C=\tau(S)$.

Let $t,u \in T \setminus \{0\}$ and let $\mathrm{A}_{t,u}$ be a $(n-k) \times (n-k)$ non-singular matrix whose first two columns are $t$ and $u$. Then $\mathrm{A}_{t,u}^{-1}\mathrm{G}$ is also a generator matrix for $C$ and we can find another set 
$$
\{M_i' \ | i=1,\ldots,n-k \}
$$
of generators of $S$, where $M_i'$ is obtained from the $i$-th row of $A_{t,u}^{-1}G$ by applying $\tau^{-1}$, in other words reversing the construction above. 

Let $S_{t,u}$ be the subgroup of $S$ generated by $M'_3,\ldots,M'_{n-k}$. 

\begin{lemma} \label{stulemma}
Suppose $\ket{\psi^t} \in Q_t(S)$ and $\ket{\psi^{u}} \in Q_{u}(S)$. Then, for all $M\in S_{t,u}$, 
$$
M\ket{\psi^t}=\ket{\psi^t} \mathrm{and} \ M \ket{\psi^{u}}=\ket{\psi^{u}}.
$$
\end{lemma}

\begin{proof}
Observe that $Q_t(S)$ depends on the set of generators we have chosen for $S$. If we use the set of generators $M'_1,\ldots,M'_{n-k}$ for $S$ then $Q_t(S)$ becomes $Q_{(1,0,0,\ldots,0)}(S)$ and $Q_u(S)$ becomes $Q_{(0,1,0,\ldots,0)}(S)$. Thus, $M'_j \ket{\psi^t}=\ket{\psi^t}$ and $M'_j \ket{\psi^{u}}=\ket{\psi^{u}}$ for all $j\in \{3,\ldots,n-k\}$.
\end{proof}.

\begin{lemma} \label{uniondetect}
Suppose $Q(S_T)$ is unable to detect an error $E$. Then there is a pair $t,u \in T$ such that $E  \in \mathrm{Centraliser}(S_{t,u})$.
\end{lemma}

\begin{proof}
Suppose there is no such pair.
Then, for all $t,u \in T$, there is a $M_{t,u} \in S_{t,u}$ for which $E$ anti-commutes with $M_{t,u}$.

Suppose $\ket{\psi^t} \in Q_t(S)$ and $\ket{\psi^{u}} \in Q_{u}(S)$ are in an orthogonal basis for $Q(S_T)$. By Lemma~\ref{stulemma},
$$
M_{t,u}\ket{\psi^t}=\ket{\psi^t} \mathrm{and} \ M_{t,u} \ket{\psi^{u}}=\ket{\psi^{u}}
$$
and so
$$
\bra{\psi^t} E \ket{\psi^u} = \bra{\psi^t} EM_{t,u} \ket{\psi^u}= -\bra{\psi^t} M_{t,u} E \ket{\psi^u}=
 -\bra{\psi^t} E\ket{\psi^u}.
$$
Hence,
$$
\bra{\psi^t} E \ket{\psi^u} =0
$$ 
and by Theorem~\ref{KLtheorem}, $E$ is detectable.
\end{proof}

Thus, according to Lemma~\ref{uniondetect}, we only need concern ourselves with the errors which are in $\mathrm{Centraliser}(S_{t,u})$ for any $t,u\in T$.

This motivates the definition
\begin{equation} \label{uniondistance}
d_T=\mathrm{min} \{d_{t,u} \ | \ t,u \in T\} 
\end{equation}
where $d_{t,u}$ is the minimum weight of a Pauli operator in $\mathrm{Centralise}(S_{t,u})$.

\begin{theorem} \label{unionthm}
The subspace $Q(S_T)$ is an $(\!(n,|T|2^k,d_T)\!)$ quantum code.
\end{theorem}

\begin{proof}
If $E$ is undetectable then it is an element of $\mathrm{Centraliser}(S_{t,u})$ for some $t,u\in T$. 
\end{proof}

\subsection{The Rains, Hardin, Shor, Sloane non-additive quantum code}

This code first appeared in \cite{RHSS1997}, although the geometric observation given here appears to be new.

\begin{example} (Rains, Hardin, Shor, Sloane) \label{562code}
Consider the following elements of $\mathcal{P}_5$.
$$
\begin{array}{rccccccccc} 
M_1= & Z & X & Y & Y & X  \\ 
M_2 =& X & Z & X & Y & Y  \\ 
M_3 =& Y & X & Z & X & Y  \\
M_4 =& Y & Y &  X & Z & X  \\
M_5 =& X & Y & Y &  X & Z  \\
\end{array}
$$
The corresponding matrix whose $i$-th row is $\tau(M_i)$ is
$$
\left(
\begin{array}{ccccc|ccccc}
0 & 1 & 1 & 1 & 1 & 1 & 0 & 1 & 1 & 0\\
1 & 0 & 1 & 1 & 1 & 0 & 1 & 0 & 1 & 1 \\
1 & 1 & 0 & 1 & 1 & 1 & 0 & 1 & 0 & 1  \\
1 & 1 & 1 & 0 & 1 & 1 & 1 & 0 & 1 & 0 \\
1 & 1 & 1 & 1 & 0 & 0 & 1 & 1 & 0 & 1  \\
\end{array}
\right).
$$
Observe that deleting any two rows of this matrix we obtain a $3 \times 10$ matrix whose $5$ pairs of columns define a quantum set of lines in $\mathrm{PG}(2,2)$. This quantum set of lines defines a stabilizer code whose minimum distance is $2$. Therefore, if we set
$$
T=\{\emptyset,\{1\},\{2\},\{3\},\{4\},\{5\} \}
$$
then, by Theorem~\ref{unionthm}, $Q(S_T)$ is a $(\!(5,6,2)\!)$ quantum code.

\end{example}

\subsection{The geometry of direct sum stabilizer codes}

Suppose that we restrict our choice of elements of $T$ to singleton subsets and the empty set, as in Example~\ref{562code}. Let $\mathcal{X}$ be the quantum set of lines of $\mathrm{PG}(n-k-1,2)$ associated with the $[\![n,k,d]\!]$ quantum stabilizer code $Q(S)$, where $S$ is the subgroup generated by $M_1,\ldots,M_{n-k}$. Let $P=\{e_1,\ldots,e_r\}$ be a set of linearly independent points of $\mathrm{PG}(n-k-1,2)$, chosen so that the projection from any two points $e_i,e_j \in P$ of the lines of $\mathcal{X}$ is a set of lines of $\mathrm{PG}(n-k-3,2)$. If this projection is a set of lines then it is necessarily a quantum set of lines, which we denote by $\mathcal{X}_{ij}$. The mini

The parameter $d({\mathcal X}_{ij})$ is the size of the smallest set of dependent points incident with distinct lines of $\mathcal{X}_{ij}$. Thus, the definition in (\ref{uniondistance}) will be
$$
d_T=\mathrm{min} \{ d(\mathcal{X}_{ij}) \ | \ i,j \in \{1,\ldots,r\} \}.
$$
Hence, we have a purely geometric way to construct direct sum stabilizer codes with parameters $(\!(n,(r+1)2^k,d_T)\!)$, for some $r \leqslant n-k$.

This is taken much further in \cite{BP2021}, where the geometrical construction is generalised to prime alphabets.

\begin{rprob}
Find quantum sets of lines $\mathcal{X}$ for which there are points with the property that the projection of the lines of $\mathcal{X}$ from any pair is onto a quantum set of lines $\mathcal{X}'$ with relatively large $d(\mathcal{X}')$. It should be possible to make direct sum stabilizer codes with good parameters from this geometrical construction. It would be of great interest if one could construct codes with parameters for which stabilizer codes could feasibly exist but none are known to exist. 
\end{rprob}

\section{Stabilizer codes for larger alphabets} \label{sectionquDit}

\subsection{The higher-dimensional Pauli group}
When a quantum system has $D$ levels we speak of a quDit.
In this section, we will consider quantum codes over such larger subsystems.
Consequently, these codes are subspaces of the Hilbert space $(\mathbb C^D)^{\ot n}$.

We will consider $({\mathbb C}^q)^{\otimes n}$, where $q=p^h$, is the power of a prime $p$. 
The restriction to prime powers allows us to use the structure of the finite field for their construction. 
In the case when $D$ is not a prime power, one can use the ring ${\mathbb Z}/{D\mathbb Z}$,
but then most of the constructions that we will consider here will not work.

We label the coordinates of ${\mathbb C}^q$ with elements of ${\mathbb F}_q$, where ${\mathbb F}_q$ denotes the finite field with $q$ elements. In this way, a basis for the space of endomorphisms of ${\mathbb C}^q$ can be indexed by the elements of ${\mathbb F}_q \times {\mathbb F}_q$.

For each $a \in {\mathbb F}_q$, we define a $q \times q$ matrix $X(a)$ to be matrix obtained from from the linear map which permutes the coordinates of ${\mathbb C}^q$ by adding $a$ to the index. 

In other words, with basis $\{ \ket x \ | \ x \in {\mathbb F}_q\}$ of ${\mathbb C}$,
$$
X(a)\ket x=\ket{x+a}.
$$
For example, if $q=3$ and the elements of ${\mathbb F}_q$ are $\{0,1,2\}$ then
$$
X(0)=\left( \begin{array}{ccc} 1 & 0 & 0 \\ 0 & 1 & 0 \\ 0 & 0 & 1 \end{array} \right),\ 
X(1)=\left( \begin{array}{ccc} 0 & 0 & 1 \\ 1 & 0 & 0 \\ 0 & 1 & 0 \end{array} \right) \ \mathrm{and} \ 
X(2)=\left( \begin{array}{ccc} 0 & 1 & 0 \\ 0 & 0 & 1 \\ 1 & 0 & 0 \end{array} \right)
.
$$

For each $b \in {\mathbb F}_q$, we define a $q \times q$ matrix $Z(b)$ to be the diagonal matrix whose $i$-th diagonal entry is $w^{\mathrm{tr}(ib)}$.  Here,
$w=e^{2\pi i/p}$ is a primitive $p$-th root of unity and 
$\mathrm{tr}$ is the trace map from ${\mathbb F}_q$ to its prime subfield ${\mathbb F}_p$,
$$
\tr(a) = \sum_{j=0}^{h-1} a^{p^j}\,.
$$
As in the previous case, if we take say $q=3$ then
$$
Z(0)=\left( \begin{array}{ccc} 1 & 0 & 0 \\ 0 & 1 & 0 \\ 0 & 0 & 1 \end{array} \right),\ 
Z(1)=\left( \begin{array}{ccc} 1 & 0 & 0 \\ 0 & \omega & 0 \\ 0 & 0 & \omega^2 \end{array} \right)\ \mathrm{and} \ 
Z(2)=\left( \begin{array}{ccc} 1 & 0 & 0 \\ 0 &  \omega^2 & 0 \\ 0 & 0 &  \omega \end{array} \right),
$$
where $\omega$ is a primitive complex third root of unity. Recall, that the rows and columns of the matrix are indexed by elements of ${\mathbb F}_q$, so $i \in {\mathbb F}_q$. Thus,
$$
Z(b)\ket x=\omega^{\mathrm{tr}(xb)}\ket x.
$$

We define the Pauli group for $q$ odd as
$$
\mathcal{P}_1=\{ \omega^cX(a)Z(b) \ | \ a,b \in {\mathbb F}_q,\ c \in {\mathbb Z}/p{\mathbb Z}\}
$$
and for $q$ even, that is when $p=2$, as
$$
\mathcal{P}_1=\{ i^f \omega^c X(a)Z(b) \ | \ a,b \in {\mathbb F}_q,\ c \in {\mathbb Z}/2{\mathbb Z},\ f \in {\mathbb Z}/2{\mathbb Z}\}.
$$

The reason that we accommodate this slightly larger group for $q$ even is due to Lemma~\ref{rthpower} below. One can check that this definition coincides with our definition of the Pauli group for $q=2$.

More generally, we define the group of Pauli operators on $({\mathbb C}^q)^{\otimes n}$ to be the $n$-fold direct product 
$\mathcal{P}_n = \mathcal{P}_1 \times \dots \times \mathcal{P}_1$ ($n$ times). Thus

$$
\mathcal{P}_n=\{\ \sigma_1\otimes \cdots \otimes \sigma_n \ | \ \sigma_j \in \mathcal{P}_1\}.
$$
The size of $\mathcal{P}_n$ is $pq^{2n}$ for $q$ odd and $4q^{2n}$ for $q$ even.

The weight of an element $c\sigma_1\otimes \cdots \otimes \sigma_n $, where $\sigma_i=X(a_i)Z(b_i)$, is the number of $i \in \{1,\ldots,n\}$ such that $\sigma_i \neq X(0)Z(0)$.

\begin{lemma} \label{XZswitch}
For all $a,b \in {\mathbb F}_q^n$,
$$
\omega^{\mathrm{tr}(a\cdot b)}X(a)Z(b)= Z(b)X(a).
$$
\end{lemma}

\begin{proof}
We have
$$
X(a)Z(b)\ket x=\omega^{\mathrm{tr}(b \cdot x)}X(a) \ket x=\omega^{\mathrm{tr}(b \cdot x)} \ket{x+a}.
$$
Meanwhile,
$$
Z(b)X(a)\ket x=Z(b) \ket{x+a}=\omega^{\mathrm{tr}(b \cdot (x+a))}\ket {x+a}.
$$
\end{proof}

The following lemma implies that non-identity elements of the Pauli group have order $p$, for $q$ odd. Note that for $q$ even this is not the case; there are elements of order four. However, we extend the Pauli group as above (defining $\sigma_y=i\sigma_x\sigma_z$) and in this way we introduce more elements of order two. We do this so that we have more options for $M_i$ in our set of pairwise commuting operators which will generate the abelian subgroup $S$. \footnote{This was overlooked in the seminal paper of Ketkar et. al. \cite{KKKS2006} on stabilizer codes over finite fields. They do not accommodate the larger Pauli group when $q$ is even, or include any version of Lemma~\ref{rthpower}. However, this larger group is necessary for all the examples of qubit stabiliser codes we have included here.}

\begin{lemma} \label{rthpower}
For all $a,b \in {\mathbb F}_q^n$ and $r \in {\mathbb N}$,
$$
(X(a)Z(b))^r= \omega^{{r \choose 2}\mathrm{tr}(a\cdot b)}X(a)^rZ(b)^r.
$$
\end{lemma}

\begin{proof}
By induction on $r$, we have
$$
(X(a)Z(b))^r=(X(a)Z(b))^{r-1}X(a)Z(b)
$$
$$
= \omega^{{r-1 \choose 2}\mathrm{tr}(a\cdot b)}X(a)^{r-1}Z(b)^{r-1}X(a)Z(b).
$$
By Lemma~\ref{XZswitch}, this is equal to
$$
\omega^{{r-1 \choose 2}\mathrm{tr}(a\cdot b)}X(a)^{r-1}\omega^{(r-1)\mathrm{tr}(a\cdot b)}X(a)Z(b)^{r-1}Z(b)
= \omega^{{r \choose 2}\mathrm{tr}(a\cdot b)}X(a)^rZ(b)^r.
$$
\end{proof}

As in the case of qubit codes, we will again be looking to construct stabilizer codes and for this reason it will be of interest to know when elements $M,N \in \mathcal{P}_n$ commute or not. For this reason the following lemma is fundamental.

\begin{lemma} \label{MNcommuteDit}
For all $a,b,a',b' \in {\mathbb F}_q^n$,
$$
X(a)Z(b)X(a')Z(b')= \omega^{\mathrm{tr}(a'\cdot b-b\cdot a')}X(a')Z(b')X(a)Z(b).
$$
\end{lemma}

\begin{proof}
$X(a)$ and $X(a')$ commute, likewise $Z(b)$ and $Z(b')$, so the lemma follows from Lemma~\ref{XZswitch}.
\end{proof}

\subsection{Error detection and correction}

As in the case of qubit codes it suffices to consider errors from the group $\mathcal{P}_n$ of Pauli-errors which are unitary operators of the form
$$
E=\sigma_1 \otimes \cdots \otimes \sigma_n
$$
where $\sigma_i =X(a)Z(b)$, for some $a,b \in {\mathbb F}_q$.

Let $Q$ be a quantum error correcting code of ${(\mathbb C^q)}^{\otimes n}$, i.e. a subspace of ${(\mathbb C^q)}^{\otimes n}$.

Then again, as in the case of qubit codes, $Q$ detects an error $E \in \mathcal{P}$ if for all $\ket \phi, \ket \psi \in Q$ with $\bra{\phi}\ket{\psi}=0$, we have that
$$
\bra \phi  E \ket \psi =0\,,
$$
and
$$
\bra \phi  E \ket \phi =c_E\,,
$$
for some constant $c_E$ which depends only on $E$.

A quantum code $Q$ has minimum distance $d$ if one can detect Pauli-errors with up to $d-1$ non-identity matrices and correct Pauli-errors with up to $\lfloor\frac{d-1}{2}\rfloor$ non-identity matrices.

We say that a quantum code of $({\mathbb C}^{q})^{\otimes n}$ of dimension $K$ and minimum distance $d$ is a $(\!(n,K,d)\!)_q$ code. 
If the code has dimension $K=q^k$ then we say that the code is a $[\![n,K,d]\!]_q$ code.
Note that some authors reserve the latter notation $[\![n,K,d]\!]_q$ for stabilizer codes only.

\subsection{Stabilizer codes}
 
A {\em stabilizer code} is the intersection of the eigenspaces with eigenvalue one of the elements of an abelian subgroup $S$ of $\mathcal{P}_n$. As before, we denote the code by $Q(S)$. We insist that $\lambda \one \not\in S$ whenever $\lambda \neq 1$, since otherwise $Q(S)$ is trivial.

As in the qubit case, a stabilizer code $Q(S)$ with stabilizer $S$ can detect all Pauli-errors that are scalar multiples of elements in $S$ or that do not commute with some element of $S$. We denote by $\mathrm{Centraliser}(S)$, the elements of ${\mathcal P}_n$ that commute with all elements of $S$. A non-detectable Pauli-error must be in $\mathrm{Centraliser}(S)$. 

Commuting elements are characterised as follows. 

By Lemma~\ref{MNcommuteDit}, two elements $M=\omega^c X(a)Z(b)$ and $N = \omega^{c'}X(a')Z(b')$ satisfy
$$
MN= \omega^{\mathrm{tr}(b \cdot a'-b' \cdot a)}MN.
$$
Therefore, $M$ and $N$ commute if and only if the trace symplectic form 
\begin{equation} \label{sform}
 \mathrm{tr}(b \cdot a'-b'\cdot a)
\end{equation}
is zero. 

As in the case for qubit codes, we introduce the map $\tau$ which maps elements of $\mathcal{P}_n$ to ${\mathbb F}_q^{2n}$ by
$$
\tau(X(a)Z(b))=(a | b).
$$
For elements $u,w \in {\mathbb F}_q^{2n}$, the trace symplectic form is
\begin{equation} \label{sform2}
(u,w)_a=\sum_{j=1}^{n}  \mathrm{tr} (u_j w_{j+n}-w_ju_{j+n}).
\end{equation}
Then with $u=(a|b)$ and $w=(a'|b')$, this is the trace symplectic form (\ref{sform}).

\subsection{Stabiliser codes as additive codes over ${\mathbb F}_q$}

Let $\tau$ be the map that maps $c X(a)Z(b)$ to $(a|b) \in {\mathbb F}_q^{2n}$.

The group $S$ is mapped to an additive code $C=\tau(S)$. The symplectic weight of  $(a|b) \in {\mathbb F}_q^{2n}$ is the number of $i \in \{1,\ldots,n\}$ such that $(a_i,b_i) \neq (0,0)$. Thus, an element $c X(a)Z(b)$ of weight $w$ is mapped to a vector of symplectic weight $w$.

The elements of $\mathrm{Centraliser}(S)$ are mapped to the dual code of $C$, namely
$$
C^{\perp_a} = \{ w\in {\mathbb F}_q^{2n} \ | \ (u,w)_a=0,\ \mathrm{for} \ \mathrm{all} \ u\in C\}\,.
$$
Here the dual $\perp_a$  is taken with respect to the trace symplectic form (\ref{sform2}).

We have the following important theorem.

\begin{theorem}  \label{quDitstabthm}
An $(\!(n, K, d)\!)_q$ stabilizer code exists if and only if there exists an additive code $C \leqslant {\mathbb F}_q^{2n}$ of size $|C| = q^n/K$ such that $C \leqslant C^{\perp_a}$. If $K \neq 1$ then $d$ is the minimum symplectic weight of an element of $C^{\perp_a} \setminus C$, otherwise $d$ is the minimum symplectic weight of an element of $C^{\perp_a}=C$.
\end{theorem}

\begin{proof}

Let $S$ be an abelian subgroup of $\mathcal P_n$ not containing non-trivial multiples of the identity.
Let $Q(S)$ be the corresponding $(\!(n, K, d)\!)_q$ stabilizer code and let
$$
P=\frac{1}{|S|} \sum_{M \in S} M.
$$
Then, as in Lemma~\ref{sumthemall}, $P$ is the orthogonal projection onto $Q(S)$. For any element $M=X(a)Z(b)$ we have that $M^\dag M=\one$, so $M \in S$ if and only if $M^\dag \in S$. Hence, $P^\dag=P$.

Thus, since $P$ is Hermitian and $P^2=P$, the dimension of its image $Q(S)$ is equal to the trace of $P$.
Since $\tr(M)=0$ for all $M \in \mathcal{P}_n$, $M \neq \one$ and $\mathrm{tr}(\one)=q^n$,
one has $\tr(P) = q^n/|S|$ and so $|S|=q^n/K$, since $\dim Q(S)=K$.

 We note that $C=\tau(S)$ is an additive code since $S$ is an abelian subgroup and has size $|S|=q^n/K$. 
Since $\tau(\mathrm{Centraliser}(S))=C^{\perp_a}$, we have $C \leqslant C^{\perp_a}$. 

For $K \neq 1$, the minimum symplectic weight of any element of $C^{\perp_a} \setminus C$ is $d$, 
 since the minimum distance of $Q(S)$ is the minimum weight of the Pauli operators in 
 $\mathrm{Centraliser}(S) \setminus S$. As in the qubit case, if $K=1$ then we define the minimum distance of $Q(S)$ to be the minimum weight of the Pauli operators in 
 $\mathrm{Centraliser}(S) =S$, which is equal to the minimum symplectic weight of any element of $C^{\perp_a} =C$

The backwards implication is similar. Let $S=\tau^{-1}(C)$ and define the stabilizer code to be $Q(S)$. Then the dimension follows as above. If $K \neq 1$ then the minimum distance of $Q(S)$ corresponds as above to the minimum symplectic weight of an element of $C^{\perp_a} \setminus C$, since $\mathrm{Centraliser}(S)$ is equal to $\tau^{-1}(C^{\perp_a})$
up to a scalar factor. If $K=1$ then the minimum distance of $Q(S)$ corresponds to the minimum non-zero symplectic weight of the elements of $C^{\perp_a} = C$.
\end{proof}

\subsection{Constructions} \label{quDitconstructions}

The following theorem is known as the Calderbank-Shor-Steane construction. 
The $\perp$ refers to the standard inner product on ${\mathbb F}_q^n$ given by
$$
u\cdot v=u_1v_1+\cdots u_nv_n.
$$

\begin{theorem} \label{CSS}
Suppose there are linear codes $C_1$ and $C_2$ with parameters $[n,k_1,d_1]_q$ and $[n,k_2,d_2]_q$, 
with the property that $C_1^{\perp} \leqslant C_2$. 
Then there is a $[\![n,k_1+k_2-n,d]\!]_q$ code, where $d$ is the minimum weight of the elements in 
$(C_1 \setminus C_2^{\perp})\cup (C_2 \setminus C_1^{\perp})$ if $k_1+k_2\neq n$ and $d$ is the minimum non-zero weight of the elements in 
$C_1 \cup C_2$ if $k_1+k_2=n$.
\end{theorem}

\begin{proof}
Let $C=C_1^{\perp} \times C_2^{\perp} \leqslant {\mathbb F}_q^{2n}$. Then $C$ is a linear code over ${\mathbb F}_q$ and for all $v=(v_1|v_2)$ and $w=(w_1|w_2)$ in $C$,
$$
(v,w)_a=\mathrm{tr}(v_1\cdot w_2-v_2\cdot w_1)=\mathrm{tr}(0-0)=0\,.
$$
In the above the first term vanishes since $v_1 \in C_1^{\perp} \leqslant C_2$ and $w_2\in C_2^{\perp}$. 
Likewise, the second term vanishes since $v_2 \in C_2^{\perp}$ and $w_1 \in C_1^{\perp} \leqslant C_2$.

Hence, $C \leqslant C^{{\perp}_a}$ and Theorem~\ref{quDitstabthm} applies.

To determine the minimum distance first note that $C^{{\perp}_a} \geqslant C_2 \times C_1$, since for all 
$v=(v_1|v_2) \in C_1^{\perp} \times C_2^{\perp}$ and 
$w=(w_2|w_1) \in C_2 \times C_1$,
$$
(v,w)_a=\mathrm{tr}(v_1\cdot w_1-v_2\cdot w_2)=\mathrm{tr}(0-0)=0.
$$
The dimension of $C_2 \times C_1$ is $k_1+k_2$ and the dimension of  $C^{{\perp}_a}$ is $2n-(n-k_1)-(n-k_2)=k_1+k_2$, so
$$
C^{{\perp}_a} =C_2 \times C_1.
$$
Thus, by Theorem~\ref{quDitstabthm}, if  $k_1+k_2\neq n$ then the minimum distance of the stabilizer code $\tau^{-1}(C)$ is the minimum weight of the elements in $(C_1 \setminus C_2^{\perp})\cup (C_2 \setminus C_1^{\perp})$. If $k_1+k_2=n$ then the minimum distance of the stabilizer code $\tau^{-1}(C)$ is the minimum non-zero weight of the elements in $C_2 \times C_1=C_1^{\perp} \times C_2^{\perp}$, which is equal to the minimum non-zero weight of the elements in $C_1 \cup C_2=C_1^{\perp}\cup C_2^{\perp}$.
\end{proof}

\begin{example}
The ternary extended Golay code $C_1$ is a $[12,6,6]_3$ code for which $C_1=C_1^{\perp}$. Applying Theorem~\ref{CSS}, this implies there is a $[\![12,0,6]\!]_3$ quantum stabilizer code.

The code $C_1$ has a generator matrix
$$
\mathrm{G}=\left(\begin{array}{cccccccccccc}
1 & 0 & 2 & 1 & 2 & 2 & 0 & 0 & 0 & 0 & 0 & 1\\
0 & 1 & 0 & 2 & 1 & 2 & 2 & 0 & 0 & 0 & 0 & 1\\
0 & 0 & 1 & 0 & 2 & 1 & 2 & 2 & 0 & 0 & 0 & 1\\
0 & 0 & 0 & 1 & 0 & 2 & 1 & 2 & 2 & 0 & 0 & 1\\
0 & 0 & 0 & 0 & 1 & 0 & 2 & 1 & 2 & 2 & 0 & 1\\
0 & 0 & 0 & 0 & 0 & 1 & 0 & 2 & 1 & 2 & 2  & 1\\
\end{array} \right)
$$
so $C=C_1 \times C_1$ has generator matrix, a $12 \times 24$ matrix
$$
\left(\begin{array}{c|c} 0 & \mathrm{G} \\ \hline \mathrm{G} & 0 \end{array} \right).
$$
The 12 Pauli operators generating the stabilizer group $S$ are
\begin{small}
$$
\left(\begin{array}{cccccccccccc}
Z(1) & 1 & Z(2) & Z(1) & Z(2) & Z(2) & 1 & 1 & 1 & 1 & 1 & Z(1)\\
1 & Z(1) & 1 & Z(2) & Z(1) & Z(2) & Z(2) & 1 & 1 & 1 & 1  & Z(1)\\
1 & 1 & Z(1) & 1 & Z(2) & Z(1) & Z(2) & Z(2) & 1 & 1 & 1  & Z(1)\\
1 & 1 & 1 & Z(1) & 1 & Z(2) & Z(1) & Z(2) & Z(2) & 1 & 1  & Z(1)\\
1 & 1 & 1 & 1 & Z(1) & 1 & Z(2) & Z(1) & Z(2) & Z(2) & 1  & Z(1)\\
1 & 1 & 1 & 1 & 1 & Z(1) & 1 & Z(2) & Z(1) & Z(2) & Z(2) & Z(1)\\
X(1) & 1 & X(2) & X(1) & X(2) & X(2) & 1 & 1 & 1 & 1 & 1 & X(1)\\
1 & X(1) & 1 & X(2) & X(1) & X(2) & X(2) & 1 & 1 & 1 & 1 & X(1)\\
1 & 1 & X(1) & 1 & X(2) & X(1) & X(2) & X(2) & 1 & 1 & 1  & X(1)\\
1 & 1 & 1 & X(1) & 1 & X(2) & X(1) & X(2) & X(2) & 1 & 1  & X(1)\\
1 & 1 & 1 & 1 & X(1) & 1 & X(2) & X(1) & X(2) & X(2) & 1  & X(1)\\
1 & 1 & 1 & 1 & 1 & X(1) & 1 & X(2) & X(1) & X(2) & X(2)  & X(1)\\
\end{array} \right).
$$
\end{small}
\end{example}

The next construction is called the ${\mathbb F}_{q^2}$ trick (for qubit codes this is the ${\mathbb F}_4$ trick). It's not really a trick at all but it is a quick and effective way to construct quantum codes. These codes are a very special type of stabilizer code in which we impose more structure on the additive code $C$.

For any two vectors $u,v$ in ${\mathbb F}_{q^2}^n$, we define the Hermitian form
\begin{equation} \label{aform}
u \circ v= u^q \cdot v
\end{equation}
and for a ${\mathbb F}_{q^2}$-linear code $E$ we define
$$
E^{\perp_h}=\{ u \in {\mathbb F}_{q^2}^n \ | \ u \circ v=0,\ \mathrm{for} \ \mathrm{all} \ v \in E\}.
$$

\begin{theorem} \label{F4trick}
If there exists a linear $[n,n-k,d]_{q^2}$ code $D$ for which $D^{{\perp}_h} \leqslant D$ then there is a $[\![n,n-2k, \geqslant d]\!]_q$ stabilizer code.
\end{theorem}

\begin{proof}
The code $D^{{\perp}_h}$ is a $[n,k,d']_{q^2}$ code for some $d'$.

Fix a basis $\{e,e^q\}$ for ${\mathbb F}_{q^2}$ over ${\mathbb F}_q$, where $e^{2q} \neq e^2$. 

Let $\theta$ be the map from  ${\mathbb F}_{q^2}^n$ to  ${\mathbb F}_{q}^{2n}$ defined by
$$
\theta((a_1e+b_1e^q,\ldots, a_ne+b_ne^q))=(a_1,\ldots a_n  | b_1,\ldots,b_n)
$$
Let $C=\theta(D^{{\perp}_h})$, a $2k$-dimensional linear code over ${\mathbb F}_q$ of length $2n$.

For $u \in D^{{\perp}_h}$ and $u' \in D$,
$$
0=u^q \cdot u'= \sum_{i=1}^n (a_ie+b_ie^q)^q(a_i'e+b_i'e^q).
$$
This implies
$$
0=\sum_{i=1}^n (a_i'b_ie^2+b_i'a_ie^{2q}+(a_ia_i'+b_ib_i')e^{q+1}).
$$
Applying the $x \mapsto x^q$ map, we get
$$
0=\sum_{i=1}^n (a_i'b_ie^{2q}+b_i'a_ie^{2}+(a_ia_i'+b_ib_i')e^{q+1}).
$$
Subtracting the last two equations,
$$
0=(e^{2q}-e^{2})\sum_{i=1}^n (a_ib_i'-b_ia_i').
$$
Hence,
$$
(\theta(u),\theta(u'))_a=0,
$$
and so $\theta(D) \leqslant C^{{\perp}_a}$. Since $|D|=|C^{{\perp}_a}|=q^{2(n-k)}$, we have that $\theta(D)= C^{{\perp}_a}$.

Moreover, $C=\theta(D^{{\perp}_h})$ and $D^{{\perp}_h} \leqslant D$, so $C \leqslant C^{{\perp}_a}$. The symplectic weight of an element of $\theta (u)$ is equal to the weight of $u$, so the minimum symplectic weight of $C^{{\perp}_a} \setminus C$ is the minimum weight of $D \setminus D^{{\perp}_h}$.

The theorem follows from Theorem~\ref{quDitstabthm}.
\end{proof}

We will use the construction of Theorem~\ref{F4trick} to obtain quantum MDS codes in the next section.

\begin{rprob}
If $k$ is small enough one can multiply the columns of a generator matrix for $D^{{\perp}_h}$ with non-zero scalars to obtain an equivalent code for which $D^{{\perp}_h} \leqslant D$ holds. It would be interesting to calculate the combinatorial threshold for codes when this can always be done and then deduce properties of codes which surpass this threshold. 
\end{rprob}

\subsection{The geometry of quqit codes}

In the case $q=p^h$, Theorem~\ref{quDitstabthm} implies that the existence of a $(\!(n,q^n/p^r,d)\!)_q$ stabilizer code $Q(S)$ is equivalent to the existence of an additive code $C \leqslant C^{\perp_a}$ of length~$2n$, such that $C$ is generated by $r$ vectors of ${\mathbb F}_q^{2n}$ that are linearly independent over ${\mathbb F}_p$.
Thus, the code $C$ is generated by a $r \times 2n$ matrix $\mathrm{G}(S)$ over ${\mathbb F}_p$ and its columns are vectors in ${\mathbb F}_q^r$. 
We have seen in Section~\ref{additivecodes} that when $h>1$, we should consider those columns as 
subspaces of $\mathrm{PG}(r-1,p)$ and not as points of $\mathrm{PG}(r-1,q)$.

Let $x_i$ be the $i$-th column of the matrix $\mathrm{G}(S)$ and 
let $e$ be an element of ${\mathbb F}_q$ with the property that $\{1,e,e^2,\ldots,e^{{h-1}}\}$ 
is a basis for ${\mathbb F}_q$ over ${\mathbb F}_p$. 

Then there are vector $x_{i,j} \in {\mathbb F}_p^r$ such that
$$
x_i=\sum_{j=0}^{h-1} x_{i,j}e^{j}.
$$

Let $\ell_i$ be the subspace 
\begin{equation} \label{ijdef}
\ell_i=\langle x_{i,0},\ldots,x_{i,h-1},x_{i+n,0},\ldots,x_{i+n,h-1} \rangle,
\end{equation}
as a subspace of $\mathrm{PG}(r-1,p)$.

The following lemma can be considered as a generalisation of Lemma~\ref{theyrelines}

\begin{lemma} \label{theyre2h-1spaces}
The subspace $\ell_i$ is a $(2h-1)$-dimensional subspace for all $i=1,\ldots,n$ if and only if the minimum non-zero weight of $\mathrm{Centraliser}(S)$ is at least two.
\end{lemma}

\begin{proof}
Suppose that $\ell_i$ is a $(2h-1)$-dimensional subspace for all $i=1,\ldots,n$ and that $E\in \mathrm{Centraliser}(S)$ has weight one. Suppose that $E$ has a $X(a)Z(b) \neq X(0)Z(0)$ in its $i$-th position, Consider any $M \in S$ and suppose that in the $i$-th coordinate $M$ has the Pauli matrix $X(a')(Z(b')$. Since $M$ and $E$ commute, 
$$
\mathrm{tr}(a'b-b'a)=0.
$$
Thus, $(a',b')$ is in the kernel of the linear (over ${\mathbb F}_p$) form
$$
\mathrm{tr}(bX-aY).
$$
The kernel of a linear form is a hyperplane of $\mathrm{PG}(2h-1,p)$, so $\ell_i$ has dimension at most $2h-2$, a contradiction.

Suppose that the minimum non-zero weight of $\mathrm{Centraliser}(S)$ is at least two and that $\ell_i$ is not a $(2h-1)$-dimensional subspace for some $i=1,\ldots,n$. Since $\ell_i$ does not span the whole of $\mathrm{PG}(2h-1,p)$, there is an element $(a,b) \in {\mathbb F}_q^2$ such that 
$$
\mathrm{tr}(a'b-b'a)=0,
$$
for all $X(a')Z(b')$ occurring in the $i$-th position of some $M \in S$. This implies that the Pauli operator of weight one $E$ with a $X(a)Z(b)$ commutes with all $M\in S$, contradicting the fact that the minimum non-zero weight of $\mathrm{Centraliser}(S)$ is at least two.
\end{proof}

Thus, by Lemma~\ref{theyre2h-1spaces}, the geometry of the stabilizer code $Q(S)$ for which the minimum non-zero weight of $\mathrm{Centraliser}(S)$ is at least two, is given by a set $\mathcal{X}$ of $(2h-1)$-dimensional subspaces of $\mathrm{PG}(r-1,p)$ of size $n$. The following lemma allows us to deduce the minimum distance of $Q(S)$, at least in the case that $Q(S)$ is pure.

\begin{lemma} \label{wdependent}
There are $w$ dependent points incident with distinct subspaces of $\mathcal{X}$ if and only if there is an element of $\mathrm{Centraliser}(S)$ of weight $w$.
\end{lemma}

\begin{proof}
Suppose that there is an element in $\mathrm{Centraliser}(S)$ of weight $w$. Then the image under $\tau$ of this element is a vector $v \in C^{\perp_a}$ with symplectic weight $w$. Let $D$ be the support of $v$ restricted to the first $n$ coordinates. As before, let $x_i$ be the $i$-th column of the matrix $\mathrm{G}(S)$ and define $x_{ij}$ as in (\ref{ijdef}). Since $v \in C^{\perp_a}$,
$$
\sum_{i\in D} \mathrm{tr}(v_{i+n}x_i-x_{i+n}v_i)=0.
$$
This implies
$$
 \sum_{i\in D} \sum_{j=0}^{h-1} (x_{ij} \mathrm{tr}(v_{i+n}e^j)-x_{i+n} \mathrm{tr}(v_{i}e^j))=0.
$$
The summand is a point of the subspace $\ell_i$ and there are $|D|=w$ such points. This proves the backwards implication.

Suppose there are $w$ dependent points incident with distinct subspaces of $\mathcal{X}$. Then there is a subset $D \subseteq \{1,\ldots,n\}$ of size $w$ and $\lambda_{i,j},\lambda_{i+n,j} \in {\mathbb F}_p$, such that
$$
\sum_{i\in D}  \sum_{j=0}^{h-1} ( \lambda_{i,j}x_{i,j}-\lambda_{i+n,j}x_{i+n,j})=0.
$$
Recall that
$$
x_i=\sum_{j=0}^{h-1} x_{i,j}e^{j}.
$$
Since $\ell_i$ is a $(2h-1)$-dimensional subspace, the points $x_j,x_j^p,\ldots,x_j^{p^{h-1}}$ are $h$ linearly independent points, which implies there are $\mu_{i,r} \in {\mathbb F}_q$ such that
$$
x_{i,j}=\sum_{r=0}^{h-1} \mu_{i,r} x_i^{p^r}.
$$
Since $x_{i,j} \in {\mathbb F}_p^{r}$, we have that $ \mu_{i,r}=\mu_i^{p^r}$, for some $\mu_i$. Substituting in the above gives,
$$
\sum_{i\in D}  \sum_{j=0}^{h-1} \sum_{r=0}^{h-1}  ( \lambda_{i,j}(\mu_{i}x_i)^{p^r}-\lambda_{i+n,j}(\mu_{i+n} x_{i+n})^{p^r})=0.
$$
Defining
$$
v_{i}=\sum_{j=0}^{h-1}  \lambda_{i,j}\mu_{i}
$$
this equation becomes
$$
\sum_{i\in D}  \mathrm{tr} (v_{i+n} x_i-v_{i}x_{i+n})=0.
$$
\end{proof}

The property that defines $\mathcal{X}$ as a quantum set of lines for $p=2$ does not carry over to the case $p \geqslant 3$. This is because we can scale any column of $\mathrm{G}$ by an element of ${\mathbb F}_q \setminus \{0,1\}$ and not alter the set of lines $\mathcal{X}$. This will alter the value of $(u,v)_a$, so the geometric interpretation of $C \leqslant C^{\perp_a}$ will not be so clean as in the qubit case. Moreover, it is difficult to deduce the pureness of the code directly from the geometry. To see this, suppose that $v \in C^{\perp_a}$ has symplectic support $D$ and for simplicity sake assume that $q$ is prime. Then 
$$
\sum_{i \in D} (v_{i+n}x_i-v_ix_{i+n})=0.
$$
Now, $v \in C$ if and only if there is an $a \in {\mathbb F}_p^r$ such that $v_i=a\cdot x_i$. This implies that the lines not incident with the dependent points are once again contained in a hyperplane, but we cannot deduce that the points of the dependencies are contained in the hyperplane $a \cdot X=0$. Indeed, the fact that
$$
a \cdot (v_{i+n}x_i-v_ix_{i+n})=0,
$$
implies that $(v_i,v_{i+n})=\lambda_i(x_i,x_{i+n})$ for some non-zero scalar $\lambda_i \in {\mathbb F}_q$. Since this $\lambda_i$ depends on $i$, we cannot deduce that $v_i=a\cdot x_i$ for all $i=1,\ldots,2n$.

However, this also means that when $p\geqslant 3$ we have some flexibility in choosing a basis for $\ell_i$ and this choice will affect whether $C \leqslant C^{\perp_a}$. Consider the set of $n$ $(2h-1)$-dimensional subspaces of $\mathrm{PG}(4n-1,p)$ associated with a pure $[\![n,n-4,3]\!]_q$ stabilizer code. By Lemma~\ref{wdependent}, these subspaces are pairwise skew. In geometrical language this is called a {\em partial spread}. To construct such a code, according to Theorem~\ref{F4trick}, it suffices to construct a $[n,n-2,3]_{q^2}$ linear code $D$ for which $D^{\perp_h} \leqslant D$. Such a code is has a generator  matrix
$$
\left(\begin{array}{cccc} x_1 & x_2 & \ldots & x_n \\ y_1 & y_2 & \ldots & y_n \\
\end{array} \right),
$$
where $x_iy_j \neq x_jy_i$ and 
\begin{equation} \label{allzero}
\sum_{i=1}^n x_i^{q+1}=\sum_{i=1}^n y_i^{q+1}=\sum_{i=1}^n x_i^{q}y_i=0.
\end{equation}
For any $n \leqslant q^2+1$ such a matrix can be found by scaling the first three columns so that the equation in (\ref{allzero}) are satisfied.

\begin{rprob}
The Glynn et al \cite{GGMG} manuscript developed the geometry of qubit stabilizer codes, introducing the concept of a quantum set of lines. This led them to prove Theorem~\ref{3GMakx}, which gives a beautiful geometric classification of qubit stabilizer codes. Here, we have generalised the concept of quantum set of lines to non-qubit stabilizer codes. Although we have seen that the existence of non-identity non-zero scalars means we cannot hope for such a clean geometric classification, one can certainly expect some geometric classification for larger $q$.
\end{rprob}

\section{Quantum MDS codes} \label{sectionMDS}

\subsection{Stabiliser MDS codes}

Let $C$ be a code of length $n$ and minimum distance $d$ over an alphabet of size $q$.
If we consider any $n-(d-1)$ coordinates then any two codewords must be different on these coordinates (if not the distance between them is at most $d-1$), so there are at most $q^{n-d+1}$ codewords in the code. This is the {\em Singleton bound}
$$
|C| \leqslant q^{n-d+1}.
$$
A code which attains the Singleton bound is called a {\it maximum distance separable code} or simply an MDS code. 

Recall that if $C$ is an additive code over ${\mathbb F}_q$, where $q=p^h$ for some prime $p$, then $C$ is linear over ${\mathbb F}_p$ and so necessarily $|C|=p^r$ for some $r$, see Section~\ref{additivecodes}. Thus, if $C$ is also an MDS code then $h$ divides $r$ and $|C|=q^k$, where $k=n-d+1$.

Theorem~\ref{quDitstabthm} states that an $[\![n, k, d]\!]_q$ stabilizer code exists if and only if there exists an additive code $C \leqslant {\mathbb F}_q^{2n}$ of size $|C| = q^{n-k}$ such that $C \leqslant C^{\perp_a}$ and the minimum symplectic weight of an element of $C^{\perp_a} \setminus C$ is $d$. Considering $C^{\perp_a}$ as a code over the alphabet ${\mathbb F}_q \times {\mathbb F}_q$, then $C^{\perp_a}$ has minimum weight $d$, so
$$
|C^{\perp_a}| \leqslant q^{2n-2d+2}.
$$
Since $|C| = q^{n-k}$ we have that $|C^{\perp_a}|=q^{n+k}$, which implies that for a $[\![n, k, d]\!]_q$ stabilizer code to exist, we must have the condition
$$
k \leqslant n-2(d-1).
$$
Compare this with the Singleton bound above
$$
k \leqslant n-(d-1),
$$
for codes of size $q^k$.

What is perhaps surprising is that this bound holds for all $[\![n,k,d]\!]_q$ quantum codes. The {\em quantum Singleton bound} states that
$$
n \geqslant k + 2(d-1)\,.
$$
Consequently, codes reaching equality are called {\em quantum maximum distance separable codes} or QMDS codes for short. We will prove this bound in Section~\ref{quantumsingleton}.

\subsection{Reed-Solomon codes}

The classical example of an MDS code is the following linear code over ${\mathbb F}_q$. Denote by $\{a_1,\ldots,a_{q}\}$ the elements of ${\mathbb F}_{q}$. The {\em Reed-Solomon code} is
$$
C=\{(f(a_1),\ldots,f(a_{q}),f_{k-1}) \ | \ f \in {\mathbb F}_{q}[X],\ \deg f \leqslant k-1\},
$$
where $f_{k-1}$ denotes the coefficient of $X^{k-1}$ in $f(X)$. If $k\leqslant q$ then each polynomial $f$ defines a different codeword, so the dimension of $C$ is $k$. A non-zero codeword has weight at least $n-k+1$, since a polynomial of degree at most $k-1$ has at most $k-1$ zeros. Lemma~\ref{minweight} then implies that the minimum distance $d=n-k+1$ and so the code is MDS.

We can use Theorem~\ref{F4trick} to construct quantum stabilizer codes from Reed-Solomon codes over ${\mathbb F}_{q^2}$, but only if we can scale the coordinates of $C$ so that $C \leqslant C^{\perp_h}$. Then $D=C^{\perp}_h$ is a $[n,n-k,k+1]_{q^2}$ linear MDS code with the property that $D^{\perp}_h \leqslant D$. Observe that replacing the $i$-th coordinate $f(a_i)$ by $\lambda_i f(a_i)$ does not alter the parameters of the code. Such a code is then called a {\em generalised Reed-Solomon code}. This can only be done for $k \leqslant q$, in which case we obtain a $[\![q^2+1,q^2+1-2k,k+1]\!]_q$ stabilizer code. For case $k=q$, one can check that the Reed-Solomon code
$$
\{(f(a_1),\ldots,f(a_{q^2}),f_{q-1}) \ | \ f \in {\mathbb F}_{q^2}[X],\ \deg f \leqslant q-1\},
$$
is contained in its Hermitian dual, so there is no need to scale in this case.

\subsection{Quantum Singleton bound} \label{quantumsingleton}

To prove the quantum Singleton bound we will need some technical tools.

\noindent 1. {\em Bloch decomposition.}
Let $\{e_i\}$ be a basis for the space of complex $D\times D$ matrices such that \(\tr(e_i^\dag e_j) = D \delta_{ij}\). 
For qubits, take for example the Pauli matrices.
Every one-quDit density matrix can then be expanded as
$$
\rho = \frac{1}{D} \sum_{i} \tr(e_i^\dag \rho) e_i,
$$
where we recall that the trace of a matrix is given by the sum of its diagonal elements, 
$\tr(M) = \sum_i m_{ii}$ for any square matrix $M = (m_{ij})$.

Consider now an $n$-partite system in the space $(\mathbb C^D)^{\otimes n}$. 
Denote by $\{E_\alpha\}$, with a multi-index $\alpha = (\alpha_1,\dots, \alpha_n)$, the matrix basis formed by tensor-products of the $e_i$'s
$$
E_{\alpha} = e_{\alpha_1} \ot \dots \ot e_{\alpha_n}.
$$
For tensor products, such as say $E \ot F$, one has $\tr(E \ot F) = \tr(E) \cdot \tr(F)$. In other words, the trace of a tensor product factorizes.
Consequently
$
\tr(E_\alpha^\dag E_\beta) = D^n \delta_{\alpha \beta},
$
and the matrix basis formed by $\{E_\alpha\}$ is orthogonal.

Denote by $\wt(E_\alpha)$ the number of non-identity terms in the tensor-decomposition,
and by $\supp(E_\alpha)$ the collection of sites where the non-identity terms act on.
Naturally, $\wt(E_\alpha) = |\supp(E_\alpha)|$.

We can expand an $n$-partite state as 
$$
\rho = \frac{1}{D^n} \sum_{E} \tr(E^\dag \rho) E\,.
$$
As above, we from now on omit the index $\alpha$ for readability. 
This is the Bloch decomposition of $\rho$.

\noindent 2. {\em Partial trace.}
Consider the linear function $\tr_j$ which maps
$$
\tr_j : e_{\alpha_1} \ot \dots \ot e_{\alpha_n} \quad \mapsto \quad \tr(e_{\alpha_j}) \cdot e_{\alpha_1} \ot \dots \ot e_{\alpha_{j-1}} \ot e_{\alpha_{j+1}} \ot \dots \ot e_{\alpha_n}\,.
$$
The function $\tr_j$ is called the {\em partial trace} and its action can be understood as that of removing 
the $j$-th tensor component. 

The partial trace does not depend on the basis. Its coordinate-free definition is the following:
Let $V$ and $W$ be two vector spaces and 
denote by $I_W$ the identity matrix on $W$.
The partial trace $tr_W$ is the unique operator, 
which for all $M$ acting on $V \otimes W$ and $N$ acting on $V$ satisfies
$$ 
\tr( M \cdot (N \otimes I_W)) = \tr( \tr_W(M) \cdot N)\,.
$$
Considering the Hilbert-Schmidt inner product $\langle M,N \rangle = \tr(M^\dagger N)$,
the partial trace can be seen as the adjoint to the map $V \to V \otimes I_W$.
Note that partial traces over different subsystems commute, $\tr_j \tr_i = \tr_i \tr_j$ and one has that
$$
\tr(M_1 \ot M_2 \ot \dots \ot M_n) = \tr(M_1) \tr(M_2) \cdots \tr(M_n)\,.
$$

\noindent 3. {\em Purification.} A density matrix $\rho$ on $\HH_A$ can always be diagonalized as
$$
\rho = \sum_{i=1}^{\dim(\HH_A)} \lambda_i \dyad{\lambda_i}_A,
$$
where $\{\ket{\lambda_i}_A\}$ is its set of eigenvectors and $\{\lambda_i\}$ is its set of corresponding eigenvalues.

The density matrix $\rho$ acting on some Hilbert space $\HH_A$ can always be 
represented as the reduction or marginal of 
a pure state on $\HH_A \ot \HH_B$ with $\dim(\HH_B) \geq \dim(\HH_A)$. 
This works as follows: choose an orthonormal basis $\{\ket{\lambda_i}^B\}$
for an arbitrary $\dim(\HH_A)$-dimensional subspace of $\HH_B$. We then write
$$
\ket{\phi} = \sum_{i=1}^{\dim(\HH_A)} \sqrt{\lambda_i} \ket{\lambda_i}_A \ot \ket{\lambda_i}^B\,.
$$
It can be checked that $\tr_B(\dyad{\phi}) = \rho$ and the state $\ket{\phi}$ is known as a {\em purification} of $\rho$.

\noindent 4. {\em Von Neumann entropy.}
Consider a classical probability distribution represented by a set of probabilities $p_i\geq 0$ with $\sum_i p_i = 1$. 
Its {\em Shannon entropy} is 
$$
S(p) = -\sum_i p_i \log(p_i)\,.
$$

We can introduce a similar quantity for quantum states. Given a density matrix $\rho$,
its {von Neumann entropy} is defined as
$$S(\rho) = - \tr \rho \log(\rho)\,.
$$
Such matrix functions of hermitian operators can be evaluated on their eigenvalues $\{\lambda_i\}$.
Then the von Neumann entropy evaluates as
$$S(\rho) = - \sum_i \lambda_i \log(\lambda_i)\,.
$$

Let us now write $S_A = S(\tr_B[\rho_{AB}])$ and so on. 
For a state $\rho$ on $\HH_A$ with purification $\ket{\phi} \in \HH_A \ot \HH_B$, we have
that 
$S_A = S_B$.

The von Neumann entropy satisfies {\em subadditivity} and {\em strong subadditivity},
\begin{align}
S_{AB} &\leq S_A + S_B\,, \nonumber\\
S_{ABC} + S_B & \leq S_{AB} + S_{BC}\,. \nonumber
\end{align}

We are now in position to prove the Quantum Singleton bound.

\begin{theorem}[Quantum Singleton bound]
Any $[\![n,k,d]\!]_q$ code with $k\geq 1$ satisfies
$$
n \geq k + 2(d-1)\,.
$$ 
\end{theorem}

\begin{proof}
The distance must be bounded by $2(d-1)<n$, as otherwise $n-(d-1) < (d-1)$ and we could recover the encoded state
from two disjoint subsystems, violating the no-cloning theorem.

Let \(\Pi_\QQ = \sum_{i=1}^{q^k} \dyad{v_i}\) be the projector onto the code space.
A purification with a reference system~\(R\) leads to
 \begin{equation}
  \ket{\psi_{QR}} = \frac{1}{\sqrt{q^k}}\sum_{i=1}^{q^k} \ket{v_i}\ot \ket{i_R}, \nonumber
 \end{equation} 
 where \(\ket{i_R}\) is any orthonormal basis for \(R\). 
 Let us partition the code into the three subsystems \(A,B,C\), such
 that \(|A|=|B|=d-1\) and $|C|=n-2(d-1)$.  Then \(S_R = \log(q^k)\).  
 As the code has distance~\(d\), any
 subsystem of size strictly smaller than $d$ cannot reveal anything
 about the reference system $R$: indeed the condition of \(\varrho_{RA} =
 \varrho_R \ot \varrho_A\) is known to be a necessary and sufficient condition
 for the subsystem $A$ to be correctable~\cite{NielsenChuang2000};
 this is also equivalent to
 \(S_{RA} = S_R + S_A\).  With the subadditivity of the von Neumann
 entropy this leads to
 \begin{align}
  S_R + S_A &= S_{RA} = S_{BC} \leq S_B + S_C\,,  \nonumber\\
  S_R + S_B &= S_{RB} = S_{AC} \leq S_A + S_C\,,\nonumber
 \end{align}
where we used that the entropies of complementary subsystems are equal for a pure state.
 The combination of the above two inequalities yields
 $$
 \log q^k= S_R \leq S_C\leq \log \dim (\HH_C) = \log q^{n-2(d-1)}.
 $$
\end{proof}

Similar to classical MDS codes, quantum MDS are, in a certain sense, extremal.
We have the following interesting properties:
\begin{enumerate}
 \item[(a)] If a $[\![n,n-2d+2,d]\!]$ quantum MDS code exists, then so do all $[\![n-s, n-2d+2+s, d-s]\!]$ codes for all $0\leq s \leq d$.
 \item[(b)] For every subset $S\subset \{1,\dots, n\}$ with $|S| \leq \frac{n+k}{2}$, we have that $\tr_{S^c}(P) \propto \one$, where $P$ is the orthogonal projection onto the quantum MDS code.
\end{enumerate}
Let us discuss these properties: a) states that QMDS codes form families of codes where $n+k$ is constant.
Within each family, only the member with the highest distance has to be determined,
as its descendants can be obtained by a partial trace: tracing out over a single particle,
one has $n\mapsto n-1$, $k \mapsto k+1$, $d\mapsto d-1$. This works because QMDS codes are {\em pure codes}, that is, 
all their $(d-1)$-party marginals are maximally mixed.
For general quantum codes, this method of making new codes from old is not necessarily possible.

\begin{figure}
\centering
 \includegraphics[width = 0.5\textwidth]{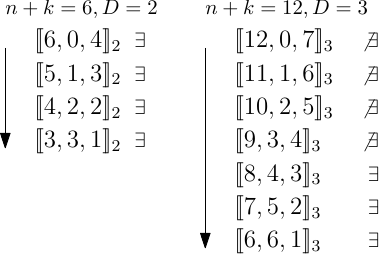}
 \caption{Two families of quantum MDS codes. Once the topmost existing parent code is known, 
 (here: $[\![6,0,4]\!]_2$ and $[\![8,4,3]\!]_3$), 
 its descendants can be obtained by partial traces.}
\end{figure}

Property (b) states that for all pure states $\ket{v}$ in the code, 
the marginals of size less than $d$ are maximally mixed. 
This implies that every vector in the code space shows maximal bipartite entanglement 
across any bipartition of $d-1$ vs. $n-d+1$ parties. Thus QMDS codes form subspaces that 
show high bipartite entanglent.
We relate this to similar property of classical MDS codes:
the parity check matrix $H$ of a classical $[n,k,d]$ code has the property that every set of $n-k$ columns are linearly independent.

A necessary condition for QMDS to exist is the following bound.
\begin{prop}[\cite{Huber2019}]
If there is a quantum MDS code with parameters $[\![n, n- 2d+2, d]\!]_q$ then
$$
 n \leqslant q^2 + d - 2\,.
$$
\end{prop}

This should be compared to the ``trivial'' upper bound for MDS codes. If there is a $(n,q^k,n-k+1)_q$ MDS code then
$$
n \leqslant q+k-1.
$$

The MDS conjecture states that if $4 \leqslant k \leqslant q$ and there is a $(n,q^k,n-k+1)_q$ MDS code then 
$$
n \leqslant q+1.
$$
This is known to hold for linear codes if $q$ is a prime, see \cite{Ball2012}.

For quantum MDS codes, the MDS conjecture states that if $5 \leqslant d \leqslant q^2-1$ and there is a linear $[\![n,n-2d+2,d]\!]_q$ MDS code then 
$$
n \leqslant q^2+1.
$$

Ketkar \cite[Corollary 65]{KKKS2006} claims that if the classical MDS conjecture holds for linear codes then quantum MDS conjecture holds for stabilizer codes. This is not the case. By Theorem~\ref{quDitstabthm} the existence of a stabilizer code is equivalent to the existence of an additive code, so \cite[Corollary 65]{KKKS2006} should state that the quantum MDS conjecture holds for stabilizer codes if the MDS conjecture holds for additive codes.

\begin{rprob}
Prove the MDS conjecture for linear codes with $q$ non-prime. 
\end{rprob}

\begin{rprob}
Prove the MDS conjecture for additive codes over ${\mathbb F}_q$, starting with $q={p^2}$ for some prime $p$.
\end{rprob}

\begin{rprob}
Find all inequalities that relate the von Neumann
entropies of the marginals of multipartite systems.
\end{rprob}

\begin{rprob}
Show that all QMDS codes are either stabilizer codes or the direct sum of stabilizer codes.
\end{rprob}

\section{Weight enumerators}

\subsection{MacWilliams identity for linear codes}

Let $C$ be an $[n,k,d]_q$ code and define $A_i$ to be the number of codewords of $C$ of weight $i$, i.e. the number of codewords of $C$ which have $i$ non-zero coordinates. Since the zero codeword is in $C$, $A_0=1$ and since the minimum distance is $d$, $A_i=0$ for all $i=1,\ldots,d-1$. Let $B_i$ denote the number of codewords of $C^{\perp}$ of weight $i$. The MacWilliam's identities relate the polynomials
$$
A(x,y)=\sum_{i=1}^n A_i x^{n-i}y^i
$$
and
$$
B(x,y)=\sum_{i=1}^n B_i x^{n-i}y^i.
$$
Specifically, we have that
$$
|C|B(x,y)=A(y+(q-1)x,y-x)
$$
and dually,
$$
|C^{\perp}|A(x,y)=B(y+(q-1)x,y-x).
$$
Let $\mathrm{G}$ be a $k\times n$ generator matrix for $C$ and let $\mathcal{X}$ be the set or multi-set of columns of $\mathrm{G}$, viewed as points of $\mathrm{PG}(k-1,q)$. In Section~\ref{geomlinear}, we saw that a non-zero codeword $u=a\mathrm{G}$ corresponds to a hyperplane $\pi_a$ of $\mathrm{PG}(k-1,q)$ and that $\pi_a=\pi_{\lambda a}$ for any $\lambda \in {\mathbb F}_q$. The number of points of $\mathcal{X}$ incident with the hyperplane $\pi_a$ is $n$ minus the weight of the codeword $u$. Thus, for $i \neq 0$, there are $A_i/(q-1)$ hyperplanes which are incident with $n-i$ points of $\mathcal{X}$.

\subsection{MacWilliams identity for quantum codes}

As for classical codes, weight enumerators can be defined for quantum codes, which again are useful to deduce the error-correcting properties of codes and to obtain bounds on their existence.

Let $Q$ be a quantum code and let $P$ be the orthogonal projection onto $Q$.
The weights of the primary and secondary {\em Shor-Laflamme enumerators} are
\begin{align}
 A_j &= \sum_{\wt(E) = j} \tr (E P ) \tr ( E^\dag P ) , \nonumber\\
 B_j &= \sum_{\wt(E) = j} \tr (E P E^\dag P) \nonumber ,
\end{align}
where the sum is over Pauli operators $E$ of weight $j$ and phase $1$.

The enumerator polynomials are given by
\begin{align}
 A(x,y) &= \sum_{j=0}^n A_j x^{n-j}y^j\,, &
 B(x,y) &= \sum_{j=0}^n B_j x^{n-j}y^j \,. \nonumber
\end{align}

\begin{lemma}
For a stabilizer code, $A_j$ is $q^{2n}/|S|^2$ times the number of elements in the stabilizer subgroup $S$ that have weight $j$. Similarly, $B_j$ is $q^{n}/|S|$ times the number of elements in the normaliser of $S$ of weight $j$.
\end{lemma}

\begin{proof}
By Lemma~\ref{sumthemall}, 
$$
P=\frac{1}{|S|} \sum_{M \in S} M.
$$
The map $\tr$ is linear and $\tr(M)=0$ unless $M=\one$ and $\tr(\one)=q^n$.

Hence, if $E \not\in S$,
$$
 \tr (E P ) \tr ( E^\dag P )=0
 $$
and if $E\in S$ then 
$$
 \tr (E P ) \tr ( E^\dag P )=q^{2n}/|S|^2.
 $$
 Thus, $A_j$ is $q^{2n}/|S|^2$ times the number of elements in the stabilizer subgroup $S$ that have weight $j$.
 
 We leave the result for $B_j$ as an exercise.

\end{proof}

The geometrical interpretation of $A_j$ for stabilizer codes is as follows. Suppose that $\mathcal{X}$ is a quantum set of lines in $\mathrm{PG}(n-k-1,q)$. Then $A_j$ is $(q-1)$ times number of hyperplanes containing $n-j$ lines of $\mathcal{X}$. 

The {\em quantum MacWilliams identity} states that
$$
q^n  B(x,y) = A (x + (q^2-1)y, x-y ),
$$
and respectively that
$$
q^n A(x,y) =  B (x + (q^2-1)y, x-y).
$$

Before proving the quantum MacWilliams identity, consider the following example.
\begin{example} (self-dual hexacode)
Consider the $[6,3,4]_4$ code $D$ generated by the matrix
$$
\left(\begin{array}{cccccc} 
1 & 0 & 0 & 1 & 1 & 1\\
0 & 1 & 0 & 1 & e & e^2 \\
0 & 0 & 1 & 1 & e^2 & e \\
\end{array} \right),
$$
where $e^2=e+1$. One can prove that the minimum distance is $4$ by checking that all $3 \times 3$ submatrices are non-singular. By verifying that the hermitian inner product (\ref{aform}) between any two rows is zero, one quickly concludes that $D=D^{\perp_h}$. Theorem~\ref{F4trick} implies that we can construct a $[\![6,0,4]\!]_2$ stabilizer code $Q(S)$ from $D$. By writing out the entries in the matrix over ${\mathbb F}_2$ and considering the ${\mathbb F}_2$ span we obtain the matrix $\mathrm{G}(S)$ for this quantum code.

 Consider the $[\![6,0,4]\!]_2$ code that can be constructed from the code $D$. 
 The code $\tau(S)$ is spanned by the generator matrix
$$
\mathrm{G}(S)= \left(\begin{array}{cccccc|cccccc}  
    1 & 0 & 0 & 1 & 1 & 1    &  0 & 0 & 0 & 0 & 0 & 0   \\
    0 & 0 & 0 & 0 & 0 & 0    & 1 & 0 & 0 & 1 & 1 & 1  \\
    0 & 1 & 0 & 1 & 0 & 1    & 0 & 0 & 0 & 0 & 1 & 1  \\
    0 & 0 & 0 & 0 & 1 & 1    & 0 & 1 & 0 & 1 & 1 & 0  \\
    0 & 0 & 1 & 1 & 1 & 0    & 0 & 0 & 0 & 0 & 1 & 1  \\
    0 & 0 & 0 & 0 & 1 & 1    & 0 & 0 & 1 & 1 & 0 & 1  
 \end{array}\right).
$$

Thus, the stabilizer subgroup has generators
$$
 \begin{array}{rccccccc}  
M_1 & = &   X &1&1   &X &X&X\\
M_2 & = &    Z & 1 &1   &Z&Z&Z\\
M_3 & = &    1 &X &1   &X&Z&Y \\
M_4 & = &    1&Z &1    &Z&Y &X \\
M_5 & = &    1 &1&X    &X &Y&Z \\
M_6 & = &    1 &1 &Z   &Z &X &Y
 \end{array}
$$
By Lemma~\ref{wdependent}, the quantum set of six lines $\mathcal{X}$ we get from the matrix $\mathrm{G}(S)$ has the property that any three lines of $\mathcal{X}$ span the whole space $\mathrm{PG}(5,2)$. Therefore, any two span a three-dimensional subspace which is contained in three hyperplanes which contain no further line of $\mathcal{X}$. Thus, there are $45$ hyperplanes which contain exactly two lines of $\mathcal{X}$. Let $\ell$ be a line of $\mathcal{X}$. There are 15 hyperplanes containing $\ell$, so counting pairs $(\ell, \pi)$ where $\ell \in \mathcal{X}$ and $\pi$ is a hyperplane containing $\ell$, we conclude that any hyperplane containing a line of $\mathcal{X}$ contains two lines of $\mathcal{X}$. 

Thus, we work out the weight distribution. For codes with $k=0$ (that is, pure states), both weight distributions coincide; this can be checked from the definition. From before, we have that $A_j$ is the $(q-1)$ times number of hyperplanes containing $n-j$ lines of $\mathcal{X}$. Thus, we have proved that the weight distribution for the quantum hexacode is
$$
(A_0, \dots, A_6) = (1,  0,  0,  0, 45,  0, 18).
$$
The corresponding enumerator polynomials are
$$
 A(x,y) = B(x,y) = x^6 + 45 x^2 y^4 + 18 y^6\,.
$$
This polynomial is indeed invariant under the quantum MacWilliams transform, since
$$
64 B(x,y)=  (x + 3y)^6 + 45 (x + 3y)^2 (x-y)^4 + 18(x-y)^6= 64(x^6 + 45 x^2 y^4 + 18 y^6 ).
$$
\end{example}

\begin{rprob}
 For stabilizer codes, $A_j$ and $B_j$ count the number of terms in the stabilizer $S$ and its normaliser $N(S)$ respectively;
 there is no such combinatorial interpretation for general quantum codes. Although $A_j$ can interpreted as the Hilbert-Schmidt norms of 
 the $j$-body correlations that appear in the code, we would like to determine what object $B_j$ is counting for non-stabilizer codes.
 \end{rprob}

We return to the proof of the quantum MacWilliams identity.

\begin{proof}[Quantum MacWilliams identity]
We will only state a proof sketch; the rather tedious combinatorial 
details can be found in~\cite{Rains1998, Huber2017}.

Let $S$ be a collection of subsystems and denote by $\tr_S$ the partial trace the systems in $S$.
Denote by $S^c$ the complement of $S$ in $\{1, \dots, n\}$.
Consider now how the partial trace $\tr_S$ followed by a ''padding`` with the identity acts on an operator $P$.
\begin{equation}\label{eq:ptrace1}
\tr_S(P) \ot \one_S
= \tr_S \Big( \frac{1}{q^n} \sum_{E} \tr(E^\dag P) E \Big) \ot \one_S
= \frac{1}{q^{n-|S|}} \sum_{\supp(E) \subseteq S^c} \tr(E^\dag P) E\,.
\end{equation}

It can be shown (c.f. Appendix A in Ref.~   \cite{Huber2017}) that this can also be written as
\begin{equation}\label{eq:ptrace2}
\tr_S(P) \ot \one_S = \int_{\substack{U(q^n) \text{ s.t.} \\ 
\supp(U) \subseteq S}} U P U^\dag dU = \frac{1}{q^{|S|}} \sum_{\supp(E) \subseteq S} E P E^\dag\,,
\end{equation} 
where the integration is over the unitarily invariant Haar measure of unitary matrices 
that act trivially on the subsystem $S^c$.
The second equality follows from the fact that any complete 
orthonormal matrix basis $\{E_\alpha\}$ containing the identity forms a unitary $1$-design~\footnote{
$t$-designs replace the integration over some compact group by a finite sum.
A unitary t-design is a set of unitaries $U_i$, $i=1,\dots, K$ acting on $\mathbb C^q$, such that
$\int_{U(D)} P_{t,t}(U) dU= \frac{1}{K} \sum_{i=1}^K P_{t,t}(U_i)$ holds for every homogeneous polynomial $P_{t,t}$
that has degree $t$ in the matrix elements of $U$ and degree $t$ in the matrix elements of $U^*$.}.

The quantum MacWilliams identity now essentially follows from equating Eqs.~\eqref{eq:ptrace1} 
and \eqref{eq:ptrace2}, summing over all subsystems of size $|S| = m$, multiplying by $P$, 
and taking the trace. This yields terms of the form $\sum \tr(E^\dag P)\tr(E P)$ and $\sum \tr(E^\dag P E P)$,
corresponding to the two types of weights $A_j$ and $B_j$.

Proceeding in this manner, Eq.~\eqref{eq:ptrace1} gives
\begin{align}
\sum_{|S|=m} \tr(\tr_S(P) \ot \one_S \cdot P) 
&= \sum_{|S|=m} \tr\Big( q^{m-n} \sum_{\supp(E) \subseteq S^c} \tr(E^\dag P) E \cdot P \Big)\nonumber\\
&= q^{m-n} \sum_{|S|=m}          \sum_{\supp(E) \subseteq S^c} \tr(E^\dag P) \tr\Big(E P\Big) \nonumber\\
&= q^{m-n} \sum_{j=0}^{n-m} \binom{n}{n-m} \binom{n-m}{j} \binom{n}{j}^{-1} A_j \nonumber \\
&= q^{m-n} \sum_{j=0}^{n-m} \binom{n-j}{m} A_j \nonumber\,.
\end{align}

Meanwhile, Eqs.~\eqref{eq:ptrace2} gives
\begin{align}
\sum_{|S|=m} \tr(\tr_S(P) \ot \one_S \cdot P)  
&=       \sum_{|S|=m} \tr\Big( q^{-m}  \sum_{\supp(E) \subseteq S} E^\dag P E \cdot P \Big) \nonumber\\
&= q^{-m}\sum_{|S|=m}          \sum_{\supp(E) \subseteq S} \tr(E^\dag P E P ) \nonumber\\
&= q^{-m}\sum_{j=0}^m         \binom{n}{m} \binom{m}{j} \binom{n}{j}^{-1} B_j \nonumber\\
&= q^{-m}\sum_{j=0}^m         \binom{n-j}{n-m} B_j \nonumber\,.
\end{align}
Thus for every operator $P$ and $0\leq m\leq n$ one has that
$$
 q^{m-n} \sum_{j=0}^{n-m} \binom{n-j}{m} A_j = q^{-m}\sum_{j=0}^m         \binom{n-j}{n-m} B_j \,.
$$

Using generating functions, in other words the weight enumerator polynomials $A(x,y)$ and $B(x,y)$, 
and Krawtchouk polynomials, this yields the MacWilliams identity
$$
q^n B(x,y) = A(x + (q^2-1)y,x-y).
$$
This ends the proof sketch.
\end{proof}

The enumerators and their weights have a couple of interesting properties:
Let $K = \dim(\mathrm{im} P)$.
\begin{itemize}
 \item[a)] The weights $A_j$ and $B_j$ are invariant under the local choice of basis and are so-called local unitary invariants (LU-invariants). That is,
 $$
    A_j(P) = A_j(P') \quad \text{and}\quad 
    B_j(P) = B_j(P') \,,
 $$
 where $P' = (U_1 \ot \dots \ot U_n)P (U_1^\dag \ot \dots \ot U_n^\dag)$ and $U_1,\dots, U_n$ are unitary $q\times q$ matrices.
            
 \item[b)] $A_0 = \dim(P)$ and $K B_j \geq A_j \geq 0$.
 \item[c)] A projection operator $P$ with $K = \dim(\mathrm{im}(P))$  
            is a code of distance $d$, if and only if it satisfies $K B_j = A_j$ for $0\leq j < d$.
 \item[d)] One can check that for codes with $K=1$, the enumerator polynomial is invariant under the quantum MacWilliams transform, and 
            one has $B(x,y) = A(x,y)$. When such a code is of stabilizer type, 
            it corresponds to a classical self-dual code.
\end{itemize}

Some comments are in order.
The weights must be LU-invariant - the properties of the code
should not depend on the way one sets up the local coordinate system for each spin particle.
The last two properties are useful to obtain weights of hypothetical codes and to apply the 
machinery of linear programming bounds~\cite{Ashikmin1997}. That is, one sets up a system
of linear equalities and inequalities in the variables $A_0,\dots, A_n$ making use of a), b),
and the quantum MacWilliams identity.

For example, it is a longstanding open problem if a (pure) code with the parameters 
$[\![24, 0, 10 ]\!]_2$ exists. 
It is known that such code must have even weights only and using linear programming, 
one can fix the weight distribution to be
\begin{align}
    [A_{10},A_{12},A_{14},...A_{24}] =[18216,156492,1147608,3736557,6248088,4399164, \nonumber \\  1038312,32778]\,.\nonumber
\end{align}
Indeed this is also the weight distribution of hypothetical $[ 24,12,10 ]$ self-dual additive code over GF(4) (see OEIS \url{http://oeis.org/A030331}).

\begin{rprob}
 Either find a quantum code with parameters $[\![24, 0, 10 ]\!]_2$, or show that no such code can exist. 
\end{rprob}

We refer to the tables by M.~Grassl \cite{codetables} for more existence results.

\end{document}